\documentclass[a4paper,UKenglish,cleveref,autoref,thm-restate]{lipics-v2021}

\newboolean{isarxiv}
\setboolean{isarxiv}{true}

\setboolean{isarxiv}{true}

\title{The Cost of Skeletal Call-by-Need, Smoothly} 

\author{Beniamino {Accattoli}}{Inria, \& LIX, École Polytechnique, France}{beniamino.accattoli@inria.fr}{https://orcid.org/0000-0003-4944-9944}{}
\author{Francesco {Magliocca}}{Università degli Studi di Napoli Federico II,
Italy}{francesco.magliocca@unina.it}{https://orcid.org/0009-0009-2310-6565}{}
\author{Loïc {Peyrot}}{IMDEA Software Institute, Spain}{loic.peyrot@imdea.org}{https://orcid.org/0000-0002-1398-7460}{}
\author{Claudio {Sacerdoti Coen}}{Alma Mater Studiorum - Università di Bologna, Italy}{claudio.sacerdoticoen@unibo.it}{https://orcid.org/0000-0002-4360-6016}{}

\authorrunning{B. Accattoli, F. Magliocca, L. Peyrot, C. Sacerdoti Coen} 

\Copyright{Jane Open Access and Joan R. Public} 



\category{} 

\relatedversion{}
\ifthenelse{\boolean{isarxiv}}{
}{
  \relatedversiondetails[cite=accattoli2025costskeletalcallbyneedsmoothly]{Full Version}{https://arxiv.org/abs/2505.09242}
}
\supplement{OCaml implementation.}

\supplementdetails[subcategory={Source Code}, cite={}, swhid={}]{Software}{https://github.com/Franciman/fullylazymad}

\acknowledgements{}

\funding{Sacerdoti Coen was supported by the INdAM-GNCS project
  "Modelli Composizionali per l'Analisi di Sistemi Reversibili Distribuiti" and by the Cost Action CA2011 EuroProofNet.}

\nolinenumbers 

\EventEditors{Maribel Fern\'{a}ndez}
\EventNoEds{1}
\EventLongTitle{10th International Conference on Formal Structures for Computation and Deduction (FSCD 2025)}
\EventShortTitle{FSCD 2025}
\EventAcronym{FSCD}
\EventYear{2025}
\EventDate{July 14--20, 2025}
\EventLocation{Birmingham, UK}
\EventLogo{}
\SeriesVolume{337}
\ArticleNo{2}

\usepackage{ifthen}
 \usepackage{xspace}
 
\newboolean{talk}
\setboolean{talk}{false}
\newboolean{paper}
\setboolean{paper}{false}

\newboolean{LMCSstyle}
\setboolean{LMCSstyle}{false}
\newboolean{IEEEstyle}
\setboolean{IEEEstyle}{false}
\newboolean{lipicsstyle}
\setboolean{lipicsstyle}{false}
\newboolean{eptcsstyle}
\setboolean{eptcsstyle}{false}
\newboolean{entcsstyle}
\setboolean{entcsstyle}{false}
\newboolean{sigplanstyle}
\setboolean{sigplanstyle}{false}
\newboolean{easychairstyle}
\setboolean{easychairstyle}{false}
\newboolean{scrartclstyle}
\setboolean{scrartclstyle}{false}
\newboolean{lncsstyle}
\setboolean{lncsstyle}{false}
\newboolean{acmartstyle}
\setboolean{acmartstyle}{false}

\newboolean{needstheorems}
\setboolean{needstheorems}{false}
\newboolean{withimages}
\setboolean{withimages}{false}
\newboolean{withproofs}
\setboolean{withproofs}{true}
\newboolean{complexity}
\setboolean{complexity}{false}

\newboolean{techr}
\setboolean{techr}{false}

\newboolean{french}
\setboolean{french}{false}

\newboolean{colored}
\setboolean{colored}{false}

\setboolean{lipicsstyle}{true}
\setboolean{withproofs}{true}

\ifthenelse{\boolean{talk}}{}{\usepackage[utf8]{inputenc}}
\ifthenelse{\boolean{french}}{\usepackage[french]{babel}}{
  \ifthenelse{\boolean{lipicsstyle}}{}{\usepackage[english]{babel}}}
\usepackage{amsmath}
\usepackage{graphicx}
\usepackage{bussproofs}
\usepackage{cmll}
\usepackage{stmaryrd}
 \usepackage{multirow}
 \ifthenelse{\boolean{talk}}{\usepackage{color}}{
   \ifthenelse{\boolean{lipicsstyle}}{\usepackage{color}}{\usepackage[usenames]{xcolor}}}
 \ifthenelse{\boolean{IEEEstyle}}{
 \usepackage[bookmarks={false}]{hyperref}}{
 \usepackage{hyperref}}
 \usepackage{fancybox}
 \usepackage{hhline}
\usepackage{nameref}
\usepackage{complexity}
\usepackage{trimclip}

\newcommand{\CodeSymbol}[1]{\textcolor{orange}{#1}}
\lstdefinelanguage{ocaml}{
  literate={=}{{\CodeSymbol{=}}}1
           {->}{{\CodeSymbol{->}}}1
           {<-}{{\CodeSymbol{<-}}}1,
  keywords={type, mutable, and, list, of, bool, option, ref, let, then, if, function, string},
  keywordstyle=\color{blue}\bfseries,
  keywords=[2]{Type},
  keywordstyle=[2]\color{blue}\bfseries,
  identifierstyle=\color{black},
  sensitive=false,
  mathescape=true,
  comment=[l]{//},
  morecomment=[n]{(*}{*)},
  commentstyle=\color{purple}\ttfamily,
  stringstyle=\color{red}\ttfamily,
  morestring=[b]',
  morestring=[b]"
}

\usepackage{listings}
\ifthenelse{\boolean{IEEEstyle}}{
\setboolean{needstheorems}{true}}{}

\ifthenelse{\boolean{eptcsstyle}}{
\setboolean{needstheorems}{true}}{}

\ifthenelse{\boolean{scrartclstyle}}{
\setboolean{needstheorems}{true}}{}

\ifthenelse{\boolean{sigplanstyle}}{
\setboolean{needstheorems}{true}}{}

\ifthenelse{\boolean{easychairstyle}}{
\setboolean{needstheorems}{true}}{}

\ifthenelse{\boolean{lipicsstyle}}{
    }{}

\ifthenelse{\boolean{needstheorems}}{
\usepackage{amsthm}
    \newtheorem{theorem}{Theorem}[section]
    \newtheorem{lemma}[theorem]{Lemma}
    
    \newtheorem{corollary}[theorem]{Corollary}
    \newtheorem{proposition}[theorem]{Proposition}
    \newtheorem{definition}[theorem]{Definition}
    
    \ifthenelse{\boolean{scrartclstyle}}{
  
}{}

\theoremstyle{remark}
    }{}

\newcommand{\myproof}[1]{
\ifthenelse{\boolean{withproofs}}{#1}{}
}

\newcommand{\lterms}{\Lambda}

\newcommand{\la}[1]{\lambda #1.}

\newcommand{\tm}{t}
\newcommand{\tmtwo}{u}
\newcommand{\tmthree}{s}
\newcommand{\tmfour}{r}
\newcommand{\tmfive}{p}

\newcommand{\var}{x}
\newcommand{\vartwo}{y}
\newcommand{\varthree}{z}
\newcommand{\varfour}{w}

\newcommand{\varthreep}{z'}

\newcommand{\avar}{a}
\newcommand{\avartwo}{b}
\newcommand{\avarthree}{c}

\newcommand{\ap}[2]{#1[#2]}
\renewcommand{\ap}[2]{#1#2}

\newcommand{\dom}[1]{\symfont{dom}(#1)}

\newcommand{\rootRew}[1]{\mapsto_{#1}}
\newcommand{\Rew}[1]{\rightarrow_{#1}}

\renewcommand{\to}{\Rew{}}

\newcommand{\tostrat}{\Rew{\symfont{str}}}

\newcommand{\symfont}[1]{\mathtt{#1}}

\newcommand{\lssym}{\symfont{ls}}
\newcommand{\ls}{\lssym}

\newcommand{\db}{\symfont{dB}}

\newcommand{\msym}{\symfont{m}}

\newcommand{\lsndsym}{\symfont{nd}{\cdot}\symfont{ls}}

\newcommand{\sksym}{\symfont{sk}}

\newcommand{\val}{v}
\newcommand{\valtwo}{\val'}

\newcommand{\rtolsv}{\rootRew{\lsvsym}}
\newcommand{\rtolsnd}{\rootRew{\lsndsym}}

\newcommand{\ndsym}{\symfont{nd}}

\newcommand{\lsneedsym}{\symfont{nd}{\cdot}\ls}
\newcommand{\tolsneed}{\Rew{\lsneedsym}}

\newcommand{\ctxholep}[1]{\langle #1\rangle}
\newcommand{\ctxhole}{\ctxholep{\cdot}}

\newcommand{\ctxholefp}[1]{\langle #1\rangle}

\makeatletter
\newsavebox{\@brx}
\newcommand{\llangle}[1][]{\savebox{\@brx}{\(\m@th{#1\langle}\)}%
  \mathopen{\copy\@brx\kern-0.7\wd\@brx\usebox{\@brx}}}
\newcommand{\rrangle}[1][]{\savebox{\@brx}{\(\m@th{#1\rangle}\)}%
  \mathclose{\copy\@brx\kern-0.7\wd\@brx\usebox{\@brx}}}
\makeatother

\renewcommand{\ctxholefp}[1]{\llangle #1\rrangle}

\newcommand{\ctx}{C}

\newcommand{\ctxp}[1]{\ctx\ctxholep{#1}}

\newcommand{\evctx}{E}
\newcommand{\evctxtwo}{\evctx'}
\newcommand{\evctxthree}{\evctx''}

\newcommand{\evctxp}[1]{\evctx\ctxholep{#1}}
\newcommand{\evctxtwop}[1]{\evctxtwo\ctxholep{#1}}

\newcommand{\evctxfp}[1]{\evctx\ctxholefp{#1}}
\newcommand{\evctxtwofp}[1]{\evctxtwo\ctxholefp{#1}}
\newcommand{\evctxthreefp}[1]{\evctxthree\ctxholefp{#1}}

\newcommand{\nbvctxtwo}[1]{\nbvctxtwo{#1}}

\newcommand{\sctx}{S}
\newcommand{\sctxtwo}{{\sctx'}}
\newcommand{\sctxthree}{{\sctx''}}
\newcommand{\sctxp}[1]{\sctx\ctxholep{#1}}
\newcommand{\sctxptwo}[1]{\sctxtwo\ctxholep{#1}}

\newcommand{\sctxtwop}[1]{\sctxtwo\ctxholep{#1}}

\newcommand{\sctxv}[1]{\sctx_{#1}}

\newcommand{\sctxvp}[2]{\sctxv{#1}\ctxholep{#2}}

\newcommand{\defeq}{:=}

\newcommand{\grameq}{::=}

\newcommand{\esub}[2]{[#1{\shortleftarrow}#2]}
\newcommand{\skesub}[2]{\langle#1{\shortleftarrow}#2\rangle}
\newcommand{\isub}[2]{\{#1{\shortleftarrow}#2\}}

\newcommand{\letexp}[3]{\letexpsym\ #1=#2\ \symfont{in}\ #3}

\newcommand{\rtodb}{\rootRew{\db}}

\newcommand{\todbp}[1]{\Rew{\db#1}}
\newcommand{\todb}{\todbp{}}

\newcommand{\llbrace}{\{ \kern -0.27em \vert}
\newcommand{\rrbrace}{\vert \kern -0.27em \}}

\newcommand{\streq}{\equiv}

\renewcommand{\l}{\lambda}
\newcommand{\ie}{{\em i.e.}\xspace}
\newcommand{\eg}{{\em e.g.}\xspace}
\newcommand{\ih}{{\textit{i.h.}}\xspace}
\newcommand{\fv}[1]{\symfont{fv}(#1)}

\ifthenelse{\boolean{entcsstyle}}{}{
	}

\newcommand{\red}[1]{{\color{red} {#1}}}

\newcommand{\ignore}[1]{}

\newcommand{\myinput}[1]{\ifthenelse{\boolean{withimages}}{\input{#1}}{}}

\newcommand{\reflemma}[1]{Lemma~\ref{l:#1}}
\newcommand{\reflemmap}[2]{Lemma~\ref{l:#1}.\ref{p:#1-#2}}

\newcommand{\reflemmaeq}[1]{{L.\ref{l:#1}}}

\newcommand{\refthm}[1]{Theorem~\ref{thm:#1}}

\newcommand{\refprop}[1]{Prop.~\ref{prop:#1}}
\newcommand{\refsect}[1]{Sect.~\ref{sect:#1}}

\newcommand{\refapp}[1]{Appendix~\ref{sect:#1}}

\ifthenelse{\boolean{easychairstyle}}{}{
  \ifthenelse{\boolean{lncsstyle}}{}{	
	
  }
}
\newcommand{\reffig}[1]{Fig.~\ref{fig:#1}}

\newcommand{\refdef}[1]{Definition~\ref{def:#1}}

\newcommand{\refpoint}[1]{Point~\ref{p:#1}}

\newcommand{\levy}{{L{\'e}vy}\xspace}

\newcommand{\set}[1]{\{#1\}}

\newcommand{\nat}{\mathbb{N}}

\newcommand{\rtom}{\mapsto_\msym}

\newcommand{\cbn}{CbN\xspace}
\newcommand{\cbneed}{CbNeed\xspace}

\newcommand{\size}[1]{|#1|}
\newcommand{\sizep}[2]{|#1|_{#2}}

\newcommand{\sizedb}[1]{\sizep{#1}{\db}}
\newcommand{\sizeskdb}[1]{\sizep{#1}{\sksym{\cdot}\db}}

\newcommand{\env}{\genv}
\newcommand{\envtwo}{\env'}

\newcommand{\genv}{\symfont{E}}

\newcommand{\stack}{\symfont{S}}

\newcommand{\States}{\symfont{States}}

\ifthenelse{\boolean{acmartstyle}}{
	\renewcommand{\state}{\symfont{Q}}
}{
	\newcommand{\state}{\symfont{Q}}
}
\newcommand{\statetwo}{\state'}
\newcommand{\statethree}{\state''}
\newcommand{\statefour}{\state'''}

\newcommand{\decode}[1]{\underline{#1}}
\newcommand{\decodep}[2]{\decode{#1}\ctxholep{#2}}

\newcommand{\stempty}{\epsilon}
\newcommand{\cons}{::}

\newcommand{\deriv}{\symfont{e}}

\newcommand{\derivtwo}{\deriv'}
\newcommand{\derivthree}{\deriv''}
\newcommand{\derivfour}{\deriv'''}

\newcommand{\renamenop}[1]{#1^\alpha}

\newcommand{\mach}{\symfont{M}}

\newcommand{\tomachhole}[1]{\leadsto_{#1}}
\newcommand{\tomach}{\tomachhole{}}

\newcommand{\tomachb}{\tomachhole{\beta}}

\newcommand{\subsym}{\symfont{sub}}
\newcommand{\seasym}{\symfont{sea}}
\newcommand{\tomachine}{\tomachhole\mach}
\newcommand{\tomachbeta}{\tomachhole{\beta}}

\newcommand{\tomachsub}{\tomachhole{\subsym}}
\newcommand{\tomachsea}{\tomachhole{\seasym}}
\newcommand{\tomachseaone}{\tomachhole{\seasym_1}}
\newcommand{\tomachseatwo}{\tomachhole{\seasym_2}}
\newcommand{\tomachseathree}{\tomachhole{\seasym_3}}

\newcommand{\sizepr}[1]{\sizep{#1}{\prsym}}

\newcommand{\withproofs}[1]{\ifthenelse{\boolean{withproofs}}{#1}{}}

\newcommand{\withoutproofs}[1]{\ifthenelse{\boolean{withproofs}}{}{#1}}

\ifthenelse{\boolean{entcsstyle}}{}{
  \ifthenelse{\boolean{lncsstyle}}{}{	
    \newcommand{\case}[1]{{\bf Case #1.}}
  }
 }

\newcommand{\caselight}[1]{\textit{#1}}

\newcounter{numberone}

\newcommand{\bigo}{{\mathcal{O}}}

\newcommand{\emptylist}{\epsilon}

\newcommand{\glamst}[4]{(#1,#2,#3,#4)}

\newcommand{\tmp}{\tm'}
\newcommand{\tmtwop}{\tmtwo'}

\ifthenelse{\boolean{acmartstyle}}{
	\renewcommand{\terms}{\Lambda}
}{
	\newcommand{\terms}{\Lambda}
}

	\newcommand{\absterms}{\symfont{Terms}}
\newcommand{\esterms}{\Lambda_{\symfont{es}}}
\newcommand{\skterms}{\Lambda_{\symfont{sk}}}

\newcommand{\semsub}[2]{[\![#1{\shortleftarrow}#2]\!]}

\definecolor{dgreen}{rgb}{0.0, 0.5, 0.0}

\newcommand{\Id}{\symfont{I}}

\newcommand{\mlterms}{\lterms_\taken}

\newcommand{\taken}{\circ}

\newcommand{\lat}[1]{\l^\taken{#1}.}

\newcommand{\skeld}[1]{\mathsf{init}(#1)}
\newcommand{\frontiersym}{\Uparrow}
\newcommand{\skelu}[1]{\overline{#1}^{\frontiersym}}
\newcommand{\appt}{\taken}

\newcommand{\absym}{\symfont{ab}}

\newcommand{\rtomarkblone}{\rootRew{\prsym1}}
\newcommand{\rtomarkbltwo}{\rootRew{\prsym2}}
\newcommand{\rtomarkblthree}{\rootRew{\prsym3}}
\newcommand{\rtomarkwhone}{\rootRew{\absym1}}
\newcommand{\rtomarkwhtwo}{\rootRew{\absym2}}
\newcommand{\rtomarkwhthree}{\rootRew{\absym3}}

\newcommand{\tomarkblone}{\Rew{\prsym1}}
\newcommand{\tomarkbltwo}{\Rew{\prsym2}}
\newcommand{\tomarkblthree}{\Rew{\prsym3}}
\newcommand{\tomarkwhone}{\Rew{\absym1}}
\newcommand{\tomarkwhtwo}{\Rew{\absym2}}
\newcommand{\tomarkwhthree}{\Rew{\absym3}}

\newcommand{\tomark}{\Rew{\symfont{alg}}}

\newcommand{\mtm}{m}
\newcommand{\mtmtwo}{n}
\newcommand{\mtmthree}{o}

\newcommand{\mctx}{M}

\newcommand{\mctxp}[1]{\mctx\ctxholep{#1}}

\newcommand{\disc}[2]{\symfont{disc}^{#2}(#1)}
\renewcommand{\disc}[2]{\symfont{mskel}^{#2}(#1)}
\newcommand{\discabs}[1]{\symfont{mskel}(#1)}

\newcommand{\skelabs}[1]{\symfont{skel}(#1)}
\newcommand{\skeldec}[2]{\symfont{sk}\symfont{dec}^{#2}(#1)}
\newcommand{\skeldecabs}[1]{\symfont{sk}\symfont{dec}(#1)}

\newcommand{\skelvars}{\theta}
\renewcommand{\skelvars}{\symfont{V}}

\newcommand{\sizet}[1]{\sizep{#1}{\taken}}

\newcommand{\splitfun}[1]{\symfont{split}({#1})}

\newcommand{\chain}{\symfont{C}}

\newcommand{\prsym}{\symfont{pr}}
\newcommand{\tomachpr}{\tomachhole{\prsym}}

\newcommand{\tomachsk}{\tomachhole{\sksym}}
\newcommand{\tomachss}{\tomachhole{\sssym}}

\newcommand{\Sctxs}{\mathcal{S}}

\newcommand{\Sskctxs}{\mathcal{S}_{\symfont{sk}}}

\newcommand{\EVctxs}{\mathcal{E}}
\newcommand{\EVctxsp}[1]{\EVctxs\ctxholep{#1}}
\newcommand{\EVskctxs}{\mathcal{E}_{\symfont{sk}}}
\newcommand{\EVskctxsp}[1]{\EVskctxs\ctxholep{#1}}

\newcommand{\Mctxs}{\mathcal{M}}
\newcommand{\Mctxsp}[1]{\Mctxs\ctxholep{#1}}

\newcommand{\toneed}{\Rew{\symfont{need}}}

\newcommand{\rtosk}{\rootRew{\sksym}}
\newcommand{\tosk}{\Rew{\sksym}}
\newcommand{\sssym}{\symfont{ss}}
\newcommand{\rtoss}{\rootRew{\sssym}}
\newcommand{\toss}{\Rew{\sssym}}
\newcommand{\skneedsym}{\sksym{\cdot}\symfont{need}}
\newcommand{\toskneed}{\Rew{\skneedsym}}
\newcommand{\skdb}{\sksym{\cdot}\db}
\newcommand{\toskneeddb}{\Rew{\skdb}}
\newcommand{\toskdb}{\toskneeddb}

\makeatletter
\DeclareRobustCommand{\circbullet}{\mathbin{\vphantom{\circ}\text{\circbullet@}}}
\newcommand{\circbullet@}{%
  \check@mathfonts
  \m@th\ooalign{%
    \clipbox{0 0 0 {\dimexpr\height-\fontdimen22\textfont2}}{$\bullet$}\cr
    $\circ$\cr
  }%
}
\makeatother

\newcommand{\fourstate}[4]{(#1,#2,#3,#4)}

\newcommand{\compilrelsym}{\triangleleft}
\newcommand{\compilrel}[2]{#1\compilrelsym#2}

\newcommand{\run}{r}
\newcommand{\runtwo}{\run'}
\newcommand{\runthree}{\run''}

\newcommand{\flesh}[1]{\symfont{flesh}(#1)}
\newcommand{\nfo}[1]{\symfont{nf}_{\seasym}(#1)}

\newcommand{\tosmad}{\leadsto_{\symfont{SMAD}}}

\renewcommand{\letexp}{\mathsf{let}}
\newcommand{\letin}[3]{{\sf let}\ #1=#2\ {\sf in}\ #3}

\usepackage{tikz}
\usetikzlibrary{calc}
\usetikzlibrary{matrix,arrows}
\usetikzlibrary{decorations.shapes}
\usetikzlibrary{decorations.text}
\usetikzlibrary{decorations.pathmorphing}
\usetikzlibrary{decorations.markings}
\usetikzlibrary{positioning}

\tikzset{
node distance=1.3cm, auto,
every node/.style={font=\scriptsize },
ocenter/.style={baseline={([yshift=-.5ex, xshift=-.5ex]current bounding box)}},  
labelBeginAbove/.style={postaction={decorate,decoration={markings,mark=at position 0 with {\node[inner sep= 0.6pt, above=1pt]{\tiny #1};}} } },
labelBeginBelow/.style={postaction={decorate,decoration={markings,mark=at position 0 with {\node[inner sep= 0.6pt, below=1pt]{\tiny #1};}}}},
labelEndAbove/.style={postaction={decorate,decoration={markings,mark=at position 1 with {\node[inner sep= 0.6pt, above=1pt]{\tiny #1};}}}},
labelEndBelow/.style={postaction={decorate,decoration={markings,mark=at position 1 with {\node[inner sep= 0.6pt, below=1pt]{\tiny #1};}}}},
labelEndRight/.style={postaction={decorate,decoration={markings,mark=at position 1 with {\node[inner sep= 0.6pt, right=1pt]{\tiny #1};}}}},
labelEndLeft/.style={postaction={decorate,decoration={markings,mark=at position 1 with {\node[inner sep= 0.6pt, left=1pt]{\tiny #1};}}}}
}

\usepackage{appendix}
\usepackage{marginnote}

\newboolean{appendix}
\setboolean{appendix}{false}
\newboolean{withlinkedproofs}
\ifthenelse{\boolean{isarxiv}}{\setboolean{withlinkedproofs}{true}}{\setboolean{withlinkedproofs}{false}}

\capturecounter{lemma}
\capturecounter{proposition}
\capturecounter{theorem}

\newcommand{\NoteProof}[1]{
	\ifthenelse{\boolean{withlinkedproofs}}{\ifthenelse{\boolean{appendix}}{
	\marginnote{Originally at p. \pageref{#1}}
	}{
	\marginnote{{Proof at p.\,{\pageref{app:#1}}}}
	}
	}{}
}

\newcommand{\applabel}[1]{$\phantomsection\label{app:#1}$}

\usepackage{apptools}
\AtAppendix{\counterwithin{theorem}{section}}
\AtAppendix{\counterwithin{proposition}{section}}
\AtAppendix{\counterwithin{definition}{section}}
\AtAppendix{\counterwithin{lemma}{section}}

\renewcommand{\cons}{{:}}

\renewcommand{\rtolsv}{\rtolsnd}

\begin{document}

\maketitle

\begin{abstract}
Skeletal call-by-need is an optimization of call-by-need evaluation also known as “fully lazy sharing”: when the duplication of a value has to take place, it is first split into “skeleton”, which is then duplicated, and “flesh” which is instead kept shared.

Here, we provide two cost analyses of skeletal call-by-need. Firstly, we provide a family of terms showing that skeletal call-by-need can be asymptotically exponentially faster than call-by-need in both time \emph{and} space; it is the first such evidence, to our knowledge.

Secondly, we prove that skeletal call-by-need can be implemented efficiently, that is, with bi-linear overhead. This result is obtained by providing a new smooth presentation of ideas by Shivers and Wand for the reconstruction of skeletons, which is then smoothly plugged into the study of an abstract machine following the distillation technique by Accattoli et al.

\end{abstract}

\section{Introduction}
\label{sec:introduction}
Call-by-need evaluation of the $\l$-calculus was introduced by Wadsworth in 1971 as an optimization of the untyped $\l$-calculus \cite{Wad:SemPra:71}. It combines the  advantages of both call-by-name and call-by-value: it  avoids diverging on unused arguments, as in call-by-name, and when arguments are evaluated they do so only once, as in call-by-value. Such a combination is trickier to specify than call-by-name or call-by-value. In particular, Wadsworth's original presentation via graph reduction was not easy to manage. 

In the 1990s, call-by-need was adopted by the then new Haskell language, and more manageable term-based calculi were developed by Launchbury \cite{DBLP:conf/popl/Launchbury93} and Ariola et al. \cite{DBLP:conf/popl/AriolaFMOW95}, and extended with non-determinism by Kutzner and Schmidt-Schau{\ss} \cite{DBLP:conf/icfp/KutznerS98}. 
Additionally, Sestoft studied abstract machines for call-by-need \cite{DBLP:journals/jfp/Sestoft97}. In 2014, Accattoli et al. introduced the \emph{\cbneed linear substitution calculus} (\cbneed LSC) \cite{DBLP:conf/icfp/AccattoliBM14}, a simplification of  Ariola et al.'s calculus resting on the \emph{at a distance} approach to explicit substitutions by Accattoli and Kesner \cite{DBLP:conf/csl/AccattoliK10}. It became a reference presentation, used in many recent studies \cite{DBLP:conf/fossacs/Kesner16,DBLP:journals/pacmpl/BalabonskiBBK17,DBLP:conf/ppdp/AccattoliB17,DBLP:conf/fossacs/KesnerRV18,DBLP:conf/ppdp/BarenbaumBM18,DBLP:conf/esop/AccattoliGL19,DBLP:journals/lmcs/KesnerPV24,DBLP:conf/fscd/BalabonskiLM21,DBLP:conf/csl/AccattoliL22,DBLP:conf/fscd/AccattoliL24}. As far as this paper is concerned, the study in \cite{DBLP:conf/icfp/AccattoliBM14} is particularly relevant: beyond introducing the \cbneed LSC that we shall use, it also relates it with abstract machines for \cbneed, building a neat bisimulation between the two, deemed \emph{distillation}.

\subparagraph{Skeletal Call-by-Need.} \cbneed was introduced together with a further optimization usually called \emph{fully lazy sharing} and here rather referred to as \emph{Skeletal \cbneed} (because \emph{fully lazy sharing} is not really descriptive). The idea is to refine the duplication of a value $\val$ so as to duplicate only the \emph{skeleton} of $\val$ while keeping its \emph{flesh} shared, thus ending up sharing more than in ordinary call-by-need. As an example, if $\val = \la\var\la\vartwo\varthree\varthree\var(\vartwo\varthree)$ then the skeleton of $\val$ is $\skelabs\val \defeq \la\var\la\vartwo\varfour\var(\vartwo\varthree)$ and its flesh is the context $\flesh\val \defeq \letin\varfour{\varthree\varthree}\ctxhole$. That is, the flesh collects and removes the maximal (non-variable) sub-terms of $\val$ that are \emph{free}, \ie such that none of their free variables is captured in $\val$.

Skeletal \cbneed is less studied than \cbneed because it is even trickier to define, as it requires to compute the skeletal decomposition of a value (that is, the split into skeleton and flesh). Early works by Turner \cite{DBLP:journals/spe/Turner79}, Hughes \cite{hughes:thesis}, and Peyton Jones \cite{DBLP:books/ph/Jones87} focus on implementative aspects. In the last ten years or so, various works provided operational insights. Balabonski developed an impressive theoretical analysis of Skeletal \cbneed, showing the equivalence of various presentations \cite{DBLP:conf/popl/Balabonski12} and studying the relationship with \levy's optimality \cite{DBLP:conf/icfp/Balabonski13}. Soon after that, Gundersen et al. showed that the \emph{atomic $\l$-calculus}, a $\l$-calculus with sharing issued from deep inference technology, naturally specifies the duplication of the skeleton \cite{DBLP:conf/lics/GundersenHP13}. Next, Kesner et al. \cite{DBLP:journals/lmcs/KesnerPV24} merged the insights of the \cbneed LSC and the atomic $\l$-calculus obtaining a Skeletal \cbneed LSC, that is, a presentation at a distance of Skeletal \cbneed.

\subparagraph{This Paper.} The present work extends the operational research line of Balabonski, Gundersen et al., and Kesner et al. with cost analyses of Kesner et al.'s Skeletal \cbneed LSC, aiming at closing the gap with some of the earlier implementative studies at the same time. Generally speaking, we aim at adding Skeletal \cbneed to the list of concepts covered by the theory of cost-based analyses of abstract machines and sharing mechanisms by Accattoli and co-authors \cite{DBLP:conf/icfp/AccattoliBM14,DBLP:journals/corr/AccattoliL16,DBLP:journals/iandc/AccattoliC17,DBLP:conf/ppdp/AccattoliB17,DBLP:conf/ppdp/AccattoliCGC19,DBLP:conf/ppdp/CondoluciAC19,DBLP:journals/scp/AccattoliG19,DBLP:journals/pacmpl/AccattoliLV21,DBLP:conf/lics/AccattoliCC21,DBLP:journals/lmcs/AccattoliLV24}. We provide two analyses.

\subparagraph{First Analysis: Exponential Time and Space Speed-Ups.} The first analysis shows that in some cases Skeletal \cbneed is considerably more efficient than \cbneed. For that, we exhibit a family of $\l$-terms $\{\tm_n\}_{n\in\nat}$ such that $\tm_n$ evaluates in $\Omega(2^n)$ time and space in \cbneed, while requiring only $\bigo(n)$ space and time in Skeletal \cbneed. 
The literature only shows examples of single $\l$-terms where Skeletal \cbneed takes a few steps less than \cbneed, but never proves any asymptotic speed-up nor deals with space. Lastly, proving that our family evaluates within the given bounds requires a fine, non-trivial analysis of evaluation sequences.

\subparagraph{Second Analysis: Bi-Linear Overhead.} The second analysis  shows that Skeletal \cbneed can be implemented efficiently, with an overhead that is \emph{bi-linear}, \ie linear in both the number of $\beta$-steps and the size of the initial term. This time, the same is true for \cbneed. The insight is that computing  skeletons  does require additional implementative efforts, but it comes at no asymptotic price.
This work was indeed triggered by showing that the operational reconstruction of the skeleton at work in Gundersen et al. \cite{DBLP:conf/lics/GundersenHP13} and Kesner et al. \cite{DBLP:journals/lmcs/KesnerPV24} can be implemented in linear time. Their specifications rest on side conditions about variables in sub-terms that lead to non-linear overhead, if implemented literally. In fact, the linear time reconstruction of skeletons by Shivers and Wand \cite{DBLP:journals/fuin/ShiversW10}, from 2010, avoids the side conditions  and pre-dates the recent works. The insightful study in \cite{DBLP:journals/fuin/ShiversW10}, however, is technical and not well-known. Here we simplify it, recasting it in the context of abstract machines.

\subparagraph{Efficient Skeletal Decompositions, or Shivers and Wand Reloaded.} Shivers and Wand present their ideas using graph-reduction for $\l$-terms. They reconstruct the skeletal decomposition via a visit of the value to skeletonize seen as a graph, based on two key points:
\begin{enumerate}
\item \emph{Bi-directional edges}: all the edges of their graphs can be traversed in \emph{both} directions. When term graphs are seen as mathematical objects, edges are often directed but the theoretical study can also look at edges in reverse. At the implementative level, however, things are different. For a parsimonious use of space, most graph-based implementations (\eg the proof nets inspired ones in \cite{DBLP:conf/ppdp/AccattoliB17,DBLP:conf/lics/AccattoliCC21,DBLP:conf/fscd/AccattoliC24}) indeed restrict some edges to be traversed in one direction only (namely, parent to children), thus sparing some pointers. Therefore, bi-directional edges are an unusual approach, usually considered redundant.

\item \emph{Abstractions, occurrences, and upward reconstruction}: in most graphical representations, there are edges between abstractions and the occurrences of their bound variables, usually directed only from the occurrences to the abstraction. Shivers and Wand's insight is that, with bi-directional edges, one can move from the abstraction back to the occurrences of the variable and then reconstruct the skeleton by visiting upward from the occurrences.
\end{enumerate}
Shivers and Wand's study is graph-based and low-level, interleaved with code snippets. 
Here, we provide a new neat presentation of their algorithm as a simple high-level rewriting system over terms, circumventing graphs and low-level details. The reformulation of their study is the main technical innovation of this paper. While the key concepts are due to Shivers and Wand, we believe that our new presentation is a notable contribution in that it allows their concepts to blossom, simplifying their study and making them more widely accessible. 

\subparagraph{Distillation.} We then smoothly lift the distillation-based study of \cbneed by Accattoli et al. \cite{DBLP:conf/icfp/AccattoliBM14} to Skeletal \cbneed, providing an abstract machines using the reloaded  algorithm for skeletal decompositions as a black-box. The new abstract machine is then shown to be implementable within a bi-linear overhead, completing the second analysis.

The overall insight is that the apparently linear space inefficiency of adding pointers for bi-directional traversals of terms/graphs enables optimizations---namely, skeletal duplications---that can bring, in some cases, exponential speed-ups for both time and space.

%
\subparagraph{Implementation.} We provide an OCaml implementation of the abstract machine, and in particular of the reloaded skeleton reconstruction algorithm it rests upon. The implementation is there to  validate the complexity analyses, as well as to integrate the low-level aspect of Shivers and Wand's study. 
The implementation is discussed in \refapp{implementation}.

\ifthenelse{\boolean{isarxiv}}{}{\subparagraph{Proofs.} A long version with proofs in the Appendix is on arXiv~\cite{accattoli2025costskeletalcallbyneedsmoothly}.}

\section{Background: Call-by-Need}
\label{sect:cbneed}

In this section, we recall \cbneed, presented as a strategy of the \cbneed linear substitution calculus (shortened to LSC). For the sake of simplicity, we do not distinguish here between the strategy and the calculus, and use \emph{\cbneed LSC} to refer to both.

 Accattoli and Kesner's LSC \cite{DBLP:conf/rta/Accattoli12,DBLP:conf/popl/AccattoliBKL14} is a micro-step $\l$-calculus with explicit substitutions. \emph{Micro-step} means that substitutions act on one variable occurrence at a time, rather than \emph{small-step}, that is, on all occurrences at the same time. The \cbneed LSC was introduced by  Accattoli et al. \cite{DBLP:conf/icfp/AccattoliBM14} as a simplified variant of Ariola et al.'s presentation of \cbneed \cite{DBLP:conf/popl/AriolaFMOW95}. The simplification is the use of rewriting rules \emph{at a distance}---that is, adopting contexts in the root rules---that allow one to get rid of the commuting rewriting rules of Ariola et al.

\begin{figure}[t!]
\centering
\fbox{
\begin{tabular}{c}
$\begin{array}{r@{\hspace{.3cm}} rll@{\hspace{.25cm}}|@{\hspace{.25cm}}r@{\hspace{.3cm}} rll}
\textsc{Terms} &\esterms\ni\tm,\tmtwo & \grameq & \var \mid \val \mid \tm\tmtwo \mid \tm\esub\var\tmtwo
&
\textsc{Values} &\val,\valtwo & \grameq & \la\var\tm \text{ with }\tm\in\terms
\end{array}$
\\[4pt]
$\begin{array}{r@{\hspace{.3cm}} rll}
\textsc{Substitutions ctxs} &\Sctxs\ni\sctx,\sctxtwo & \grameq & \ctxhole \mid  \sctx\esub\var\tmtwo
\\
\textsc{\cbneed Evaluation ctxs} &\EVctxs\ni\evctx,\evctxtwo & \grameq & \ctxhole \mid \evctx\tmtwo  \mid \evctx\esub\var\tmtwo \mid \evctxfp\var\esub\var\evctxtwo
\end{array}$
\\[5pt]\hline
\textsc{Root rules}
\\[2pt]
$\begin{array}{r@{\hspace{.5cm}} rlll}
\textsc{Distant $\beta$} &\sctxp{\la\var\tm}\tmtwo & \rtodb & \sctxp{\tm\esub\var\tmtwo}
\\[2pt]
\textsc{\cbneed linear subst.} &\evctxfp\var\esub\var{\sctxp\val} & \rtolsnd & \sctxp{\evctxfp\val\esub\var\val}
\end{array}$
\\[10pt]\hline
\begin{tabular}{c|c}
$\begin{array}{rll@{\hspace{.25cm}}|@{\hspace{.25cm}}rll}
\multicolumn{6}{c}{\textsc{Rewriting rules}}
\\
\todb & \grameq & \EVctxsp\rtodb
&
\tolsneed & \grameq & \EVctxsp\rtolsnd
\end{array}$

&
$\begin{array}{rlll}
\multicolumn{3}{c}{\textsc{Notation}}
\\
\toneed & \grameq & \todb \cup \tolsneed
\end{array}$
\end{tabular}
\end{tabular}
}
\caption{The Call-by-Need Linear Substitution Calculus (\cbneed LSC).}
\label{fig:cbneed-strategy}
\end{figure}
\subparagraph{Terms and Contexts.} The set of terms of the  $\l$-calculus is noted $\terms$.  The \cbneed LSC is defined in \reffig{cbneed-strategy}.  The set of terms of the \cbneed LSC is noted $\esterms$. They add  \emph{explicit substitutions}  $\tm\esub{\var}{\tmtwo}$ (shortened to ESs) to $\l$-terms, that is 
a more compact notation for $\letin\var\tmtwo\tm$, but where the order of evaluation between $\tm$ and $\tmtwo$ is a priori not fixed; evaluation contexts shall fix it.
The set $\fv{\tm}$ of \emph{free} variables is defined as expected, in particular,
$\fv{\tm\esub{\var}{\tmtwo}} \defeq
(\fv{\tm} \setminus \set{\var}) \cup \fv{\tmtwo}$. Both $\la{\var}\tm$ and $\tm\esub{\var}{\tmtwo}$ bind $\var$ in $\tm$, and terms are identified up to $\alpha$-renaming. 
A term $\tm$ is \emph{closed} if $\fv{\tm} = \emptyset$, \emph{open} otherwise.
Meta-level capture-avoiding substitution is noted $\tm\isub\var\tmtwo$. The size $\size\tm$ of $\tm$ is the number of its construct.

Note that values are only abstraction; this choice is standard in the literature on \cbneed. Note that, moreover, their bodies are $\l$-terms \emph{without} ESs. This is in order to avoid dealing with calculating the skeleton of ESs in the next section. It is a harmless choice because evaluation (that creates ESs) never enters inside abstraction.

Contexts are terms with exactly one occurrence of the \emph{hole} $\ctxhole$, an additional constant, standing for a removed sub-term. We shall heavily use two notions of contexts: \emph{substitution contexts} $\sctx$, that are lists of ESs, and evaluation contexts $\evctx$.

The main operation about contexts is \emph{plugging} $\evctxp{\tm}$ where the hole $\ctxhole$ in context 
$\evctx$ is replaced by $\tm$. Plugging, as usual with contexts, can
capture variables; for instance $((\ctxhole \tm)\esub\var\tmtwo)\ctxholep\var = (\var\tm)\esub\var\tmtwo$. 
We write $\evctxfp{\tm}$ when we want to stress that the context $\evctx$ does not capture the free variables of $\tm$.

\subparagraph{Rewriting Rules.}
The reduction rules of the \cbneed LSC are slightly unusual as they use \emph{contexts} both to allow one to reduce redexes located in sub-terms (via evaluation contexts $\evctx$), which is standard, \emph{and} to define the redexes themselves (via both substitution $\sctx$ and evaluation contexts $\evctx$), which is less standard. This approach is 
called \emph{at a distance} and related to cut-elimination on proof nets; see Accattoli \cite{DBLP:journals/tcs/Accattoli15,DBLP:conf/ictac/Accattoli18} for the link with proof nets. The notion of evaluation context $\evctx$ used for \cbneed is exactly the one by Ariola et al. \cite{DBLP:conf/popl/AriolaFMOW95}, that extend call-by-name evaluation contexts with the production $\evctxfp\var\esub\var\evctxtwo$, which---by enterings ESs---is what enables the memoization aspect of \cbneed.

The \emph{distant beta rule} $\rtodb$ is essentially the $\beta$-rule, except that the argument goes into a new ES, rather then being immediately substituted, and that there can be a substitution context $\sctx$ in between the abstraction and the argument. Example: $(\la\var\vartwo)\esub\vartwo\tm\tmtwo \rtom \vartwo\esub\var\tmtwo\esub\vartwo\tm$. One with on-the-fly $\alpha$-renaming is $(\la\var\vartwo)\esub\vartwo\tm\vartwo \rtom \varthree\esub\var\vartwo\esub\varthree\tm$.

The \emph{linear substitution rule by need} $\rtolsnd$ replaces a single variable occurrence by a copy of the value in the ES, commuting the substitution context around the value (if any) out of the ES. Example: $(\var\var)\esub\var{\Id\esub\vartwo\tm} \rtolsv (\Id\var)\esub\var\Id\esub\vartwo\tm$.

The use of substitution contexts $\sctx$ in the root rules is what allows one to avoid the two commuting rules of Ariola et al. \cite{DBLP:conf/popl/AriolaFMOW95}, namely (rephrasing their $\letexp$-expressions as ESs)
$\tm\esub\var\tmthree\tmtwo \rootRew{let\text{-}C} (\tm\tmtwo)\esub\var\tmthree$
 and $\tm\esub\vartwo{\tmtwo\esub\var\tmthree} \rootRew{let\text{-}A} \tm\esub\vartwo\tmtwo\esub\var\tmthree$.

The two root rules $\rtodb$ and $\rtolsv$ are then closed by evaluation contexts $\evctx$. In \reffig{cbneed-strategy}, we use the compact notation $\todb \defeq \EVctxsp\rtodb$ to denote that $\todb$ is defined as $\evctxp\tm \todb \evctxp\tmtwo$ if $\tm\rtodb\tmtwo$, and similarly for the other rule.


\section{Skeletal Call-by-Need}
\label{sect:skeletal}
In this section, we present our version of Skeletal \cbneed, which is a very minor variant of the one by Kesner et al. \cite{DBLP:journals/lmcs/KesnerPV24} (the difference is discussed at the end of the section), itself defined by tweaking the \cbneed LSC via the notion of skeleton.

The basic idea of skeletal \cbneed is that every value $\val\defeq \la\var\tm$ can be split in two: the skeleton $\la\var\tmtwo$, which is a sort of sub-term of $\tm$, and the flesh $\sctx$, which is a substitution context collecting the maximal sub-terms of $\val$ that do not depend on $\var$. The flesh is extracted from $\val$ before duplicating it, as to avoid duplicating the code (and possibly redexes) of the flesh, thus increasing sharing. Some definitions are in order.

\subparagraph{Free Sub-Terms.} Extractable sub-terms are called \emph{free sub-terms}, defined next. We actually need a parametrized notion of free sub-term, which is relative to a set of variables $\skelvars$ that are supposed not to occur in the sub-term. 

\begin{definition}[Maximally free sub-terms]
Let $\tm \in \lterms$ and $\skelvars$ be a set of variables. A sub-term $\tmtwo$ of $\tm$ is (the term part of) a decomposition $\tm=\ctxp\tmtwo$ for some context $\ctx$; additionally, $\tmtwo$ is \emph{$\skelvars$-free in $\tm$} if no variable in $\fv{\tmtwo}$ is captured by $\ctx$ and if $\fv\tmtwo\cap \skelvars = \emptyset$, and 
\emph{maximally $\skelvars$-free} if it is not a strict sub-term of any other $\skelvars$-free sub-term in $\tm$.
\end{definition}

\subparagraph{Skeleton.} The skeletal decomposition of a term $\tm$ is the splitting of $\tm$ into its maximal free sub-terms on one side, collected as a substitution context $\sctx$ called \emph{the flesh}, and what is left of $\tm$ after the removal, that is, its \emph{skeleton}. As for free sub-terms, we rather define parametrized notions of skeletal decomposition and skeleton. Then, we specialize the definitions for values.

\begin{definition}[(Relative) Skeleton]
Let $\skelvars$ be a set of variables. The skeletal decomposition of $\tm$ relative to a set of variables $\skelvars$ is the pair $(\tmtwo,\sctx)$ of an ordinary $\l$-term $\tmtwo$ and a substitution context $\sctx$ defined as $\skeldec{\tm}{\skelvars}  \defeq  (\var, \ctxhole\esub\var\tm)$ with $\var$ fresh, if $\fv \tm \cap \skelvars = \emptyset$ and $\tm$ is not a variable, otherwise:
  \begin{center}$\begin{array}{r@{\hspace{.2cm}}l@{\hspace{.2cm}}lllllll}
\skeldec{\var}{\skelvars} & \defeq & (\var,\ctxhole) 
\\
\skeldec{\la\var\tmtwo}{\skelvars} & \defeq & (\la{\var}\tmthree,\sctx)
&\text{where }\skeldec{\tmtwo}{\skelvars \cup \set\var} = (\tmthree,\sctx)
\\
\skeldec{\tmtwo\ap\tmthree}{\skelvars} & \defeq & (\tmfour\tmfive, \sctxtwop\sctx)
&\text{where }\skeldec{\tmtwo}{\skelvars \cup \set\var} = (\tmfour,\sctx) \text{ and }\skeldec{\tmthree}{\skelvars \cup \set\var} = (\tmfive,\sctxtwo)
\\
  \end{array}$\end{center}
The skeleton, the flesh, and the skeletal decomposition of $\la\var\tm$ are defined respectively as $\skelabs{\la\var\tm} \defeq \la\var\tmtwo$, $\flesh{\la\var\tm}\defeq\sctx$, and $\skeldecabs{\la\var\tm} \defeq (\skelabs{\la\var\tm},\flesh{\la\var\tm})$, where $\skeldec\tm{\set\var} = (\tmtwo,\sctx)$.
\end{definition}

\begin{figure}[t!]
\centering
\fbox{
\begin{tabular}{c}
$\begin{array}{r@{\hspace{.25cm}} rll}
\textsc{Terms} &\skterms\ni\tm,\tmtwo & \grameq & \var \mid \val \mid \tm\tmtwo \mid \tm\esub\var\tmtwo \mid \tm\skesub\var\val \text{ with }\val=\skelabs\val
\\
\textsc{Values} &\val,\valtwo & \grameq & \la\var\tm \text{ with }\tm\in\terms

\\[4pt]
\textsc{Substitutions ctxs} &\Sskctxs\ni\sctx,\sctxtwo & \grameq & \ctxhole \mid  \sctx\esub\var\tmtwo \mid  \sctx\skesub\var\val 
\\
\textsc{Skeletal eval. ctxs} &\EVskctxs\ni\evctx,\evctxtwo & \grameq & \ctxhole \mid \evctx\tmtwo  \mid \evctx\esub\var\tmtwo \mid  \evctx\skesub\var\val \mid \evctxfp\var\esub\var\evctxtwo
\end{array}$
\\[5pt]\hline
\textsc{Root rules}
\\[2pt]
$\begin{array}{r@{\hspace{.25cm}} rlll}
\textsc{Distant $\beta$} &\sctxp{\la\var\tm}\tmtwo & \rootRew{\db} & \sctxp{\tm\esub\var\tmtwo}
\\[2pt]
\textsc{Skeletonization} &\evctxfp\var\esub\var{\sctxp\val} & \rtosk & \sctxp{\sctxtwop{\evctxfp\var\skesub\var\valtwo}} 
\mbox{ with }\skeldecabs{\val}=(\valtwo,\sctxtwo)
\\[2pt]
\textsc{Skeletal subst.} &\evctxfp\var\skesub\var\val & \rtoss & \evctxfp\val\skesub\var\val
\end{array}$
\\[20pt]\hline
\begin{tabular}{c|c}
$\begin{array}{rll@{\hspace{.25cm}}|@{\hspace{.25cm}}rll}
\multicolumn{6}{c}{\textsc{Rewriting rules}}
\\
\toskneeddb & \grameq & \EVskctxsp\rtodb
&
\tosk & \grameq & \EVskctxsp\rtosk
\\
\toss & \grameq & \EVskctxsp\rtoss
\end{array}$

&
$\begin{array}{rlll}
\multicolumn{3}{c}{\textsc{Notations}}
\\
\toskneed & \grameq & \toskneeddb \cup \tosk \cup \toss
\\[4pt]
\multicolumn{3}{l}{\semsub\var\tm\text{ stands for either }\esub\var\tm\text{ or }\skesub\var\val}
\
\end{array}$
\end{tabular}
\end{tabular}
}
\caption{The skeletal \cbneed strategy $\toskneed$.}
\label{fig:skeletal-strategy}
\end{figure}
\subparagraph{Skeletal \cbneed.} We adopt the presentation of skeletal \cbneed by Kesner et al. \cite{DBLP:journals/lmcs/KesnerPV24} (refered to as \emph{fully lazy \cbneed} by them), that adds a new \emph{skeletal (explicit) substitution} $\tm\skesub\var\val$ containing \emph{skeletal values}, that is, values that are equal to their skeleton. The language of terms and the strategy is defined in \reffig{skeletal-strategy}. The new set of terms is noted $\skterms$.

\textbf{\textit{Notation}}: when we need to treat explicit and skeletal substitution together, we use the notation $\semsub\var\tm$, which stands for either $\esub\var\tm$ or (with a slight abuse) $\skesub\var\val$.

The \cbneed substitution rule $\tolsneed$ for a redex $\evctxfp\var\esub\var{\sctxp\val}$ is refined via two rules. The first \emph{skeletonization} rule $\tosk$ does three further mini tasks at once: firstly, it decomposes $\val$ in its skeleton $\valtwo$ and its flesh $\sctxtwo$ (thus obtaining $\evctxfp\var\esub\var{\sctxp{\sctxtwop\val}}$); secondly, it re-organizes the term commuting $\sctx$ and $\sctxtwo$ out of the ES (obtaining  $\sctxp{\sctxtwop{\evctxfp\var\esub\var\valtwo}}$); thirdly, it turns the ES into a skeletal substitutions, finally producing $\sctxp{\sctxtwop{\evctxfp\var\skesub\var\valtwo}}$.

The second \emph{skeletal substitution} rule $\toss$ simply replaces $\var$ with $\valtwo$.

Note that evaluation contexts $\evctx$ are adapted to include skeletal substitutions, but that they do not enter inside them, since their content is always a value. The distant $\beta$ rule for Skeletal \cbneed is noted $\toskneeddb$ to distinguish it from the one for \cbneed. The Skeletal \cbneed reduction $\toskneed$ is the union of $\toskneeddb$, $\tosk$, and $\toss$.

The next proposition is an easy adaptation of the one for \cbneed in Accattoli et al. \cite{DBLP:conf/icfp/AccattoliBM14}.

\begin{toappendix}
\begin{proposition}
\label{skneed-determinism}\NoteProof{skneed-determinism}
The reduction $\toskneed$ is deterministic.
\end{proposition}
\end{toappendix}

\subparagraph{Structural Equivalence.} To study the relationship with abstract machines, we shall need the concept of structural equivalence, defined in \reffig{structural-equivalence}. The concept is standard for $\l$-calculi with ESs at a distance, and it also standard to use it in relationship with abstract machines, as systematically done by Accattoli et al. \cite{DBLP:conf/icfp/AccattoliBM14}. Intuitively, structural equivalence allows one to move explicit/skeletal substitutions around because the move does not change the behaviour of the term. This fact is captured by the following strong bisimulation property, proved by an immediate adaptation of the similar property in \cite{DBLP:conf/icfp/AccattoliBM14}.
\begin{figure}[t!]
\centering
\fbox{
\begin{tabular}{c}
\textsc{Root structural equalities}
\\[4pt]
$\begin{array}{r@{\hspace{.25cm}} rlll}
\textsc{Expl. subst.} & \evctxp{\tm\esub\var\tmthree} & \streq_1 & \evctxp\tm\esub\var\tmthree & \text{if }\dom\evctx\cap\fv{\tm\esub\var\tmthree}=\emptyset\text{ and }\var\notin\fv\evctx
\\[3pt]
\textsc{Skel. subst.} & \evctxp{\tm\skesub\var\val} & \streq_2 & \evctxp\tm\skesub\var\val & \text{if }\dom\evctx\cap\fv{\tm\skesub\var\val}=\emptyset\text{ and }\var\notin\fv\evctx
\end{array}$
\\[10pt]\hline
\begin{tabular}{ccccccccc}
\multicolumn{6}{c}{\textsc{Closure rules}}
\\
\AxiomC{$\tm \streq_1 \tmtwo$}
\UnaryInfC{$\tm\streq\tmtwo$}
\DisplayProof
&
\AxiomC{$\tm \streq_2 \tmtwo$}
\UnaryInfC{$\tm\streq\tmtwo$}
\DisplayProof
&
\AxiomC{}
\RightLabel{\scriptsize ref}
\UnaryInfC{$\tm\streq\tm$}
\DisplayProof
&
\AxiomC{$\tm \streq \tmtwo$}
\RightLabel{\scriptsize sym}
\UnaryInfC{$\tmtwo\streq\tm$}
\DisplayProof
&
\AxiomC{$\tm \streq \tmtwo$}
\AxiomC{$\tmtwo\streq \tmthree$}
\RightLabel{\scriptsize tr}
\BinaryInfC{$\tm\streq\tmthree$}
\DisplayProof
&
\AxiomC{$\tm \streq \tmtwo$}
\RightLabel{\scriptsize ev}
\UnaryInfC{$\evctxp\tm\streq\evctxp\tmtwo$}
\DisplayProof

\end{tabular}
\end{tabular}
}
\caption{Structural equivalence $\streq$.}
\label{fig:structural-equivalence}
\end{figure}

\begin{proposition}[$\streq$ is a strong bisimulation]
\label{prop:streq-is-a-strong-bisim}
Let $a\in\set{\db,\sksym,\sssym}$. If $\tm \streq \tmtwo$ and $\tm \Rew{a} \tm'$ then there exists $\tmtwo'$ such that $\tmtwo \Rew{a} \tmtwo'$ and $\tm' \streq \tmtwo'$.
\end{proposition}

\subparagraph{Difference with Kesner et al's Skeletal \cbneed \cite{DBLP:journals/lmcs/KesnerPV24}.}
There is a small difference between our definition of Skeletal \cbneed and Kesner et al.'s, namely in the definition of skeletons. For us, the skeletal decomposition of, say, $\val \defeq\la\var\vartwo\var(\varthree\varthree)$ is $(\la\var\vartwo\var\varfour, \ctxhole\esub\varfour{\varthree\varthree})$, while for them it is $(\la\var\vartwo'\var\varfour, \ctxhole\esub{\vartwo'}\vartwo\esub\varfour{\varthree\varthree})$, that is, they \emph{flesh-out} also single variable occurrences while we do \emph{not}, obtaining a slightly more parsimonious decomposition.

This difference, however, has essentially no impact on the theory, nor on the asymptotic costs. In particular, the operational equivalence of Skeletal \cbneed with respect to \cbneed and \cbn is preserved. The proof of Kesner et al., indeed, works for all decompositions such that substituting the flesh into the skeleton recovers the original value; our decomposition has this property. Therefore, we can import the following result.

\begin{theorem}[Operational equivalence \cite{DBLP:journals/lmcs/KesnerPV24}]
Two $\l$-terms $\tm$ and $\tmtwo$ are \cbn observational equivalent if and only if they are Skeletal \cbneed observational equivalent.
\end{theorem}


\section{Skeletal \cbneed Can Bring an Exponential Asymptotic Speed-Up}
\label{sect:family}
In this section, we show that Skeletal \cbneed can be both exponentially faster and use exponentially less space than \cbneed, taking as reference cost models (following the terminology of Accattoli et al. \cite{DBLP:journals/lmcs/AccattoliLV24}):
\begin{itemize}
\item \emph{Abstract time}, that is the number of (distant) $\beta$-steps;
\item \emph{Ink space}, that is, the maximum size of the involved terms.
\end{itemize}
By smoothly adapting results in the literature \cite{DBLP:conf/birthday/SandsGM02,DBLP:conf/rta/AccattoliL12,DBLP:journals/pacmpl/ForsterKR20}, these can be proved to be reasonable cost models for time for \cbneed and for (super-)linear space for both \cbneed and Skeletal \cbneed. In fact, the relevant part of proving that abstract time is reasonable for Skeletal \cbneed shall be our own \refthm{SkMAD-is-bilinear} at the end of the paper.

 We prove the result by defining a family of terms $\set{\tm_n}_{n\in\nat}$ and showing that $\tm_n$ evaluates in an exponential number of $\db$ steps and producing terms of exponential size in \cbneed, while it uses only a linear number of $\skdb$ steps and terms of linear size in Skeletal \cbneed.

\subparagraph{The Family.} We shall use the abbreviations $\Id \defeq \la\var\var$ and $\gamma \defeq \la\vartwo\la\varthree{\vartwo \Id (\vartwo \Id) \varthree}$. Next, we define an auxiliary family $\set{\tmtwo_n}_{n\in\nat}$ as follows:
$\tmtwo_0 \defeq \Id$ and $\tmtwo_{n+1} \defeq \gamma \tmtwo_n$.
Lastly, the actual family is defined as $\tm_n \defeq (\la\var{\var \Id (\var \Id)}) \tmtwo_n$.

Our result is that the term $\tm_n$ evaluates in $\bigo(n)$  $\toskdb$-steps (namely $6n + 4$ steps) with skeletal \cbneed and
$\Omega(2^n)$  $\todb$-steps (namely $8 \times 2^n + n -4$ steps) with \cbneed.
This is shown via quite technical statements. Firstly, note that $\tm_n \todb (\var\Id (\var\Id)) \esub\var{\tmtwo_n}$; it is actually for $(\var\Id (\var\Id)) \esub\var{\tmtwo_n}$ that we shall prove bounds. These bounds are expressed stating that there are families of evaluation contexts $\sctx_n^{\ndsym}$ and $\sctx_n^{\sksym}$ (one for \cbneed and one for Skeletal \cbneed) such that
$(\var\Id (\var\Id)) \esub\var{\tmtwo_n}$ reduces to
$\sctx_n^{\ndsym}\ctxholep\Id$  and $\sctx_n^{\sksym}\ctxholep\Id$ (which is a normal form in both cases) in $6n + 3$ and $8 \times 2^n + n - 3$ steps respectively. The results about space are obtained by showing that $\sctx_n^{\ndsym}$ has size exponential in $n$, while $\sctx_n^{\sksym}$ is linear in $n$.

\subparagraph{Evaluation in Skeletal \cbneed.} The skeletal case is, perhaps surprisingly, the easiest case for which one can find an inductive invariant leading to the bound.

\begin{toappendix}
\begin{lemma}
\label{l:skneed-family-aux}\NoteProof{l:skneed-family-aux}
  Let $n \in \nat$ and $\tmthree_n \defeq (\var\Id(\var\Id))\esub\var{\tmtwo_n}$. Then there is a substitution context $\sctx_n^{\sksym}$ and an evaluation 
  $\deriv_n: \tmthree_n \toskneed^* \sctx_n^{\sksym}\ctxholep\Id$ such that $\sizeskdb{\deriv_n}=6n + 3$ and $\size{\sctx_n^{\sksym}} = \bigo(n)$.
\end{lemma}
\end{toappendix}

\subparagraph{Evaluation in \cbneed.} For the non-skeletal case, the bound is formulated via two families of substitution contexts.
For $n \in \nat$, the first family is given by
$\sctxv{0} \defeq \ctxhole\esub{\var_0}\Id$ and
$\sctxv{n+1} \defeq
\sctxv{n}\ctxholep{\ctxhole\esub{\var_{n+1}}{\tmthree_{n+1}}}$
with $\tmthree_{n+1} \defeq \la{\vartwo_n}\var_n\Id(\var_n\Id)\vartwo_n$. The second family $\sctx^{\ndsym}_n$ is defined in the proof of the lemma, itself relying on an auxiliary lemma in
\ifthenelse{\boolean{isarxiv}}{\refapp{app-family}.}{the technical report~\cite{accattoli2025costskeletalcallbyneedsmoothly}.}
The statement is parametrized by an evaluation context $\evctx$ for the induction to go through, but the relevant case is the one with $\evctx$ empty.

\begin{toappendix}
\begin{lemma}
\label{l:need-family-aux2}\NoteProof{l:need-family-aux2}
    Let $n \in \nat$ and $\evctx$ be an evaluation context.
  Then there exist a substitution context $\sctx^{\ndsym}_n$ and an evaluation
  $\deriv_n: \evctxfp{\var_n\Id(\var_n\Id)}\esub{\var_n}{\tmtwo_n}
  \toneed^* \sctxvp{n}{\evctxfp{\sctx^{\ndsym}_n\ctxholep\Id}}$ such that
  $\sizedb{\deriv_n} = 8 \times 2^n + n - 5$  and $\size{\sctx^{\ndsym}_n} = \Omega(2^n)$.
\end{lemma}
\end{toappendix}

The next statement sums up the results, and is based on the result at the end of the paper stating  that Skeletal \cbneed can be implemented in bilinear time, that is, linear in the size $\size{\tm_n}$ of the initial term and in the number of $\skdb$ steps (\refthm{SkMAD-is-bilinear}).

\begin{theorem}[Instance of exponential skeletal speed-up for both time and space]
\label{thm:family-speed-up}
There is a family $\set{\tm_n}_{n\in\nat}$ of $\l$-terms such that $\tm_n$ normalizes in $\bigo(n)$ abstract time and $\bigo(n)$ ink space with Skeletal \cbneed and $\Omega(2^n)$ abstract time and $\Omega(2^n)$ ink space with \cbneed.
\end{theorem}

\begin{proof}
\hfill
\begin{itemize}
\item For \cbneed, by \reflemma{need-family-aux2} there is an evaluation sequence:
\begin{center}
 $\tm_n =(\la\var{\var \Id (\var \Id)}) \tmtwo_n \todb 
 (\var\Id (\var\Id)) \esub\var{\tmtwo_n} \toneed^*   \sctxvp{n}{\sctx^{\ndsym}_n\ctxholep\Id}
 $
 \end{center}
taking $k_n \defeq 8 \times 2^n + n - 4 = \Omega(2^n)$ $\db$-steps and such that $\size{\sctx^{\ndsym}_n} = \Omega(2^n)$.

\item For Skeletal \cbneed, by \reflemma{skneed-family-aux} there is an evaluation sequence:
\begin{center}
 $\tm_n =(\la\var{\var \Id (\var \Id)}) \tmtwo_n \toskdb 
 (\var\Id (\var\Id)) \esub\var{\tmtwo_n} \toskneed^*   \sctx_n^{\sksym}\ctxholep\Id
 $
 \end{center}
taking $k_n \defeq 6n+4 = \bigo(n)$ $\skdb$-steps and such that $\size{\sctx_n^{\sksym}} = \bigo(n)$. Note that the calculus has no garbage collection rule, so the size of terms diminishes only by $\skdb$ steps, and of a constant amount (a $\skdb$ step removes two constructs and adds one). Therefore, the maximum size of terms is bounded by the size of the final term plus the number of steps. That is, the space cost is $\bigo(n)$.
\qedhere
\end{itemize}
\end{proof}


\section{A Rewriting-Based Algorithm for Computing the Skeleton}
\label{sect:algorithm}

In this section, we first define \emph{marked} terms and a marked reformulation of skeletal decompositions. Then, we use the marked setting to give a graphically-inspired rewriting system computing the skeletal decomposition of a value, giving a neat new presentation of ideas first developed by Shivers and Wand \cite{DBLP:journals/fuin/ShiversW10}. The rewriting system formalizes an algorithm. It is presented as a rewriting system as to smoothly manage it as an operational semantics notion.

\begin{figure}[t!]
\centering
\fbox{
$\begin{array}{rrllll}
  \textsc{Marked terms}&
   \mlterms\ni\mtm,\mtmtwo & \grameq  &\var\mid \var^{\taken} \mid  \la\var\mtm \mid  \lat\var\mtm  \mid \mtm\mtmtwo \mid \mtm\taken\mtmtwo \mid \skelu\mtm
  \\[4pt]
\textsc{Marked contexts} &
  \Mctxs\ni \mctx & \grameq  &\ctxhole \mid  \la\var\mctx \mid  \lat\var\mctx \mid \mctx\mtmtwo \mid \mtm\mctx
   \\
   &&& \mid \mctx\taken\mtmtwo \mid \mtm\taken\mctx  \mid \skelu\mctx
\end{array}$
}
\caption{Marked terms and contexts.}
\label{fig:marked-terms}
\end{figure}

\subparagraph{Marked Skeleton.} Marked terms and contexts are defined in \reffig{marked-terms}. The idea is to incrementally mark with a tag $\circ$ the constructs belonging to the skeleton during a visit of the term which is allowed to jump from an abstraction to the occurrences of the variables, since these are usually connected in a graphical representation. 

Constructs of the $\l$ calculus are then either the usual ones (referred to as \emph{unmarked}) or marked with $\taken$, which means that the constructor is part of the skeleton. Additionally, there is an extra \emph{up} construct $\skelu\tm$ that shall be used to indicate the frontier of the visit of the term. 

In the marked setting, we can recast skeletal decompositions by marking as white the constructs belonging to the skeleton and leaving unmarked all the constructs of the flesh.
\begin{definition}[Marked skeleton]
Let $\skelvars$ be a set of variables. The marked skeleton of $\tm$ relative to $\skelvars$ is defined as $\disc{\tm}{\skelvars}  \defeq  \tm$ if $\fv \tm \cap \skelvars = \emptyset$, otherwise:
  \begin{center}$\begin{array}{rll@{\hspace{.5cm}}|@{\hspace{.5cm}}rll}
   \disc{\var}{\skelvars} & \defeq & \var^\taken 
   &
   \disc{\tm\ap\tmtwo}{\skelvars} & \defeq & \disc{\mtm}{\skelvars} \appt \disc{\mtmtwo}{\skelvars} 
   \\
   \disc{\la\var\tm}{\skelvars} & \defeq & \lat{\var}{\disc{\tm}{\skelvars \cup \set\var}} 
  \end{array}$\end{center}
The marked skeleton of $\la\var\tm$ is defined as $\discabs{\la\var\tm} \defeq \lat\var\disc\tm{\set\var}$.
\end{definition}
Example: for $\val \defeq \la\var\la\vartwo\varthree\varthree\var(\vartwo\varthree)$, one has $\discabs{\val} = \lat\var\lat\vartwo(\varthree\varthree)\taken\var^{\taken}\taken(\vartwo^{\taken}\taken\varthree)$.

\subparagraph{From the Marked Skeleton to the Skeletal Decomposition.} Next, we turn marked skeletons into skeletal decompositions, by extracting unmarked sub-terms and turning them into the flesh, while removing the marks from the skeleton. We shall do this via a \emph{splitting function}, defined below, which operates on (parametrized) marked skeletons. For that, we give a lemma guaranteeing  that the unmarked sub-terms of marked skeletons can be recognized from their top constructor, that is in  $\bigo(1)$, thus removing the need to inspect them.

\begin{toappendix}
\begin{lemma}[Being unmarked is downward closed in marked skeletons]
\label{l:unmarked-downward-closed}\NoteProof{l:unmarked-downward-closed}
Let $\disc{\tm}{\skelvars} = \mctxp{\mtm}$ with $\mtm$ starting unmarked. Then $\mtm$ has no marks.
\end{lemma}
\end{toappendix}

The lemma implies the correctness of the first clause of the following definition, as it ensures that in that case $\mtm$ has no marks.

\begin{definition}[Splitting]
Let $\mtm=\disc{\tm}{\skelvars}$ for some $\tm$ and $\skelvars$. The following splitting function
separates the skeleton from the flesh and removes all marks:
\begin{center}$\arraycolsep=2pt
\begin{array}{rrll}
\splitfun{\mtm} & \defeq & (\varthree, \ctxhole\esub\varthree{\mtm}) & \text{if $\mtm$ starts unmarked,  is not a variable, and $\varthree$ is fresh;} \\
\splitfun{\var} & \defeq & (\var,\ctxhole); & \\
\splitfun{\var^\taken} & \defeq & (\var,\ctxhole); & \\
\splitfun{\lat\var\mtmtwo} & \defeq & (\la\var\tmtwo, \sctx) &
\text{with $\splitfun{\mtmtwo} = (\tmtwo,\sctx)$;} \\
\splitfun{\mtm_1\appt\mtm_2} & \defeq & (\tm_1\ap\tm_2,\sctxptwo\sctx) &
\text{with $\splitfun{\mtm_1} = (\tm_1,\sctx)$ and
$\splitfun{\mtm_2} = (\tm_2,\sctxtwo)$}.
\end{array}$
\end{center}
\end{definition}
Example: $\splitfun{\lat\var\lat\vartwo(\varthree\varthree)\taken\var^{\taken}\taken(\vartwo^{\taken}\taken\varthree)}=(\la\var\la\vartwo\varfour\var(\vartwo\varthree), \ctxhole\esub\varfour{\varthree\varthree})$.

We can now prove that the marked notions allow one to retrieve the standard ones. 
\begin{toappendix}
\begin{proposition}[Soundness of the marked skeleton]
\label{prop:marked-skeleton-soundness}\NoteProof{prop:marked-skeleton-soundness}
Let $\tm \in \terms$ and $\skelvars$ be a set of variables.
\begin{enumerate}
\item \emph{Auxiliary parametrized statement}: $\splitfun{\disc\tm\skelvars} = \skeldec\tm\skelvars$;
\item \emph{Soundness}: $\skeldecabs{\la\var\tm} = \splitfun{\discabs{\la\var\tm}}$.
\end{enumerate}
\end{proposition}
\end{toappendix}

\subparagraph{Computing the Marked Skeleton.} We shall now define a parallel algorithm for computing the marked skeleton via a rewriting system on marked terms. The definition is in \reffig{skeleton-rewriting}. The algorithm shall be applied to abstractions, say to $\la\var\tm$. The algorithm is described by inserting the \emph{up} symbols $\frontiersym$, denoting where the algorithm is operating, and propagating through the term. At the beginning, there is exactly one $\frontiersym$, as the algorithm starts on $\la\var\skelu\tm$.

The algorithm has two sets of rules. Propagation rules turn the encountered unmarked constructs into marked ones. The interesting rule is $\tomarkblthree$ for abstractions, that is also the very first one to be applied on the initial term $\la\var\skelu\tm$. Its effect is that $\lambda\var$ is marked as $\l^\taken\var$ and all the occurrences of $\var$ in $\tm$ are also marked, obtaining $\lat\var\tm\isub\var{\var^\taken}$; this is how we recast the role played by Shivers and Wand's bi-directional edges connecting abstractions and variables. Essentially, propagation rules move $\frontiersym$ upwards but they also add $\frontiersym$ at the bottom of the syntax tree on the variable occurrences associated to marked abstractions.

Absorption rules, instead, remove the $\frontiersym$ symbol when it encounters a marked construct, as it means that the concerned thread of the algorithm is passing where another thread already passed before, thus the concerned thread becomes redundant and can be terminated. 

Intuitively, the set of $\frontiersym$-sub-terms is the frontier of where the parallel algorithm is operating. When the algorithm is over, no node is marked with $\frontiersym$. 

The algorithm is parallel in that the initial substitution $\tm\isub\var{\var^\taken}$ and rule $\rtomarkblthree$ might introduce many occurrences of $\frontiersym$, which can be propagated independently.
\begin{figure}[t!]
\centering
\fbox{
\begin{tabular}{c|c}
$\begin{array}{rcl}
\multicolumn{3}{c}{\textsc{Propagating rules}}
\\[2pt]
\skelu{\mtm} \mtmtwo & \rtomarkblone & \skelu{\mtm\appt \mtmtwo}
\\
\mtm \skelu{\mtmtwo} & \rtomarkbltwo & \skelu{\mtm\appt \mtmtwo}
\\
\la\var\skelu{\mtm} & \rtomarkblthree & \skelu{\lat\var\mtm\isub\var{\skelu{\var^\taken}}}
\\[2pt]
\hline
\multicolumn{3}{c}{\textsc{Absorbing rules}}
\\[2pt]
\skelu{\mtm}\appt \mtmtwo & \rtomarkwhone & \mtm\appt \mtmtwo
\\
\mtm\appt \skelu{\mtmtwo} & \rtomarkwhtwo & \mtm\appt \mtmtwo
\\
\lat\var\skelu{\mtm} & \rtomarkwhthree & \lat\var\mtm
\end{array}$

&
\begin{tabular}{c}
$\textsc{Initialization}$
\\
$\skeld{\la\var\tm} \defeq \la\var\skelu\tm$
\\[4pt]
\hline
\textsc{Contextual closures}
\\[4pt]
	$\Rew{a} \defeq \Mctxsp{\rootRew{a}}$ \\ for $a\in\set{\prsym1,\prsym2,\prsym3,\absym1,\absym2,\absym3}$
\\[8pt]
\hline
$\begin{array}{lll}
\multicolumn{3}{c}{\textsc{Algorithm}}
\\
\tomark & \defeq & \tomarkblone\cup\tomarkbltwo\cup\tomarkblthree\cup
\\
&&\tomarkwhone\cup\tomarkwhtwo\cup\tomarkwhthree
\end{array}$
\\[4pt]
\end{tabular}
\end{tabular}
}
\caption{Rewriting rules for marked terms.}
\label{fig:skeleton-rewriting}
\end{figure}


As an example, consider $\val \defeq \la\var\la\vartwo\varthree\varthree\var(\vartwo\varthree)$, for which one has, for instance: 
\begin{center}
$\begin{array}{lclcllllllll}
\skeld{\val} 
&= &
\la\var\skelu{\la\vartwo\varthree\varthree\var(\vartwo\varthree)}
&\tomarkblthree &
\skelu{\lat\var\la\vartwo\varthree\varthree\skelu{\var^{\taken}}(\vartwo\varthree)}
\\[3pt]
&\tomarkbltwo &
\skelu{\lat\var\la\vartwo\skelu{(\varthree\varthree)\taken\var^{\taken}}(\vartwo\varthree)}
&\tomarkblone &
\skelu{\lat\var\la\vartwo\skelu{(\varthree\varthree)\taken\var^{\taken}\taken(\vartwo\varthree)}}
\\[3pt]
&\tomarkblthree &
\skelu{\lat\var\skelu{\lat\vartwo(\varthree\varthree)\taken\var^{\taken}\taken(\skelu{\vartwo^{\taken}}\varthree)}}
&\tomarkblone &
\skelu{\lat\var\skelu{\lat\vartwo(\varthree\varthree)\taken\var^{\taken}\taken\skelu{\vartwo^{\taken}\taken\varthree}}}
\\[3pt]
&\tomarkwhtwo\tomarkwhthree &
\skelu{\lat\var\lat\vartwo(\varthree\varthree)\taken\var^{\taken}\taken(\vartwo^{\taken}\taken\varthree)}
&= &
\skelu{\discabs\val}
\end{array}$
\end{center}

\begin{proposition}
\label{prop:algorithm-sn-confluence}
Reduction $\tomark$ is strongly normalizing and diamond (thus confluent).
\end{proposition}

\begin{proof}
\applabel{prop:algorithm-sn-confluence}\hfill
\begin{enumerate}
\item \emph{Strong normalization}. Let $\mtm$ be a marked term. As terminating measure, consider $(\sizep\mtm{\neg\taken},\sizep\mtm\frontiersym)$ where $\sizep\mtm{\neg\taken}$ and $\sizep\mtm\frontiersym$ are the number of unmarked and $\frontiersym$ constructs in $\mtm$. Note that the propagation rules decrease $\sizep\mtm{\neg\taken}$ while the absorption rules decrease $\sizep\mtm\frontiersym$ and leave $\sizep\mtm{\neg\taken}$ unchanged. Then $\tomark$ is strongly normalizing.
\item \emph{Diamond}. We prove the diamond property, which implies confluence. Clearly, redexes cannot duplicate or erase each other. There are only two critical pairs, namely:
\begin{center}
\begin{tabular}{c@{\hspace{.5cm}} |@{\hspace{.5cm}} c}
\begin{tikzpicture}[ocenter]
		\node at (0,0)[align = center](source){\normalsize $\mctxp{\skelu{\mtm} \skelu\mtmtwo}$};
		\node at (source.center)[right = 40pt](source-right){\normalsize $\mctxp{\mtm\appt \skelu\mtmtwo}$};
		\node at (source.center)[below = 25pt](source-down){\normalsize $\mctxp{\skelu{\mtm}\appt \mtmtwo}$};
		\node at (source-right|-source-down)(target){\normalsize $\mctxp{\mtm\appt \mtmtwo}$};
		
		\draw[->](source) to node[above] {\scriptsize $\prsym1$} (source-right);
		\draw[->](source) to node[left] {\scriptsize $\prsym2$}(source-down);	
		\draw[->, dotted](source-right) to node[right] {\scriptsize $\absym2$}(target);
		\draw[->, dotted](source-down) to node[above] {\scriptsize $\absym1$} (target);
	\end{tikzpicture}
&
\begin{tikzpicture}[ocenter]
		\node at (0,0)[align = center](source){\normalsize $\mctxp{\skelu{\mtm}\appt \skelu\mtmtwo}$};
		\node at (source.center)[right = 40pt](source-right){\normalsize $\mctxp{\mtm\appt \skelu\mtmtwo}$};
		\node at (source.center)[below = 25pt](source-down){\normalsize $\mctxp{\skelu{\mtm}\appt \mtmtwo}$};
		\node at (source-right|-source-down)(target){\normalsize $\mctxp{\mtm\appt \mtmtwo}$};
		
		\draw[->](source) to node[above] {\scriptsize $\absym1$} (source-right);
		\draw[->](source) to node[left] {\scriptsize $\absym2$}(source-down);	
		\draw[->, dotted](source-right) to node[right] {\scriptsize $\absym2$}(target);
		\draw[->, dotted](source-down) to node[above] {\scriptsize $\absym1$} (target);
	\end{tikzpicture}\qedhere
	\end{tabular}
\end{center}
\end{enumerate}
\end{proof}

To state the correctness of the algorithm and give a bound on the number of its steps, we need two definitions.

\begin{definition}
Let $\mtm$ be a marked term.
\begin{itemize}
\item \emph{Starts with $\taken$/unmarked/$\frontiersym$}:
$\mtm$ \emph{starts with $\taken$} (resp. \emph{unmarked}, resp. \emph{with} $\frontiersym$) if it has shape $\var^\taken$, $\lat\var\mtmtwo$, or $\mtmtwo\appt\mtmthree$ (resp. $\var$, $\la\var\mtmtwo$, or $\mtmtwo\mtmthree$, resp. $\skelu\mtm$).
\item \emph{White size}: the white size $\sizet\mtm$ of $\mtm$ is the number of $\taken$ constructs in $\mtm$.
\end{itemize}
\end{definition}

The following lemma is the key step in the proof of correctness of the algorithm. It is a technical lemma, the statement of which is more easily understood by looking at the proof of the following theorem.
\begin{toappendix}
\begin{lemma}[Crucial]
\label{l:disc-comp-step}\NoteProof{l:disc-comp-step}
Let $\tm$ be a term, $\skelvars$ a set of variables, $\var \in\fv\tm\setminus \skelvars$, $\mtm \defeq \disc\tm\skelvars$, and $\mtmtwo \defeq \disc\tm{\skelvars \cup \{\var\}}$. Then:
\begin{center}$
\mtm\isub\var{\skelu{\var^\taken}} \tomark^{k} \begin{cases}
    \skelu{\mtmtwo} & \text{with $k = \sizet{\mtmtwo} - \sizet{\mtm}-1$, if $\mtm$ starts unmarked;} \\
    \mtmtwo & \text{with $k =\sizet{\mtmtwo} - \sizet{\mtm}$, if $\mtm$ starts with $\taken$.}
\end{cases}
$\end{center}
\end{lemma}
\end{toappendix}

\begin{theorem}[Soundness of the algorithm]
\label{thm:number-steps-algorithm}
Let $\tm$ be a term. Then
$\skeld{\la\var\tm} \tomark^n \skelu{\discabs{\la\var\tm}}$ with $n = \sizet{\discabs{\la\var\tm}}$.
\end{theorem}

\begin{proof}
By definition $\skeld{\la\var\tm} = \la\var\skelu{\tm} \tomarkblthree \skelu{\lat\var\tm\isub\var{\skelu{\var^\taken}} }= \skelu{\lat\var\disc\tm\emptyset\isub\var{\skelu{\var^\taken}}}$, while $\discabs{\la\var\tm} = \lat\var\disc\tm{\set\var}$.
We consider two cases.
\begin{itemize}
\item \case{$\var \notin \fv{\tm}$} Then $\disc\tm{\{\var\}} = \tm$ and $\tm\isub\var{\skelu{\var^\taken}} = \tm$, so $\lat\var\tm\isub\var{\skelu{\var^\taken}}=\lat\var\tm=\lat\var\disc\tm{\{\var\}}=\discabs{\la\var\tm}$. Thus, $\skeld{\la\var\tm} \tomarkblthree \skelu{\discabs{\la\var\tm}}$. Moreover, $n = 1 = \sizet{\discabs{\la\var\tm}}$, since the only white construct of $\discabs{\la\var\tm}$ is the root abstraction.
\item \case{$\var \in \fv{\tm}$}
By \reflemma{disc-comp-step}, $\lat\var\tm\isub\var{\skelu{\var^\taken}} \tomark^k \lat\var\skelu{\disc\tm{\{\var\}}}$, where $k = \sizet{\disc\tm{\{\var\}}} - \sizet{\tm} -1 = \sizet{\disc\tm{\{\var\}}}-1$.
Therefore:
\begin{center}
$\begin{array}{lclcll}
\skeld{\la\var\tm} & \tomarkblthree & \skelu{\lat\var\tm\isub\var{\skelu{\var^\taken}} }

& \tomark^{k} & 
\skelu{\lat\var\skelu{\disc\tm{\{\var\}}}}
\\
& \tomarkwhthree & 
\skelu{\lat\var\disc\tm{\{\var\}} }
& = & 
\skelu{\discabs{\la\var\tm}}.
\end{array}$
\end{center}
Now, $n = k + 2 = \sizet{\disc\tm{\{\var\}}} -1 +2 = \sizet{\disc\tm{\{\var\}}} +1= \sizet{\discabs{\la\var\tm}}$.\qedhere
\end{itemize}
\end{proof}

\subsection{Complexity Analysis} 
We proved that the number of rewriting steps of the algorithm for computing the marked skeleton $\discabs{\val}$ is exactly its marked size $\sizet{\discabs{\la\var\tm}}$ (\refthm{number-steps-algorithm}), which is the size of the skeleton (\refprop{marked-skeleton-soundness}). The complexity of the algorithm is indeed linear but it requires some discussion, since $\tomarkblthree$ steps are not constant time operations. Moreover, we need to take into account the cost of splitting.

\subparagraph{Complexity Analysis 1: Propagation + Absorption.} For computing $\discabs{\val}$ from $\val$, we assume that abstractions have pointers to the occurrences of the variable, so that the substitution $\mtm\isub\var{\skelu{\var^\taken}}$ in $\tomarkblthree$ does not require going through the whole of $\mtm$. This is Shivers and Wand's key idea, and also how our implementation works (see \refapp{implementation}). A global argument gives the cost: the algorithm turns $\sizet{\discabs{\val}}$ unmarked constructs into $\taken$-constructs exactly once, and it generates at most $\sizet{\discabs{\val}}$ $\frontiersym$-constructs that are never duplicated and are absorbed exactly once, thus it globally requires $\bigo(\sizet{\discabs{\val}})$ time.
\begin{lemma}[Cost of propagation + absorption]
Computing $\discabs{\val}$ from $\val$ requires time $\bigo(\sizet{\discabs{\val}})$.
\end{lemma}

\subparagraph{Complexity Analysis 2: Splitting.} For the cost of splitting, \reflemma{unmarked-downward-closed} guarantees that the unmarked sub-terms of marked skeletons can be recognized in $\bigo(1)$ by only inspecting their topmost constructor, so that implementing the splitting function takes time proportional to the number of marked constructors only.

\begin{lemma}[Cost of splitting]
Computing $\splitfun{\discabs{\val}}$ from $\discabs{\val}$ requires time $\bigo(\sizet{\discabs{\val}})$.
\end{lemma}

The next theorem sums it all up, recasting the obtained costs in the unmarked setting via the soundness of the algorithm (\refprop{marked-skeleton-soundness}).
\begin{theorem}[The algorithm is linear in the skeleton]
Let $\val\in\terms$ be a value. Then the skeletal decomposition $(\skelabs\val,\flesh\val)$ of $\val$ can be computed in $\bigo(\size{\skelabs\val})$.
\end{theorem}

\section{Preliminaries about Abstract Machines}
\label{sect:prel-machines}
In this section, we introduce terminology and basic concepts about abstract machines, that shall be the topic of the following sections.

\subparagraph{Abstract Machines Glossary.}  Abstract machines manipulate \emph{pre-terms}, that is, terms without implicit $\alpha$-renaming. In this paper, an \emph{abstract 
machine} is a quadruple $\mach = (\States, \tomach, \compilrel\cdot\cdot, \decode\cdot)$ the components of which are as follows.
\begin{itemize}

\item \emph{States.} A state $\state\in\States$ is a quadruple $\fourstate\chain\tm\stack\env$ composed by the \emph{active term} $\tm$, plus three data structure, namely the \emph{chain} $\chain$, the \emph{(applicative) stack} $\stack$, and the \emph{(global) environment} $\env$. Terms in states are actually pre-terms.

\item  \emph{Transitions.} The pair $(\States, \tomach)$ is a transition 
system with transitions $\tomach$ partitioned into \emph{principal transitions}, whose union is noted $\tomachpr$ and that are meant to correspond to steps on the calculus, and \emph{search transitions}, whose union is noted $\tomachsea$, that take care of searching for (principal) redexes.

\item \emph{Initialization.} The component $\compilrel{}{}\subseteq\terms\times\States$ is the \emph{initialization relation} associating closed terms without ESs to 
initial states. It is a \emph{relation} and not a function because $\compilrel\tm\state$ maps a \emph{closed} $\l$-term $\tm$ (considered modulo $\alpha$) to a state $\state$ having a \emph{pre-term representant} of $\tm$ (which is not modulo $\alpha$) as active term. Intuitively, any two states $\state$ and $\statetwo$ such that $\compilrel\tm\state$ and $\compilrel\tm\statetwo$ are $\alpha$-equivalent. 
A state $\state$ is \emph{reachable} if it can be reached starting from an initial state, that is, if $\statetwo \tomach^*\state$ where $\compilrel\tm\statetwo$ for some $\tm$ and $\statetwo$, shortened as $\compilrel\tm\statetwo \tomach^*\state$.

\item \emph{Read-back.} The read-back function $\decode\cdot:\States\to\absterms$ turns reachable states into 
terms (possibly with ESs, that is, $\absterms = \esterms$ or $\absterms = \skterms$) and satisfies the \emph{initialization constraint}: if $\compilrel\tm\state$ then $\decode{\state}=_\alpha\tm$.
\end{itemize}
A state is \emph{final} if no transitions apply.
 A \emph{run} $\run: \state \tomach^*\statetwo$ is a possibly empty finite sequence of transitions, the length of which is noted 
$\size\run$; note that the first and the last states of a run are not necessarily initial and final. 
If $a$ and $b$ are transitions labels (that is, $\tomachhole{a}\subseteq \tomach$ and 
$\tomachhole{b}\subseteq \tomach$) then $\tomachhole{a,b} \defeq \tomachhole{a}\cup \tomachhole{b}$ and $\sizep\run a$ 
is the number of $a$ transitions in $\run$, $\sizep\run {a,b} \defeq \sizep\run a + \sizep\run b $, and similarly for more than two labels. 

For the machines at work in this paper, the pre-terms in initial states shall be \emph{well-bound}, that is, they have pairwise distinct bound names; for instance $(\la\var\var)\la\vartwo\vartwo$ is well-bound while $(\la\var\var)\la\var\var$ is not. 
We shall also write $\renamenop{\tm}$ in a state $\state$ for a \emph{fresh well-bound renaming} of $\tm$,
\ie $\renamenop{\tm}$ is $\alpha$-equivalent to $\tm$, well-bound, and its bound variables
are fresh with respect to those in $\tm$ and in the other components of $\state$.

\subparagraph{Mechanical Bismulations.} Machines are usually showed to be correct with respect to a strategy via some form of bisimulation relating terms and machine states. The notion that we adopt is here dubbed \emph{mechanical bisimulation}, and it should not be confused with the strong bisimulation property of structural equivalence (which actually plays a role in mechanical bisimulations). The definition, tuned towards complexity analyses, requires a perfect match between the steps of the evaluation sequence and the  principal transitions of the machine run. \emph{Terminology}: a structural strategy $(\tostrat,\streq)$ is a strategy $\tostrat$ plus a relation $\streq$ that is a strong bisimulation with respect to $\tostrat$ (see \refprop{streq-is-a-strong-bisim}).

\begin{definition}[Mechanical bisimulation]
\label{def:implem}
A machine $\mach=(\States, \tomach, \compilrel\cdot\cdot, \decode\cdot)$ and a structural strategy $(\tostrat,\streq)$ on $\skterms$-terms are mechanical bisimilar when, given an initial state $\compilrel\tm\state$:
\begin{enumerate}
\item \emph{Runs to evaluations}: for any run $\run: \compilrel\tm\state \tomach^* \statetwo$ there exists an evaluation $\deriv: \tm \tostrat^* \streq\decode\statetwo$;

\item \emph{Evaluations to runs}: for every evaluation $\deriv: \tm \tostrat^* \tmtwo$ there exists a 
run $\run: \compilrel\tm\state \tomach^* \statetwo$ such that $\decode\statetwo \streq \tmtwo$;

\item \emph{Principal matching}: for every principal transition $\tomachhole{\symfont{a}}$ of label $a$ of $\mach$, in both previous points the number $\sizep\run{\symfont{a}}$ of $\symfont{a}$-transitions in $\run$ is exactly the number $\sizep\deriv{a}$ of $\symfont{a}$-steps in the evaluation $\deriv$, \ie $\sizep\deriv{\symfont{a}} = \sizep\run{\symfont{a}}$.
\end{enumerate}
\end{definition}

The proof that a machine and a strategy are in a mechanical bisimulation follows from some basic properties, grouped under the notion of distillery, following Accattoli et al. \cite{DBLP:conf/icfp/AccattoliBM14}.
\begin{definition}[Distillery]
  \label{def:distillery}
  A machine $\mach=(\States, \tomach, \compilrel\cdot\cdot, \decode\cdot)$ and a structural strategy $(\tostrat,\streq)$ are a \emph{distillery} if the following conditions hold:
  \begin{enumerate}
		\item\label{p:def-beta-projection} \emph{Principal projection}: $\state \tomachhole{\symfont{a}} \statetwo$ implies $\decode\state \Rew{\symfont{a}} \streq\decode\statetwo$ for every principal transition of label $\symfont{a}$;
		\item\label{p:def-overhead-transparency} \emph{Search transparency}: $\state \tomachsea \statetwo$ implies $\decode\state = \decode\statetwo$;
		\item\label{p:def-overhead-terminate}	\emph{Search transitions terminate}:  $\tomachsea$ terminates;

	\item\label{p:def-determinism} \emph{Determinism}: $\tostrat$ is deterministic;

	\item\label{p:def-progress} \emph{Halt}: $\mach$ final states decode to $\tostrat$-normal terms.
  \end{enumerate}
\end{definition}

\begin{toappendix}
\begin{theorem}[Sufficient condition for implementations]
\label{thm:abs-impl}\NoteProof{thm:abs-impl}
  Let a machine $\mach$ and a structural strategy $(\tostrat,\streq)$ be a distillery.  Then, they are mechanical bisimilar.
\end{theorem}
\end{toappendix}

\section{The MAD and the Skeletal MAD}
\label{sect:skeletal-MAD}
In this section, we study an abstract machine for skeletal \cbneed, obtained as a minor modification of a known machine for \cbneed, that is recalled here. Since the two machines share most data structures and transitions, they are presented together in \reffig{cbneed-machine} (slightly abusing notations: the two machines define environments $\env$ differently, so the notation of the common transitions is overloaded).

\begin{figure}[t!]
\centering
\small
$\begin{array}{c}
\begin{array}{r@{\hspace{.25cm}} rcl  @{\hspace{.5cm}}  @{\hspace{.5cm}} r @{\hspace{.25cm}} rcl}
	\textsc{Stacks} & \stack 	& \grameq & \stempty \mid \tm \cons \stack
	&
	\textsc{Chains} &	\chain	& \grameq & \stempty \mid \chain\cons(\var, \stack, \env\esub\var\cdot)
	\\

	\textsc{MAD envs} & \env	& \grameq & \stempty \mid \esub\var\tm \cons \env 
	&
	\textsc{States} &	\state	& \defeq & \glamst\chain\tm\stack\env
	\\
	\textsc{Skel. MAD envs} & \env	& \grameq & \stempty \mid \esub\var\tm \cons \env \mid \skesub\var\val \cons \env 	
	&
	\textsc{Init.} &	\tm	& \compilrel{}{} & \fourstate\emptylist{\renamenop\tm}\emptylist\emptylist \ \ (\#)
\end{array}
\\

  	{\setlength{\arraycolsep}{0.3em}
  	\begin{array}{|c|c|c|c||l||c|c|c|c|l}
	\hhline{-|-|-|-|-|-|-|-|-}
	\multicolumn{9}{|c|}{\textsc{Common transitions}}
	\\
		 \mbox{Chain} &\mbox{Code} & \mbox{Stack} &  \mbox{Env} 
		&&
		 \mbox{Chain} &\mbox{Code} & \mbox{Stack} &  \mbox{Env} \\
\hhline{=|=|=|=|=|=|=|=|=}
		 \chain & \tm\tmtwo & \stack &  \env
	  	&\tomachseaone&
	  	 \chain & \tm & \tmtwo\cons\stack & \env
	  	\\
%
		 \chain & \la\var\tm & \tmtwo \cons\stack  &\env
		& \tomachbeta &
		\chain &  \tm & \stack  &\esub\var\tmtwo\cons\env
		\\
		
		 \chain & \var & \stack & \env\cons\esub\var\tm\cons\envtwo
	  	&\tomachseatwo &
		 \chain\cons(\var, \stack, \env\esub\var\cdot) & \tm & \stempty & \envtwo 
		 & (*)
	  	\\
		
	  	\chain\cons(\var, \stack, \env\esub\var\cdot) &  \val & \stempty &  \envtwo
		& \tomachseathree &
		\chain &  \var & \stack &  \env\cons\esub\var{\val}\cons\envtwo
		\\
		\hhline{=|=|=|=|=|=|=|=|=}
		\multicolumn{9}{|c|}{\textsc{MAD specific transition}}
		\\
		 \chain & \var & \stack & \env\cons\esub\var\val\cons\envtwo
	  	&\tomachsub &
		 \chain & \renamenop\val & \stack & \env\cons\esub\var\val\cons\envtwo
		 &(\#)
	  	\\
		\hhline{=|=|=|=|=|=|=|=|=}
		\multicolumn{9}{|c|}{\textsc{Skeletal MAD specific transitions}}
		\\		
		\chain & \var & \stack & \env\cons\esub\var\val\cons\envtwo
	  	&\tomachsk &
		\chain & \var & \stack & \env\cons\skesub\var{\val'}\cons\env_\sctx\cons\envtwo
		&(\%)
	  	\\
		
		\chain & \var & \stack & \env\cons\skesub\var\val\cons\envtwo
	  	&\tomachss &
		\chain & \renamenop{\val} & \stack & \env\cons\skesub\var\val\cons\envtwo
		&(\#)
	  	\\

\hhline{-|-|-|-|-|-|-|-|-}
	\end{array}}
\\[4pt]
\begin{array}{l}
  \mbox {(\#) $\renamenop{\tm}$ is any well-bound code $\alpha$-equivalent to $\tm$ such that
  its bound names are fresh}
  \\ \mbox{with respect to those in the rest of the state.}
  \\ \mbox{(*) If $\tm$ is not a value \quad
  (\%) Where $\skeldecabs\val=(\valtwo,\sctx)$ is the skeletal decomposition of $\val$}
  \\ \mbox{(computed as $\splitfun{\discabs{\val}}$), and $\env_\sctx$ is $\sctx$ seen as an environment.}
\end{array}
\\[4pt]
\begin{array}{r@{\hspace{.25cm}} rcl  @{\hspace{.25cm}}|@{\hspace{.25cm}}  r @{\hspace{.25cm}} rcl}
	\multicolumn{8}{c}{\textsc{Read-back}}
	\\
	\textsc{Empty} & \decode\stempty & \defeq & \ctxhole
	&
	\textsc{Envs} & \decode{\esub\var\tm \cons \env } & \defeq & \decodep\env{\ctxhole\esub\var\tm}
	\\
	\textsc{Stacks} & \decode{\tm \cons \stack} & \defeq & \decodep\stack{\ctxhole\tm}
	&
	&\decode{\skesub\var\tm \cons \env } & \defeq & \decodep\env{\ctxhole\skesub\var\tm}
	\\
	\textsc{Chains} &	\decode{\chain\cons(\var, \stack, \env\esub\var\cdot)}	& \defeq & \decodep\env{\decodep\chain{\decodep\stack\var}}\esub\var\ctxhole
	&
	\textsc{States} &	\decode{\fourstate\chain\tm\stack\env} & \defeq & \decodep\env{\decodep\chain{\decodep\stack\tm}}
\end{array}
\end{array}$
\caption{The Milner Abstract machine by Need (MAD) and the Skeletal Milner Abstract machine by Need (Skeletal MAD) presented as different extension of a common set of transitions (but be careful: the common transitions use different notions of environment).}
\label{fig:cbneed-machine}
\end{figure}

\subparagraph{The MAD.} The \emph{Milner Abstract machine by neeD} (MAD) of \reffig{cbneed-machine} implements the \cbneed strategy $\toneed$. It is a \cbneed variant of the \emph{Milner Abstract Machine} (MAM), itself a simplification of the Krivine abstract machine (KAM). The MAM is a \cbn abstract machine using a single global environment (sometimes called \emph{heap}), a stack for arguments, and no closures---it is omitted here (the KAM instead uses closures and many local environments).  Both the MAM and the MAD are introduced by Accattoli et al. in \cite{DBLP:conf/icfp/AccattoliBM14}.

The MAD modifies the MAM by adding a mechanism realizing the memoization characteristic of \cbneed. Namely, when execution encounters a variable occurrence $\var$ associated to an environment entry $\esub\var\tmtwo$ such that $\tmtwo$ is not a value, then the MAM duplicates $\tmtwo$ while the MAD does not. Instead, part of the current state is saved on the new \emph{chain} data structure $\chain$, and the MAD starts evaluating $\tmtwo$ inside the environment entry itself; this is transition $\tomachseatwo$. When such an evaluations ends with a value $\val$, the machines resumes the state saved on the chain, and the environment is updated replacing $\tmtwo$ with $\val$---this is transition $\tomachseathree$---so that future encounters of $\var$ will not need to re-evaluate $\tmtwo$. Both the MAM and the MAD rely on $\alpha$-renaming as a black-box operation $\renamenop\cdot$, at work in transition $\tomachsub$.

The presentation in \reffig{cbneed-machine} differs from the MAD in \cite{DBLP:conf/icfp/AccattoliBM14} in a very minor point. In \cite{DBLP:conf/icfp/AccattoliBM14}, transitions $\tomachseathree$ and $\tomachsub$ are concatened into a single transition. Splitting them helps in stating the invariants of the machines.

\subparagraph{The Skeletal MAD.} The Skeletal MAD, also in \reffig{cbneed-machine}, modifies the MAD exactly as the Skeletal \cbneed LSC modifies the \cbneed LSC, that is, by adding skeletal entries $\skesub\var\val$ to the global environment and splitting the substitution transtition $\tomachsub$ into two: the skeletonizing transition $\tomachsk$ that separates the skeleton $\valtwo$ from the flesh of the value $\val$, and the skeletal substitution transition $\tomachss$ that substitutes the ($\alpha$-renamed) skeleton $\renamenop\valtwo$. 

\textit{\textbf{Notation}}. For specifying transition $\tomachsk$, we use the following notation for the flesh: a substitution context $\sctx=\ctxhole\semsub{\var_1}{\tm_1}\ldots\semsub{\var_k}{\tm_k}$, where each $\semsub{\var_i}{\tm_i}$ can be $\esub{\var_i}{\tm_i}$ or $\skesub{\var_i}{\tm_i}$, induces a global environment $\env_\sctx \defeq \semsub{\var_1}{\tm_1}\cons\ldots\cons\semsub{\var_k}{\tm_k}\cons\emptylist$.

The union of all the transition rules of the Skeletal MAD (namely, $\seasym_1$, $\seasym_2$, $\seasym_3$, $\beta$, $\sksym$, and $\sssym$) is noted $\tosmad$.

The read-back of (Skeletal) MAD states to terms (possibly with ESs) is defined in \reffig{cbneed-machine}. The definition uses auxiliary notions of read-back for chains, stacks, and environments that turn them into \cbneed evaluation contexts. For stacks and environments, their read-back is clearly an evaluation context. For chains, it shall be proved as an invariant of the machine.

\subparagraph{Invariants.} To prove properties of the Skeletal MAD we need to isolate some of its invariants in \reflemma{skeletal-mam-qual-invariants} below. The closure one ensures that the closure of the initial term extends, in an appropriate sense, to all reachable states. The well-bound invariant, similarly, ensures that binders in reachable states are pairwise distinct, as in the initial term. To compactly formulate the closure invariant we need the notion of terms of a state. 

\begin{definition}[Terms of a state]
Let $\state=\fourstate\chain{\tmtwo}{\stack}{\env}$ be a Skeletal MAD state where $\chain$ is a possibly empty chain $ \emptylist\cons(\var_k, \stack_k, \env_k\cons\esub{\var_k}\cdot)\cons\ldots\cons(\var_1, \stack_1, \env_1\cons\esub{\var_1}\cdot)$ for some $k\geq0$. 
The \emph{terms of $\state$} are $\tmtwo$, every term in $\stack$ and $\stack_i$, and every term in an entry of $\env$ or $\env_i$, for $1\leq i \leq k$. 
\end{definition}

\begin{toappendix}
\begin{lemma}[Qualitative invariants]
\label{l:skeletal-mam-qual-invariants}\NoteProof{l:skeletal-mam-qual-invariants}
Let $\compilrel\tm\state \tomach^*\statetwo=\fourstate\chain{\tmtwo}{\stack}{\env}$ be a Skeletal MAD run.
\begin{enumerate}
\item \emph{Closure}: for every term $\tmthree$ of $\statetwo$  and for any variable $\var\in\fv\tmthree$ there is an environment entry $\esub\var\tm$, $\skesub\var\val$, or $\esub\var\cdot$ on the right of $\tmthree$ in $\statetwo$.
\item \emph{Well-bound}: if $\la\var\tmthree$ occurs in $\statetwo$ and $\var$ has any other
occurrence in $\statetwo$ then it is as a free variable of $\tmthree$, and for any substitution $\esub\vartwo\tm$, $\skesub\vartwo\val$, or $\esub\vartwo\cdot$ in $\statetwo$ the name $\vartwo$ can occur (in any form) only on the left of that entry in $\statetwo$.

\item \emph{Contextual chain read-back}: $\decode\chain$, $\decodep\chain{\decode\stack}$, $\decodep\env{\decode\chain}$ and $\decodep\env{\decodep\chain{\decode\stack}}$ are \cbneed evaluation contexts.
\end{enumerate}
\end{lemma}
\end{toappendix}

\subparagraph{Mechanical Bisimulation.} The invariants allow us to prove that the Skeletal MAD and the strategy $\toskneed$ modulo $\streq$ form a distillery, from which the mechanical bisimulation follows (by \refthm{abs-impl}). The closure invariant is used for Point 4, the other two for Point 1.
\begin{toappendix}
\begin{theorem}[Skeletal distillery]
\label{thm:SkMAD-properties}\NoteProof{thm:SkMAD-properties}
The Skeletal MAD and the structural strategy $(\toskneed,\streq)$ form a distillery. Namely, let $\state$ be a reachable Skeletal MAD state.
\begin{enumerate}
\item \emph{Principal projection}: 
\begin{enumerate}
\item if $\state \tomachb \statetwo$ then $\decode\state \toskdb\streq \decode\statetwo$.
\item if $\state \tomachsk \statetwo$ then $\decode\state \tosk \decode\statetwo$.
\item if $\state \tomachss \statetwo$ then $\decode\state \toss \decode\statetwo$.
\end{enumerate}

\item \emph{Search transparency}: if $\state \tomachhole{\seasym_1,\seasym_2,\seasym_3} \statetwo$ then $\decode\state = \decode\statetwo$.

\item \emph{Search terminates}: $\tomachhole{\seasym_1,\seasym_2,\seasym_3}$ is terminating.

\item \emph{Halt}: if $\state$ is final then it has shape $\fourstate\emptylist\val\emptylist\env$ and $\decode\state$ is $\toskneed$-normal.

\end{enumerate}
\end{theorem}
\end{toappendix}

\begin{corollary}[Skeletal mechanical bisimulation]
The Skeletal MAD and the structural strategy $(\toskneed,\streq)$ are in a mechanical bisimulation.
\end{corollary}

\section{Complexity Analysis}
\label{sect:complexity}
In this section, we show that the Skeletal MAD can be concretely implemented within the same complexity of the MAD, that is, with \emph{bi-linear} overhead.

\subparagraph{Sub-Term Property.} As it is standard, the complexity analysis crucially relies on the sub-term property, that bounds the size of manipulated (and thus duplicated) terms using the size of the initial term. The proof requires a similar property about skeletons.
\begin{toappendix}
\begin{lemma}[Skeleton and flesh size]
\label{l:skeleton-size}\NoteProof{l:skeleton-size}\hfill
\begin{enumerate}
\item \emph{Auxiliary parametrized statement}: let $\tm \in \terms$, $\skelvars$
  be a set of variables, and $\skeldec\tm\skelvars \defeq (\tmtwo,\sctx)$. Then
  $\size\tmtwo \leq \size\tm$ and $\size\tmthree \leq\size\tm$ for any ES $\esub\var\tmthree$ in $\sctx$;
\item \emph{Main}: let $\val\in\terms$. Then $\size{\skelabs\val} \leq \size\val$ and $\size\tmthree \leq\size\val$ for any ES $\esub\var\tmthree$ in $\flesh\val$.
\end{enumerate}
\end{lemma}
\end{toappendix}

\begin{toappendix}
\begin{lemma}[Quantitative sub-term property]
\label{l:sub-term}\NoteProof{l:sub-term}
Let $\compilrel\tm\state \tosmad^*\statetwo$ be a Skeletal MAD run. Then $\size\tmtwo\leq\size\tm$ for any term $\tmtwo$ of $\statetwo$.
\end{lemma}
\end{toappendix}

Next, some basic observations about the transitions, together with the sub-term property, allow us to bound their number using the two key parameters (that is, number of $\beta$-steps/transitions and size of the initial term).

\begin{proposition}[Number of transitions]
Let $\run:\compilrel\tm\state \tosmad^*\statetwo$ be a Skeletal MAD run. 
\begin{enumerate}
\item $\sizep\run{\sksym, \sssym, \seasym_2, \seasym_3}\in\bigo(\sizep\run\beta)$, and 
\item $\sizep\run{\seasym_1}\in\bigo(\size{\tm}\cdot(\sizep\run\beta +1))$.
\end{enumerate}
\end{proposition}

\begin{proof}
\hfill
\begin{enumerate}
\item
\begin{enumerate}
\item $\sizep\run{\seasym_2} \leq \sizep\run\beta$, because every $\seasym_2$ transition consumes one non-value entry from the environment, which are created only by $\beta$ transitions.

\item $\sizep\run{\seasym_3} \le \sizep\run\beta$, because every $\seasym_3$ transition consumes one entry from the chain, which are created only by $\seasym_2$ transitions, which in turn are bound by $\beta$ transitions.

\item $\sizep\run{\sssym} \in \bigo(\sizep\run\beta)$, because
after a $\sssym$ transition there can be either a $\seasym_3$ transition, a $\beta$ transition, or no transition (if the $\sssym$ transition is the last one of the run $\run$). Therefore, $\sizep\run{\sssym} \leq \sizep\run{\seasym_3} + \sizep\run\beta +1 \leq 2\sizep\run\beta +1$.

\item $\sizep\run\sksym \in \bigo(\sizep\run\beta)$, because every $\sksym$ transition is followed by a $\sssym$ transition.
\end{enumerate}

\item Note that $\seasym_1$ transitions decrease the size of the code, which is increased only by transitions $\sssym$ and $\seasym_2$. By the sub-term property (\reflemma{sub-term}), the code increase is bounded by the size $\size\tm$ of the initial term. By Point 1, $\sizep\run{\sssym,\seasym_2} =\bigo(\sizep\run\beta)$. Then $\sizep\run{\seasym_1} \in \bigo(\size{\tm}\cdot(\sizep\run{\sssym,\seasym_2} +1))=\bigo(\size{\tm}\cdot(\sizep\run\beta +1))$.\qedhere
\end{enumerate}
\end{proof}

Lastly, we need some assumptions on how the Skeletal MAD can be concretely implemented. Our OCaml implementation (see \refapp{implementation}) represents variables as memory locations and variable occurrences as pointers to those locations, obtaining random access to environment entries in $\bigo(1)$. As already mentioned in \refsect{algorithm}, all constructors have pointers to their parents in the syntactic tree, and in particular variables  have pointers to their occurrences, as to implement efficiently the computation of the skeleton. Moreover, environments are implemented as doubly linked lists, to split/concatenate them in $\bigo(1)$ in transitions $\seasym_2$ and $\seasym_3$. The cost of transitions $\sksym$ and $\sssym$ is bound by the size of the value to skeletonize/copy, which is bound by the sub-term property.

\begin{lemma}[Cost of single transitions]
  Let $\compilrel\tm\state \tosmad^*\statetwo$ be a run. Implementing any transitions from $\statetwo$ costs $\bigo(1)$ but for $\sksym$ and $\sssym$ transitions, which cost $\bigo(\size\tm)$ each.
\end{lemma}

Putting all together, we obtain a bilinear bound.
\begin{theorem}[The Skeletal MAD is bi-linear]
\label{thm:SkMAD-is-bilinear}
Let $\run: \compilrel\tm\state \tosmad^*\statetwo$ be a Skeletal MAD run. Then $\run$ can be implemented on random access machines in $\bigo(\size{\tm}\cdot(\sizep\run\beta +1))$.
\end{theorem}

\section{Conclusions}
\label{sect:conclusions}

We show that skeletal call-by-need provides, in some cases, exponentially shorter evaluation sequences involving exponentially smaller terms than call-by-need. We also show that the required skeleton reconstruction can be implemented in linear time and space, by giving a new simpler and graph-free presentation of ideas by Shivers and Wand. Lastly, we smoothly plug this result in the existing distillation technique for abstract machines, obtaining an implementation of skeletal call-by-need with bi-linear overhead. In particular, the bi-linear overhead allows us to establish that the shrinking of evaluation sequences is a real time speed-up: the smart specification of skeletal call-by-need does not hide an expensive overhead.
\bibliography{main.bbl}

\withproofs{
\newpage
\appendix
\setboolean{appendix}{false}
\section{OCaml Implementation}
\label{sect:implementation}

We provide at \url{https://github.com/Franciman/fullylazymad}
a reference implementation in OCaml of the Skeletal MAD,
meant to guarantee the soundness of the assumptions that underly the complexity analyses.
The implementation takes in input a user provided term and it shows the full run printing all transitions and all intermediate machine states. We highlight here a few details on the implementation.

\subparagraph{Data structures.}
Pre-terms are represented by the following algebraic data type (ADT):
\begin{lstlisting}
type term =
  | Var of { v: var_ref; mutable taken: bool; mutable parent: term option }
  | Abs of { v: var_ref; mutable body: term; mutable occurrences: term list; mutable taken: bool; mutable parent: term option }
  | App of { mutable head: term; mutable arg: term; mutable taken: bool; mutable parent: term option }
and var_ref =
 { mutable prev : var_ref option;
   orig_name : string; (* to help generating new names *)
   name : string;
   mutable sub : sub;
   mutable next : var_ref option }
and sub =
  NoSub | Sub of term | SubSkel of term | Hole | Copy of (term list ref * var_ref)
\end{lstlisting}
whose constructors \lstinline{App} for applications, \lstinline{Abs} for abstractions and \lstinline{Var} for variable occurrence hold mutable pointers to their subterms, a mutable optional pointer to the parent and a mutable boolean \lstinline{taken} for marking terms during skeleton extraction.

Variable occurrences point to a shared, unique variable declaration/substitution in memory of type \lstinline{var_ref}, which holds the variable name, a substitution description field of type \lstinline{sub} and two optional pointers (\lstinline{prev} and \lstinline{next}) to insert the variable in an environment, that is a double-linked lists of variable declarations. A substitution description is an ADT whose constructors are \lstinline{Sub $\tm$} for the explicit substitution $\esub\var\tm$, \lstinline{SubSkel $\tm$} for the skeletal substitutions $\skesub\var\val$, \lstinline{NoSub} for bound variables, \lstinline{Hole} to represent $\esub\var\cdot$ in chain items and \lstinline{Copy (parents,$\var$)} to point to another variable declaration $\var$ together with the list of its parents. The latter form is used temporarily to preserve sharing of variable declarations during the implementation of $\renamenop\cdot$: the first time a variable is to be copied its substitution is make to point to the new copy and each time an occurrence of the variable is to be copied, the same variable declaration is reused and a new parent is added for it.

Abstractions also point to the list of occurrences of the bound variable, to implement rule $\rtomarkblthree$ in $\bigo(1)$. This is a minor memory-saving divergence from the paper where it was suggested to have backpointers to the occurrences stored in the bound variable declaration node, of type \lstinline{var_ref}: once an abstraction is consumed, the backpointers become useless and thus avoiding them in the implementation saves some space.

Machine states are implemented straightforwardly using lists for stacks, chains, and the set of roots (which is then actually represented as a list, not as a set), and doubly linked lists of variable declarations for environments:
\begin{lstlisting}
type env = var_ref option
  (* Some head of the doubly-linked list, or None for empty list *)
type stack = term list                                                                                   
type chain = (var_ref * stack * env) list
type state = chain * term * stack * env
\end{lstlisting}

\subparagraph{Skeleton Reconstruction.}
The rewriting system for marking terms to extract the skeleton and flesh is implemented sequentially as a recursive function.
Each activation record for a call to \lstinline{mark_skeleton} in the OCaml stack corresponds
in the paper to an occurrence of $\skelu\cdot$ in the rewritten term.
\begin{lstlisting}
let rec mark_skeleton =
 function
  | None -> ()
  | Some (Var v) -> (v.taken <- true; mark_skeleton v.parent)
  | Some (Abs a) ->
      if not a.taken
      then (a.taken <- true;
            List.iter (fun v -> mark_skeleton (Some v)) a.occurrences;
            mark_skeleton a.parent)
  | Some (App a) ->
      if not a.taken
      then (a.taken <- true; mark_skeleton a.parent)
\end{lstlisting}
Parallel implementations of \lstinline{mark_skeleton} are possible, but because of the necessity of syncronization over application nodes and because the size of skeletonized terms is supposed to be small, it is unclear if actual parallelization would provide any speed-up.

\subparagraph{Remaining Code.}
The code to extract the skeleton and flesh from marked terms, and the code for machine transitions follow
straightforwardly the rules in the paper, up to noisy code to keep the doubly linked structure of terms and environments. Here is an example of the code for one transition in the Skeletal MAD:
\newcommand{\quotx}{"}
\begin{lstlisting}
let step : state -> string * state =
 function
  ...
  | chain, Abs { v; body; _ }, arg :: args, env ->
      set_parent body None;  (* unliks body from its parent (double link severed)*)
      v.sub <- Sub arg;      (* turns v into an explicit substitution *)
      push v env;            (* pushes v on top of the env *)
      $\quotx$$\ensuremath{\beta}$$\quotx$,(chain, body, args, Some v)
  ...
\end{lstlisting}

The code that implements $\renamenop\cdot$ also follows the standard algorithm to copy a graph in linear
time preserving the sharing, described by Accattoli and Barras in \cite{DBLP:conf/ppdp/AccattoliB17}. It uses the \lstinline{Copy} constructor to temporarily associate
each variable to its copy.

\section{Example of Skeletal MAD Run}
\label{sect:machine-run}

We show here an example of invocation of our implementation.
Suppose that the file \texttt{family\_3.lambda} contains the following input that encodes the pure
term that yields the third term of the family we studied, where backslashes are ASCII art equivalents
of $\lambda$.
\begin{tiny}
  \begin{lstlisting}{latex}
(\i. (\g. (\z. (z i) (z i)) (g (g (g i)))) (\x.\y. (x i) (x i) y)) (\w. w)
\end{lstlisting}
\end{tiny}
The command
\begin{tiny}
  \begin{lstlisting}{text}
dune exec bin/main.exe fully-lazy-functional-linked-env family_3.lambda
  \end{lstlisting}
\end{tiny}
yields the following output\footnote{
  The option \texttt{fully-lazy-functional-linked-env} selects the Skeletal MAD implementation described in the paper, that stays as close as possible to the abstract machine description.
The repository also contains code for alternative implementations of the Skeletal MAD, comprising versions that keep the machine state garbage free.}:

\begin{tiny}
$$
\begin{array}{ll}
 & {\color{red}}|\underline{(\lambda i.(\lambda g.(\lambda z.zi(zi))(g(g(gi))))(\lambda x.\lambda y.xi(xi)y))(\lambda w.w)}||{\color{dgreen}}|{\color{gray}}\\
\rightarrow_{sea_1} & {\color{red}}|\underline{\lambda i.(\lambda g.(\lambda z.zi(zi))(g(g(gi))))(\lambda x.\lambda y.xi(xi)y)}|(\lambda w.w)|{\color{dgreen}}|{\color{gray}}\\
\rightarrow_{\beta} & {\color{red}}|\underline{(\lambda g.(\lambda z.zi(zi))(g(g(gi))))(\lambda x.\lambda y.xi(xi)y)}||{\color{dgreen}\esub{i}{\lambda w.w}}|{\color{gray}}\\
\rightarrow_{sea_1} & {\color{red}}|\underline{\lambda g.(\lambda z.zi(zi))(g(g(gi)))}|(\lambda x.\lambda y.xi(xi)y)|{\color{dgreen}\esub{i}{\lambda w.w}}|{\color{gray}}\\
\rightarrow_{\beta} & {\color{red}}|\underline{(\lambda z.zi(zi))(g(g(gi)))}||{\color{dgreen}\esub{g}{\lambda x.\lambda y.xi(xi)y}:\esub{i}{\lambda w.w}}|{\color{gray}}\\
\rightarrow_{sea_1} & {\color{red}}|\underline{\lambda z.zi(zi)}|g(g(gi))|{\color{dgreen}\esub{g}{\lambda x.\lambda y.xi(xi)y}:\esub{i}{\lambda w.w}}|{\color{gray}}\\
\rightarrow_{\beta} & {\color{red}}|\underline{zi(zi)}||{\color{dgreen}\esub{z}{g(g(gi))}:\esub{g}{\lambda x.\lambda y.xi(xi)y}:\esub{i}{\lambda w.w}}|{\color{gray}}\\
\rightarrow_{sea_1} & {\color{red}}|\underline{zi}|zi|{\color{dgreen}\esub{z}{g(g(gi))}:\esub{g}{\lambda x.\lambda y.xi(xi)y}:\esub{i}{\lambda w.w}}|{\color{gray}}\\
\rightarrow_{sea_1} & {\color{red}}|\underline{z}|i:zi|{\color{dgreen}\esub{z}{g(g(gi))}:\esub{g}{\lambda x.\lambda y.xi(xi)y}:\esub{i}{\lambda w.w}}|{\color{gray}}\\
\rightarrow_{sea_2} & {\color{red}(z,i:zi,\esub{z}{.})}|\underline{g(g(gi))}||{\color{dgreen}\esub{g}{\lambda x.\lambda y.xi(xi)y}:\esub{i}{\lambda w.w}}|{\color{gray}}\\
\rightarrow_{sea_1} & {\color{red}(z,i:zi,\esub{z}{.})}|\underline{g}|g(gi)|{\color{dgreen}\esub{g}{\lambda x.\lambda y.xi(xi)y}:\esub{i}{\lambda w.w}}|{\color{gray}}\\
\rightarrow_{sk} & {\color{red}(z,i:zi,\esub{z}{.})}|\underline{g}|g(gi)|{\color{dgreen}\skesub{g}{\lambda x.\lambda y.xi(xi)y}:\esub{i}{\lambda w.w}}|{\color{gray}}\\
\rightarrow_{ss} & {\color{red}(z,i:zi,\esub{z}{.})}|\underline{\lambda x_{1}.\lambda y_{2}.x_{1}i(x_{1}i)y_{2}}|g(gi)|{\color{dgreen}\skesub{g}{\lambda x.\lambda y.xi(xi)y}:\esub{i}{\lambda w.w}}|{\color{gray}}\\
\rightarrow_{\beta} & {\color{red}(z,i:zi,\esub{z}{.})}|\underline{\lambda y_{2}.x_{1}i(x_{1}i)y_{2}}||{\color{dgreen}\esub{x_{1}}{g(gi)}:\skesub{g}{\lambda x.\lambda y.xi(xi)y}:\esub{i}{\lambda w.w}}|{\color{gray}}\\
\rightarrow_{sea_3} & {\color{red}}|\underline{z}|i:zi|{\color{dgreen}\esub{z}{\lambda y_{2}.x_{1}i(x_{1}i)y_{2}}:\esub{x_{1}}{g(gi)}:\skesub{g}{\lambda x.\lambda y.xi(xi)y}:\esub{i}{\lambda w.w}}|{\color{gray}}\\
\rightarrow_{sk} & {\color{red}}|\underline{z}|i:zi|{\color{dgreen}\skesub{z}{\lambda y_{2}.p_{3}y_{2}}:\esub{p_{3}}{x_{1}i(x_{1}i)}:\esub{x_{1}}{g(gi)}:\skesub{g}{\lambda x.\lambda y.xi(xi)y}:\esub{i}{\lambda w.w}}|{\color{gray}}\\
\rightarrow_{ss} & {\color{red}}|\underline{\lambda y_{4}.p_{3}y_{4}}|i:zi|{\color{dgreen}\skesub{z}{\lambda y_{2}.p_{3}y_{2}}:\esub{p_{3}}{x_{1}i(x_{1}i)}:\esub{x_{1}}{g(gi)}:\skesub{g}{\lambda x.\lambda y.xi(xi)y}:\esub{i}{\lambda w.w}}|{\color{gray}}\\
\rightarrow_{\beta} & {\color{red}}|\underline{p_{3}y_{4}}|zi|{\color{dgreen}\esub{y_{4}}{i}:\skesub{z}{\lambda y_{2}.p_{3}y_{2}}:\esub{p_{3}}{x_{1}i(x_{1}i)}:\esub{x_{1}}{g(gi)}:\skesub{g}{\lambda x.\lambda y.xi(xi)y}:\esub{i}{\lambda w.w}}|{\color{gray}}\\
\rightarrow_{sea_1} & {\color{red}}|\underline{p_{3}}|y_{4}:zi|{\color{dgreen}\esub{y_{4}}{i}:\skesub{z}{\lambda y_{2}.p_{3}y_{2}}:\esub{p_{3}}{x_{1}i(x_{1}i)}:\esub{x_{1}}{g(gi)}:\skesub{g}{\lambda x.\lambda y.xi(xi)y}:\esub{i}{\lambda w.w}}|{\color{gray}}\\
\rightarrow_{sea_2} & {\color{red}(p_{3},y_{4}:zi,\esub{y_{4}}{i}:\skesub{z}{\lambda y_{2}.p_{3}y_{2}}:\esub{p_{3}}{.})}|\underline{x_{1}i(x_{1}i)}||{\color{dgreen}\esub{x_{1}}{g(gi)}:\skesub{g}{\lambda x.\lambda y.xi(xi)y}:\esub{i}{\lambda w.w}}|{\color{gray}}\\
\rightarrow_{sea_1} & {\color{red}(p_{3},y_{4}:zi,\esub{y_{4}}{i}:\skesub{z}{\lambda y_{2}.p_{3}y_{2}}:\esub{p_{3}}{.})}|\underline{x_{1}i}|x_{1}i|{\color{dgreen}\esub{x_{1}}{g(gi)}:\skesub{g}{\lambda x.\lambda y.xi(xi)y}:\esub{i}{\lambda w.w}}|{\color{gray}}\\
\rightarrow_{sea_1} & {\color{red}(p_{3},y_{4}:zi,\esub{y_{4}}{i}:\skesub{z}{\lambda y_{2}.p_{3}y_{2}}:\esub{p_{3}}{.})}|\underline{x_{1}}|i:x_{1}i|{\color{dgreen}\esub{x_{1}}{g(gi)}:\skesub{g}{\lambda x.\lambda y.xi(xi)y}:\esub{i}{\lambda w.w}}|{\color{gray}}\\
\rightarrow_{sea_2} & {\color{red}(p_{3},y_{4}:zi,\esub{y_{4}}{i}:\skesub{z}{\lambda y_{2}.p_{3}y_{2}}:\esub{p_{3}}{.}):(x_{1},i:x_{1}i,\esub{x_{1}}{.})}|\underline{g(gi)}||{\color{dgreen}\skesub{g}{\lambda x.\lambda y.xi(xi)y}:\esub{i}{\lambda w.w}}|{\color{gray}}\\
\rightarrow_{sea_1} & {\color{red}(p_{3},y_{4}:zi,\esub{y_{4}}{i}:\skesub{z}{\lambda y_{2}.p_{3}y_{2}}:\esub{p_{3}}{.}):(x_{1},i:x_{1}i,\esub{x_{1}}{.})}|\underline{g}|gi|{\color{dgreen}\skesub{g}{\lambda x.\lambda y.xi(xi)y}:\esub{i}{\lambda w.w}}|{\color{gray}}\\
\rightarrow_{ss} & {\color{red}(p_{3},y_{4}:zi,\esub{y_{4}}{i}:\skesub{z}{\lambda y_{2}.p_{3}y_{2}}:\esub{p_{3}}{.}):(x_{1},i:x_{1}i,\esub{x_{1}}{.})}|\underline{\lambda x_{5}.\lambda y_{6}.x_{5}i(x_{5}i)y_{6}}|gi|{\color{dgreen}\skesub{g}{\lambda x.\lambda y.xi(xi)y}:\esub{i}{\lambda w.w}}|{\color{gray}}\\
\rightarrow_{\beta} & {\color{red}(p_{3},y_{4}:zi,\esub{y_{4}}{i}:\skesub{z}{\lambda y_{2}.p_{3}y_{2}}:\esub{p_{3}}{.}):(x_{1},i:x_{1}i,\esub{x_{1}}{.})}|\underline{\lambda y_{6}.x_{5}i(x_{5}i)y_{6}}||{\color{dgreen}\esub{x_{5}}{gi}:\skesub{g}{\lambda x.\lambda y.xi(xi)y}:\esub{i}{\lambda w.w}}|{\color{gray}}\\
...
\end{array}
$$
Number of betas: 24\\
Interpreter fully-lazy-functional-linked-env, final result: $\lambda w.w$
\end{tiny}

~\\
where in each machine state the {\color{red}chain} is red, the {\underline{code}} is underlined, the stack has no special marking and the {\color{dgreen} environemnt} is green.

\setboolean{appendix}{true}

\ifthenelse{\boolean{isarxiv}}{
\section{Proofs Omitted from \refsect{skeletal} (Skeletal Call-by-Need)}
\label{sect:app-skeletal}

\subparagraph{Proof of Determinism.}
We first need an auxiliary result lemma.
\begin{lemma}
	\label{l:CbNeedEvCtx}
	Let $\tm:=\evctxfp\var$ for an evaluation context $\evctx$. Then:
	\begin{enumerate}
		\item \label{p:CbNeedEvCtx-1} \emph{Not an answer}: $\tm\neq\sctxp\val$ for every substitution context $\sctx$ and abstraction $\val$;
		\item \label{p:CbNeedEvCtx-2} \emph{Unique decomposition}: for every evaluation context $\evctxtwo$ and variable $\vartwo$, $\tm=\evctxtwofp\vartwo$ implies $\evctxtwo=\evctx$ and $\vartwo=\var$;
		\item \label{p:CbNeedEvCtx-3} \emph{Normal}: $\tm$ is a $\toskneed$ normal form.
	\end{enumerate}
\end{lemma}
\begin{proof}
	In all points we use a structural induction on $\evctx$. 
	\begin{enumerate}
	\item \emph{Not an answer}:
	\begin{itemize}
		\item $\evctx=\ctxhole$: obvious.
		\item $\evctx=\evctx_1\tmtwo$: obvious.
		\item $\evctx=\evctx_1\esub\vartwo\tmthree$: suppose that $\sctx=\sctxtwo\esub\vartwo\tmthree$ (for otherwise the result is obvious); then we apply the \ih to $\evctx_1$ to obtain $\evctx_1\ctxholep\var\neq\sctxtwop\val$.
		\item $\evctx=\evctx_1\skesub\vartwo\val$: analogous to the previous case.
		\item $\evctx=\evctx_1\ctxholefp\vartwo\esub\vartwo{\evctx_2}$: suppose that $\sctx=\sctxtwo\esub\vartwo{\evctx_2\ctxholep\var}$ (for otherwise the result is obvious); then we apply the \ih to $\evctx_1$ to obtain $\evctx_1\ctxholefp\vartwo\neq\sctxtwop\val$.
	\end{itemize}
	
\item \emph{Unique decomposition}:
	\begin{itemize}
		\item $\evctx=\ctxhole$: obvious.
		\item $\evctx=\evctx_1\tmtwo$: then necessarily $\evctxtwo=\evctxtwo_1\tmtwo$, and it follows from the \ih
		\item $\evctx=\evctx_1\skesub\varthree\tmtwo$: then necessarily $\evctxtwo=\evctx_1'\skesub\varthree\tmtwo$, and it follows from the \ih
		\item $\evctx=\evctx_1\esub\varthree\tmtwo$: two sub-cases. 
		\begin{enumerate}
		\item $\evctxtwo=\evctxtwo_1\esub\varthree\tmtwo$: it follows from the \ih, as above.
		\item $\evctxtwo=\evctxtwo_1\ctxholefp\varthree\esub\varthree{\evctxtwo_2}$, with $\evctxtwo_2\ctxholefp\vartwo=\tmtwo$. This is actually impossible. In fact, it implies $\evctx_1\ctxholefp\var=\evctxtwo_1\ctxholefp\varthree$, which by \ih gives us $\varthree=\var$, contradicting the hypothesis that $\evctx$ does not capture $\var$.
		\end{enumerate}
		\item $\evctx=\evctx_1\ctxholefp\varthree\esub\varthree{\evctx_2}$: by symmetry with the above case, the only possibility is $\evctxtwo=\evctx_1\ctxholefp\varthree\esub\varthree{\evctxtwo_2}$, that follows from the \ih
	\end{itemize}

\item \emph{Normal}:  let $\tmfour$ be a redex (\ie, a term matching the left-hand side of $\rtodb$, $\rtosk$, or $\rtoss$) and let $\evctxtwo$ be an evaluation context. We shall show by structural induction on $\evctx$ that $\tm\neq\evctxtwop\tmfour$. We shall do this by considering, in each inductive case, all the possible shapes of $\evctxtwo$.
	\begin{itemize}
		\item $\evctx=\ctxhole$: obvious.
		\item $\evctx=\evctx_1\tmtwo$: two sub-cases.
		\begin{enumerate}
			\item  $\evctxtwo=\ctxhole$. Since $\tmfour$ is a redex, the only possibility is $\tmfour=\sctxp\val\tmtwo$ for some abstraction $\val$ and substitution context $\sctx$. By point 1, $\evctx_1\ctxholep\var\neq\sctxp\val$, so this case is impossible.
			\item  $\evctxtwo=\evctxtwo_1\tmtwo$. Then it follows from the \ih (on $\evctx_1$). 
		\end{enumerate}
		\item $\evctx=\evctx_1\skesub\vartwo\val$: two sub-cases.
		\begin{itemize}
			\item $\evctxtwo=\ctxhole$:  this time, the fact that $\tmfour$ is a redex forces $\tmfour=\evctxtwo_1\ctxholefp\vartwo\skesub\vartwo\val$. Note that $\evctx_1\ctxholefp\var=\evctxtwo_1\ctxholefp\vartwo $ so that by Point 2 the two decompositions coincide. This is absurd because then $\var\neq\vartwo$ while by hypothesis $\evctx$ does not capture $\var$.
			\item $\evctxtwo=\evctxtwo_1\skesub\vartwo\tmtwo$. Then it follows from the \ih
		\end{itemize}
		\item $\evctx=\evctx_1\esub\vartwo\tmtwo$: three sub-cases.
		\begin{itemize}
			\item $\evctxtwo=\ctxhole$: this time, the fact that $\tmfour$ is a redex forces $\tmfour=\evctxtwo_1\ctxholefp\vartwo\esub\vartwo\tmtwo$ and $\tmtwo=\sctxp\val$. Note that $\evctx_1\ctxholefp\var=\evctxtwo_1\ctxholefp\vartwo $ so that by Point 2 the two decompositions coincide. This is absurd because then $\var\neq\vartwo$ while by hypothesis $\evctx$ does not capture $\var$.
			\item $\evctxtwo=\evctxtwo_1\esub\vartwo\tmtwo$: it follows from the \ih on $\evctx_1$.
			\item $\evctxtwo=\evctxtwo_1\ctxholefp\vartwo\esub\vartwo{\evctxtwo_2}$ with $\evctxtwo_2\ctxholep\tmfour=\tmtwo$: note that $\evctx_1\ctxholefp\var=\evctxtwo_1\ctxholefp\vartwo $ so that by Point 2 the two decompositions coincide. This is absurd because then $\var\neq\vartwo$ while by hypothesis $\evctx$ does not capture $\var$.
		\end{itemize}
		\item $\evctx=\evctx_1\ctxholep\vartwo\esub\vartwo{\evctx_2}$: three sub-cases.
		\begin{itemize}
			\item $\evctxtwo=\ctxhole$: the fact that $\tmfour$ is a redex implies $\tmfour=\evctxtwo_1\ctxholep\vartwo\esub\vartwo{\sctxp\val}$. This case is impossible: by Point 2, $\evctxtwo_1=\evctx_1$, which implies $\evctx_2\ctxholep\var\neq\sctxp\val$, that contradicts Point 1.
			\item $\evctxtwo=\evctxtwo_1\esub\vartwo{\evctx_2\ctxholep\var}$: the specular case has already been treated above.
			\item $\evctxtwo=\evctx_1\ctxholep\vartwo\esub\vartwo{\evctxtwo_2}$: it follows from the \ih on $\evctx_2$.\qedhere
		\end{itemize}
	\end{itemize}
\end{enumerate}
\end{proof}

\gettoappendix{skneed-determinism}
\begin{proof}
\applabel{skneed-determinism}
Let $\tm$ be a term that decomposes in two ways as an evaluation context surrounding a root redex, that is, $\tm=\evctx_1\ctxholep{\tmfour_1}=\evctx_2\ctxholep{\tmfour_2}$. We prove that $\evctx_1=\evctx_2$ and $\tmfour_1=\tmfour_2$ by structural induction on $\tm$. Cases:
\begin{itemize}
 \item \emph{Variable or abstraction}
 Impossible, since variables and abstractions are both call-by-need normal.
 
 \item \emph{Application, \ie\ $\tm=\tmtwo\tmthree$} Suppose that one of the two evaluation contexts, for instance $\evctx_1$, is equal to $\ctxhole$. Then, we must have $\tmtwo=\l\var.\tmtwo'$, but in that case it is easy to see that the result holds, because $\evctx_2$ cannot have its hole to the right of an application (in $\tmthree$) or under an abstraction (in $\tmtwo'$). We may then assume that none of $\evctx_1,\evctx_2$ is equal to $\ctxhole$. Then, they have to have shape $\evctx_1=\evctx_1'\tmthree$ and $\evctx_2=\evctx_2'\tmthree$. The required equalities then are given by the \ih
  
 \item \emph{Skeletal substitution, \ie\ $\tm=\tmtwo\skesub\var\val$}  Cases:
 \begin{itemize}
 \item \caselight{Both contexts have their holes in $\tmtwo$}. It follows from the \ih. 
 
 \item \caselight{One of the contexts---say $\evctx_1$---is empty, \ie\  $\tmtwo=\evctx_3\ctxholefp\var$, $\tmthree=\val$, and $\tmfour_1=\evctx_3\ctxholefp\var\skesub\var {\val}$}. Then $\evctx_2$ is also empty because $\evctxfp{\var}$ is normal by \reflemmap{CbNeedEvCtx}{3}.
\end{itemize}
    
 \item \emph{Explicit substitution, \ie\ $\tm=\tmtwo\esub\var\tmthree$}
 Cases:
 \begin{itemize}
 \item \caselight{Both contexts have their holes in $\tmtwo$ or $\tmthree$}. It follows from the \ih. 
 
 \item \caselight{One of the contexts---say $\evctx_1$---is empty, \ie\  $\tmtwo=\evctx_3\ctxholefp\var$, $\tmthree=\sctxp\val$, and $\tmfour_1=\evctx_3\ctxholefp\var\esub\var {\sctxp\val}$}. Then $\evctx_2$ is also empty. Indeed:
 \begin{enumerate}
 \item The hole of $\evctx_2$ cannot be in $\sctxp\val$, because it is normal, and 
 \item It cannot be inside $\evctx_3\ctxholefp\var$ because by Lemma \ref{l:CbNeedEvCtx}.\ref{p:CbNeedEvCtx-3} $\evctxfp{\var}$ is normal.
 \end{enumerate}
 
 \item \caselight{One of the contexts---say $\evctx_1$---has its hole in $\tmthree$ and the other one has its hole in $\tmtwo$, \ie\  $\evctx_1=\evctx_3\ctxholefp\var\esub\var{\evctx_4}$ and $\evctx_2=\evctx_5\esub\var\tmthree$}. This case is impossible, because by Lemma \ref{l:CbNeedEvCtx}.\ref{p:CbNeedEvCtx-3} $\evctx_3\ctxholefp\var$ is normal. \qedhere
 \end{itemize}
 \end{itemize}
 \end{proof}
 

%
%
%
%
%
\section{Proofs from \refsect{family}}
\label{sect:app-family}

\gettoappendix{l:skneed-family-aux}
\begin{proof}
\applabel{l:skneed-family-aux}
  By induction on $n$.
  \begin{itemize}
    \item \case{$n = 0$}
     \[\begin{array}{lllll}
	\tmthree_0 
	&= &(\var\Id(\var\Id))\esub\var\Id \\
	&\tosk &(\var\Id(\var\Id))\skesub\var\Id \\
	&\toss &(\Id\Id(\var\Id))\skesub\var\Id \\
	&\red\toskdb &\avar\esub\avar\Id(\var\Id)\skesub\var\Id \\
	&\tosk &\avar\skesub\avar\Id(\var\Id)\skesub\var\Id \\
	&\toss &\Id\skesub\avar\Id(\var\Id)\skesub\var\Id \\
	&\red\toskdb &\avartwo\esub\avartwo{\var\Id}
	\skesub\avar\Id\skesub\var\Id \\
	&\toss &\avartwo\esub\avartwo{\Id\Id}
	\skesub\avar\Id\skesub\var\Id \\
	&\red\toskdb &\avartwo\esub\avartwo{\avarthree\esub\avarthree\Id}
	\skesub\avar\Id\skesub\var\Id \\
	&\tosk\toss& \avartwo\esub\avartwo{\Id\skesub\avarthree\Id}
	\skesub\avar\Id\skesub\var\Id \\
	&\tosk\toss& \Id \skesub\avartwo\Id
	\skesub\avarthree\Id \skesub\avar\Id\skesub\var\Id
      \end{array}\]
      Then, setting $\sctx_0^{\sksym} \defeq \ctxhole\skesub\avartwo\Id
	\skesub\avarthree\Id \skesub\avar\Id\skesub\var\Id$, we obtain that $\tmthree_0 \toskneed^* \sctx_0^{\sksym}\ctxholep\Id$ using $3$ $\toskdb$-steps.

    \item \case{$n+1$}
     \[\begin{array}{lllll}
	\tmthree_{n+1}
	&= &
	(\var\Id(\var\Id))\esub\var{\tmtwo_{n+1}}
	\\
	&= &(\var\Id(\var\Id))
	\esub\var{(\la\vartwo\la\varthree{\vartwo \Id (\vartwo \Id) \varthree}) \tmtwo_n} 
	\\
	&\red\toskdb &(\var\Id(\var\Id))
	\esub\var{(\la\varthree{\vartwo \Id (\vartwo \Id) \varthree})
	\esub\vartwo{\tmtwo_n}} 
	\\
	&\tosk & (\var\Id(\var\Id))
	\skesub\var{\la\varthree\varfour\varthree}
	\esub\varfour{\vartwo \Id (\vartwo \Id)}
	\esub\vartwo{\tmtwo_n} 
	\\
	&\toss &((\la\varthree{\varfour\varthree})\Id)(\var\Id)
	\skesub\var{\la\varthree\varfour\varthree}
	\esub\varfour{\vartwo \Id (\vartwo \Id)}
	\esub\vartwo{\tmtwo_n} 
	\\
	&\red\toskdb &(\varfour\varthree)\esub\varthree\Id(\var\Id)
	\skesub\var{\la\varthree\varfour\varthree}
	\esub\varfour{\vartwo \Id (\vartwo \Id)}
	\esub\vartwo{\tmtwo_n}
	\\
	&\streq &(\varfour\varthree)\esub\varthree\Id(\var\Id)
	\skesub\var{\la\varthree\varfour\varthree}
	\esub\varfour{\vartwo \Id (\vartwo \Id)\esub\vartwo{\tmtwo_n}}
	
      \end{array}\]
      Please note the use of structural equivalence at the last step. Now, consider the context $\evctx \defeq
      (\varfour\varthree)\esub\varthree\Id(\var\Id)
      \skesub\var{\la\varthree\varfour\varthree}
      \esub\varfour\ctxhole$.
      By \ih,
      there is $\sctx_n^{\sksym}$ such that: 
      \[\begin{array}{ccc}\evctxfp{\vartwo \Id (\vartwo \Id)\esub\vartwo{\tmtwo_n}}
	& \toskneed^* &\evctxfp{\sctx_n^{\sksym}\ctxholep\Id}\end{array}\]
	 by doing $6n +3$ $\db$-step. The evaluation continue as follows:
     \[\begin{array}{lllll}
	\evctxfp{\sctx_n^{\sksym}\ctxholep\Id} 
	&= &(\varfour\varthree)\esub\varthree\Id(\var\Id)
	  \skesub\var{\la\varthree\varfour\varthree}
	\esub\varfour{\sctx_n^{\sksym}\ctxholep\Id} \\
	&\tosk &\sctx_n^{\sksym}\ctxholep{(\varfour\varthree)\esub\varthree\Id(\var\Id)
	  \skesub\var{\la\varthree\varfour\varthree}
	\skesub\varfour{\Id}} \\
	&\toss &\sctx_n^{\sksym}\ctxholep{(\Id\varthree)\esub\varthree\Id(\var\Id)
	  \skesub\var{\la\varthree\varfour\varthree}
	\skesub\varfour{\Id}} \\
	&\red\toskdb &\sctx_n^{\sksym}\ctxholep{\avar\esub\avar\varthree\esub\varthree\Id(\var\Id)
	  \skesub\var{\la\varthree\varfour\varthree}
	\skesub\varfour{\Id}} \\
	&\tosk\toss &\sctx_n^{\sksym}\ctxholep{\avar\esub\avar\Id\skesub\varthree\Id(\var\Id)
	  \skesub\var{\la\varthree\varfour\varthree}
	\skesub\varfour{\Id}} \\
	&\tosk\toss& \sctx_n^{\sksym}\ctxholep{\Id\skesub\avar\Id\skesub\varthree\Id(\var\Id)
	  \skesub\var{\la\varthree\varfour\varthree}
	\skesub\varfour{\Id}} \\
	&\red\toskdb &\sctx_n^{\sksym}\ctxholep{\avartwo\esub\avartwo{\var\Id}\skesub\avar\Id\skesub\varthree\Id
	  \skesub\var{\la\varthree\varfour\varthree}
	\skesub\varfour{\Id}} \\
	&\toss&
	\sctx_n^{\sksym}\ctxholep{
	    \avartwo\esub\avartwo{(\la\varthreep{\varfour\varthreep})\Id}
	  \skesub\avar\Id\skesub\varthree\Id
	  \skesub\var{\la\varthree\varfour\varthree}
	\skesub\varfour{\Id}} \\
	&\red\toskdb&
	\sctx_n^{\sksym}\ctxholep{
	    \avartwo\esub\avartwo
	    {(\varfour\varthreep)\esub\varthreep\Id}
	  \skesub\avar\Id\skesub\varthree\Id
	  \skesub\var{\la\varthree\varfour\varthree}
	\skesub\varfour{\Id}} \\
	&\toss&
	\sctx_n^{\sksym}\ctxholep{
	    \avartwo\esub\avartwo
	    {(\Id\varthreep)\esub\varthreep\Id}
	  \skesub\avar\Id\skesub\varthree\Id
	  \skesub\var{\la\varthree\varfour\varthree}
	\skesub\varfour{\Id}} \\
	&\red\toskdb&
	\sctx_n^{\sksym}\ctxholep{
	    \avartwo\esub\avartwo
	    {\avarthree\esub\avarthree\varthreep\esub\varthreep\Id}
	  \skesub\avar\Id\skesub\varthree\Id
	  \skesub\var{\la\varthree\varfour\varthree}
	\skesub\varfour{\Id}} \\
	&\tosk\toss&
	\sctx_n^{\sksym}\ctxholep{
	    \avartwo\esub\avartwo
	    {\avarthree\esub\avarthree\Id\skesub\varthreep\Id}
	  \skesub\avar\Id\skesub\varthree\Id
	  \skesub\var{\la\varthree\varfour\varthree}
	\skesub\varfour{\Id}} \\
	&\tosk\toss&
	\sctx_n^{\sksym}\ctxholep{
	    \avartwo\esub\avartwo
	    {\Id\skesub\avarthree\Id\skesub\varthreep\Id}
	  \skesub\avar\Id\skesub\varthree\Id
	  \skesub\var{\la\varthree\varfour\varthree}
	\skesub\varfour{\Id}} \\
	&\tosk\toss&
	\sctx_n^{\sksym}\ctxholep{
	    \Id\skesub\avartwo\Id
	    \skesub\avarthree\Id\skesub\varthreep\Id
	  \skesub\avar\Id\skesub\varthree\Id
	  \skesub\var{\la\varthree\varfour\varthree}
	\skesub\varfour{\Id}}
      \end{array}\]
      Then, let $\sctx_\bullet \defeq \sctx_n^{\sksym}\ctxholep{
	    \ctxhole\skesub\avartwo\Id
	    \skesub\avarthree\Id\skesub\varthreep\Id
	  \skesub\avar\Id\skesub\varthree\Id
	  \skesub\var{\la\varthree\varfour\varthree}
	\skesub\varfour{\Id}}$, so that we can write $\evctxfp{\sctx_n^{\sksym}\ctxholep\Id}  \toskneed^* \sctx_\bullet\ctxholep\Id$. Now, by putting all together, we obtain:
	\[\begin{array}{cccccccccc} 
	\tmthree_{n+1} &\toskneed^*\streq\evctxfp{\vartwo \Id (\vartwo \Id)\esub\vartwo{\tmtwo_n}}
	& \toskneed^* &\evctxfp{\sctx_n^{\sksym}\ctxholep\Id} &\toskneed^* &\sctx_\bullet\ctxholep\Id 
	\end{array}\]
	using $2 + (6n + 3) + 4 = 6(n+1) + 3$ $\sksym{\cdot}\db$-steps. By strong bisimulation of $\streq$ (\refprop{streq-is-a-strong-bisim}), we obtain, for some $\tmfour$, the following re-organization of the previous sequence, having in particular the same number and kind of steps:
		\[\begin{array}{cccccccccc} 
	\tmthree_{n+1} &\toskneed^*& \tmfour &\streq &\sctx_\bullet\ctxholep\Id 
	\end{array}\]
Note also that $\tmfour$ has shape $\sctx_\circ\ctxholep\Id$ for some substitution context $\sctx_\circ$, because $\streq$ cannot move anything in/out of abstractions, and that $\size{\sctx_\circ} = \size{\sctx_\bullet}$ because $\streq$ preserves the number of constructors. Therefor, by setting $\sctx_{n+1}^{\sksym} \defeq \sctx_\circ$, we obtain a sequence $\tmthree_{n+1} \toskneed^* \sctx_{n+1}^{\sksym}\ctxholep\Id$ satisfying the statement (and free from $\streq$).
  \end{itemize}
    The fact that $\size{\sctx_n^{\sksym}} = \bigo(n)$ is evident from the definition of $\sctx_\bullet$, which uses the $n$ case linearly.\qedhere
\end{proof}

\subparagraph{Evaluation in \cbneed.}
The family of contexts $\sctx^\bullet_n$ in the statement of the next lemma is defined in its proof.

\begin{lemma}
\label{l:need-family-aux}\NoteProof{l:need-family-aux}
    Let $n \in \nat$ and $\evctx$ be an evaluation context.
  Then there exist a substitution context $\sctx^\bullet_n$ and an evaluation
  $\deriv_n: \sctxvp{n}{\evctxfp{\var_n\Id(\var_n\Id)}} \toneed^*
  \sctxvp{n}{\evctxfp{\sctx^\bullet_n\ctxholep\Id}}$ such that
  $\sizedb{\deriv_n} = 8 \times 2^n - 5$ and $\size{\sctx^\bullet_n} = \Omega(2^n)$.
\end{lemma}
\begin{proof}
\applabel{l:need-family-aux}
  First of all, we recall the definition of $\sctx_n$ from the body of the paper: 
  \begin{itemize}
  \item $\sctxv{0} \defeq \ctxhole\esub{\var_0}\Id$, and
  \item $\sctxv{n+1} \defeq \sctxv{n}\ctxholep{\ctxhole\esub{\var_{n+1}}{\tmthree_{n+1}}}$
with $\tmthree_{n+1} \defeq \la{\vartwo_n}\var_n\Id(\var_n\Id)\vartwo_n$.
\end{itemize}
By induction on n.
  \begin{itemize}
    \item \case{$n = 0$}
     \[\begin{array}{lll}
	\sctxvp{0}{\evctxfp{\var_0\Id(\var_0\Id)}}
	&= &\evctxfp{\var_0\Id(\var_0\Id)}\esub{\var_0}\Id \\
	&\tolsneed &\evctxfp{\Id\Id(\var_0\Id)}\esub{\var_0}\Id \\
	&\red\todb &\evctxfp{\avar\esub\avar\Id(\var_0\Id)}\esub{\var_0}\Id \\
	&\tolsneed &\evctxfp{\Id\esub\avar\Id(\var_0\Id)}\esub{\var_0}\Id \\
	&\red\todb &\evctxfp{\avartwo\esub\avartwo{\var_0\Id}\esub\avar\Id}\esub{\var_0}\Id \\
	&\tolsneed &\evctxfp{\avartwo\esub\avartwo{\Id\Id}\esub\avar\Id}\esub{\var_0}\Id \\
	&\red\todb &\evctxfp{
	    \avartwo\esub\avartwo{\avarthree\esub\avarthree\Id}
	\esub\avar\Id}\esub{\var_0}\Id \\
	&\tolsneed
	\tolsneed &\evctxfp{
	    \Id\esub\avartwo\Id\esub\avarthree\Id
	\esub\avar\Id}\esub{\var_0}\Id
      \end{array}\]
      Then, setting $\sctx^\bullet_0 \defeq \ctxhole
      \esub\avartwo\Id \esub\avarthree\Id \esub\avar\Id$,
      we obtain that $\sctxvp{0}{\evctxfp{\var_0\Id(\var_0\Id)}}
      \toneed^* \sctxvp{0}{\evctxfp{\sctx^\bullet_0\ctxholep\Id}}$ using $3$ $\todb$-steps.

    \item \case{$n+1$}
      \[ \begin{array}{lll}
	\sctxvp{n+1}{\evctxfp{\var_{n+1}\Id(\var_{n+1}\Id)}}
	&=& \sctxvp{n}{\evctxfp{\var_{n+1}\Id(\var_{n+1}\Id)}
	\esub{\var_{n+1}}{\tmthree_{n+1}}}
	\\
	&\tolsneed& \sctxvp{n}{\evctxfp
	  {(\la{\vartwo_n}\var_n\Id(\var_n\Id)\vartwo_n)\Id(\var_{n+1}\Id)}
	\esub{\var_{n+1}}{\tmthree_{n+1}}}
	\\
	&\red\todb& \sctxvp{n}{\evctxfp
	  {(\var_n\Id(\var_n\Id)\vartwo_n)\esub{\vartwo_n}\Id(\var_{n+1}\Id)}
	\esub{\var_{n+1}}{\tmthree_{n+1}}}
      \end{array} \]
      We have a context
      $\evctxtwo \defeq
      \evctxfp{(\ctxhole \vartwo_n) \esub{\vartwo_n}\Id (\var_{n+1}\Id)}
      \esub{\var_{n+1}}{\tmthree_{n+1}}$.
      By \ih, there is $\sctx^\bullet_n$ such that
      $\sctxvp{n}{\evctxtwofp{\var_n\Id(\var_n\Id)}}
      \toneed^* \sctxvp{n}{\evctxtwofp{\sctx^\bullet_n\ctxholep\Id}}$
      by doing $8 \times 2^n - 5$ $\db$-steps.
      The evaluation continues as follows:
      \[ \begin{array}{lll}
	\sctxvp{n}{\evctxtwofp{\sctx^\bullet_n\ctxholep\Id}}
	&=& \sctxvp{n}{\evctxfp{
	  (\sctx^\bullet_n\ctxholep\Id \vartwo_n) \esub{\vartwo_n}\Id (\var_{n+1}\Id)}
	  \esub{\var_{n+1}}{\tmthree_{n+1}}
	} \\
	&\red\todb& \sctxvp{n}{\evctxfp{\sctx^\bullet_n\ctxholep{
	  \avar\esub\avar{\vartwo_n}} \esub{\vartwo_n}\Id (\var_{n+1}\Id)}
	  \esub{\var_{n+1}}{\tmthree_{n+1}}
	} \\
	&\tolsneed\tolsneed&
	\sctxvp{n}{\evctxfp{\sctx^\bullet_n\ctxholep{
	  \Id\esub\avar\Id \esub{\vartwo_n}\Id} (\var_{n+1}\Id)}
	  \esub{\var_{n+1}}{\tmthree_{n+1}}
	} \\
	&\red\todb& \sctxvp{n}{\evctxfp{\sctx^\bullet_n\ctxholep{
	      \avartwo\esub\avartwo{\var_{n+1}\Id}
	  \esub\avar\Id \esub{\vartwo_n}\Id}}
	  \esub{\var_{n+1}}{\tmthree_{n+1}}
	} \\
	&\tolsneed& \sctxvp{n}{\evctxfp{\sctx^\bullet_n\ctxholep{
	      \avartwo\esub\avartwo{
	      (\la{\vartwo_n}\var_n\Id(\var_n\Id)\vartwo_n)\Id}
	      \esub\avar\Id \esub{\vartwo_n}\Id
	  }}
	  \esub{\var_{n+1}}{\tmthree_{n+1}}
	} \\
	&\red\todb& \sctxvp{n}{\evctxfp{\sctx^\bullet_n\ctxholep{
	      \avartwo\esub\avartwo{
	      (\var_n\Id(\var_n\Id)\vartwo_n)\esub{\vartwo_n}\Id}
	  \esub\avar\Id \esub{\vartwo_n}\Id}}
	\esub{\var_{n+1}}{\tmthree_{n+1}} }
      \end{array} \]
      We have a context $\evctxthree =
      \evctxfp{\sctx^\bullet_n\ctxholep{
	  \avartwo\esub\avartwo{
	  (\ctxhole\vartwo_n)\esub{\vartwo_n}\Id}
      \esub\avar\Id \esub{\vartwo_n}\Id}}
      \esub{\var_{n+1}}{\tmthree_{n+1}}$.
      By \ih, there is $\sctx^\bullet_n$ such that
      $\sctxvp{n}{\evctxthreefp{\var_{n}\Id(\var_{n}\Id)}}
      \toneed^* \sctxvp{n}{\evctxthreefp{\sctx^\bullet_n\ctxholep\Id}}$
      by doing $8 \times 2^n - 5$ $\db$-steps.
      The evaluation continues as follows:
      \[ \begin{array}{lll}
	\sctxvp{n}{\evctxthreefp{\sctx^\bullet_n\ctxholep\Id}}
	&=& \sctxvp{n}{\evctxfp{\sctx^\bullet_n\ctxholep{
	      \avartwo\esub\avartwo{
	      (\sctx^\bullet_n\ctxholep\Id\vartwo_n)\esub{\vartwo_n}\Id}
	\esub\avar\Id \esub{\vartwo_n}\Id}}
	\esub{\var_{n+1}}{\tmthree_{n+1}} } \\
	&\red\todb& \sctxvp{n}{\evctxfp{\sctx^\bullet_n\ctxholep{
	      \avartwo\esub\avartwo{
	      \sctx^\bullet_n\ctxholep{\avarthree\esub\avarthree{\vartwo_n}}
	    \esub{\vartwo_n}\Id}
	\esub\avar\Id \esub{\vartwo_n}\Id}}
	\esub{\var_{n+1}}{\tmthree_{n+1}} } \\
	&\tolsneed\tolsneed& \sctxvp{n}{\evctxfp{\sctx^\bullet_n\ctxholep{
	      \avartwo\esub\avartwo{
	      \sctx^\bullet_n\ctxholep{\Id\esub\avarthree\Id}
	    \esub{\vartwo_n}\Id}
	\esub\avar\Id \esub{\vartwo_n}\Id}}
	\esub{\var_{n+1}}{\tmthree_{n+1}} } \\
	&\tolsneed& \sctxvp{n}{\evctxfp{\sctx^\bullet_n\ctxholep{
	      \sctx^\bullet_n\ctxholep{\Id\esub\avartwo\Id
	      \esub\avarthree\Id}
	      \esub{\vartwo_n}\Id
	\esub\avar\Id \esub{\vartwo_n}\Id}}
	\esub{\var_{n+1}}{\tmthree_{n+1}} }
      \end{array} \]
      Then, setting $\sctx^\bullet_{n+1} \defeq
      \sctx^\bullet_n\ctxholep{\sctx^\bullet_n\ctxholep{\ctxhole\esub\avartwo\Id
	\esub\avarthree\Id} \esub{\vartwo_n}\Id
      \esub\avar\Id \esub{\vartwo_n}\Id}$,
      we obtain that
      $\sctxvp{0}{\evctxfp{\var_0\Id(\var_0\Id)}}
      \toneed^* \sctxvp{n+1}{\evctxfp{\sctxthree^\bullet_n\ctxholep\Id}}$
      using $1 + (8 \times 2^n - 5) + 3 + (8 \times 2^n - 5) + 1
      = 8 \times 2^{n+1} - 5$ $\db$-steps.
  \end{itemize}
  The fact that $\size{\sctx_n^{\bullet}} = \Omega(2^n)$ is evident from the definition of $\sctx_n^{\bullet}$, which for $n+1$ uses twice the $n$ case.\qedhere
\end{proof}

\gettoappendix{l:need-family-aux2}
\begin{proof}
\applabel{l:need-family-aux2}
First of all, we recall the definition of $\sctx_n$  from the body of the paper: 
  \begin{itemize}
  \item $\sctxv{0} \defeq \ctxhole\esub{\var_0}\Id$, and
  \item $\sctxv{n+1} \defeq \sctxv{n}\ctxholep{\ctxhole\esub{\var_{n+1}}{\tmthree_{n+1}}}$
with $\tmthree_{n+1} \defeq \la{\vartwo_n}\var_n\Id(\var_n\Id)\vartwo_n$.
\end{itemize}
  By induction on n.
  \begin{itemize}
    \item \case{$n = 0$}
      \[ \begin{array}{lll}
	\evctxfp{\var_0\Id(\var_0\Id)} \esub{\var_0}{\tmtwo_0}
	&=& \evctxfp{\var_0\Id(\var_0\Id)} \esub{\var_0}{\Id} \\
	&\tolsneed& \evctxfp{\Id\Id(\var_0\Id)} \esub{\var_0}{\Id} \\
	&\red\todb& \evctxfp{\avar\esub\avar\Id(\var_0\Id)}
	\esub{\var_0}{\Id} \\
	&\tolsneed& \evctxfp{
	\Id\esub\avar\Id(\var_0\Id)} \esub{\var_0}{\Id} \\
	&\red\todb& \evctxfp{
	  \avartwo\esub\avartwo{\var_0\Id}\esub\avar\Id}
	\esub{\var_0}{\Id} \\
	&\tolsneed& \evctxfp{
	  \avartwo\esub\avartwo{\Id\Id}\esub\avar\Id}
	\esub{\var_0}{\Id} \\
	&\red\todb& \evctxfp{
	      \avartwo\esub\avartwo{\avarthree\esub\avarthree\Id}
	  \esub\avar\Id}
	\esub{\var_0}{\Id} \\
	&\tolsneed\tolsneed& \evctxfp{
	      \Id\esub\avartwo\Id\esub\avarthree\Id
	  \esub\avar\Id}
	\esub{\var_0}{\Id}
      \end{array} \]
      Then, setting $\sctx^{\ndsym}_0 \defeq
      \ctxhole \esub\avartwo\Id\esub\avarthree\Id \esub\avar\Id$,
      we obtain that
      $\evctxfp{\var_0\Id(\var_0\Id)} \esub{\var_0}{\tmtwo_0}
      \toneed^* \sctxvp{0}{\evctxfp{\sctx^{\ndsym}_0\ctxholep\Id}}$
      using $3$ $\db$-steps.

    \item \case{$n+1$}
      \[ \begin{array}{lll}
	\evctxfp{\var_{n+1}\Id(\var_{n+1}\Id)}
	\esub{\var_{n+1}}{\tmtwo_{n+1}}
	&=& \evctxfp{\var_{n+1}\Id(\var_{n+1}\Id)}
	  \esub{\var_{n+1}}{(\la{\var_n}\la{\vartwo_n}{({\var_n} \Id) ({\var_n}
	\Id) {\vartwo_n}})\tmtwo_n} \\
	&\red\todb& \evctxfp{\var_{n+1}\Id(\var_{n+1}\Id)}
	  \esub{\var_{n+1}}{
	    (\la{\vartwo_n}{{\var_n}\Id({\var_n}\Id) 
	{\vartwo_n}})\esub{\var_n}{\tmtwo_n}} \\
	&\tolsneed& \evctxfp{
	      (\la{\vartwo_n}{{\var_n}\Id({\var_n}\Id) {\vartwo_n}})\Id
	  (\var_{n+1}\Id)}
	  \esub{\var_{n+1}}{\la{\vartwo_n}{{\var_n}\Id({\var_n}\Id) {\vartwo_n}}}
	\esub{\var_n}{\tmtwo_n} \\
	&\red\todb& \evctxfp{
	      ({\var_n}\Id({\var_n}\Id) {\vartwo_n})\esub{\vartwo_n}\Id
	  (\var_{n+1}\Id)}
	  \esub{\var_{n+1}}{\la{\vartwo_n}{{\var_n}\Id({\var_n}\Id) {\vartwo_n}}}
	\esub{\var_n}{\tmtwo_n}
      \end{array} \]
      Consider the context $\evctxtwo \defeq
      \evctxfp{(\ctxhole {\vartwo_n})\esub{\vartwo_n}\Id (\var_{n+1}\Id)}
      \esub{\var_{n+1}}{\la{\vartwo_n}{{\var_n}\Id({\var_n}\Id) {\vartwo_n}}}$.
      By \ih, there is $\sctx^{\ndsym}_n$ such that
      $\evctxtwofp{{\var_n}\Id({\var_n}\Id)} \esub{\var_n}{\tmtwo_n}
      \toneed^* \sctxvp{n}{\evctxtwofp{\sctx^{\ndsym}_n\ctxholep\Id}}$
      by doing $8 \times 2^n - 5$ $\db$-steps.
      The evaluation continues as follows:
      \[ \begin{array}{lll}
	\sctxvp{n}{\evctxtwofp{\sctx^{\ndsym}_n\ctxholep\Id}}
	&=& \sctxvp{n}{
	  \evctxfp{(\sctx^{\ndsym}_n\ctxholep\Id {\vartwo_n})\esub{\vartwo_n}\Id (\var_{n+1}\Id)}
	\esub{\var_{n+1}}{\la{\vartwo_n}{{\var_n}\Id({\var_n}\Id) {\vartwo_n}}}} \\
	&=& \sctxvp{n}{
	  \evctxfp{(\sctx^{\ndsym}_n\ctxholep\Id {\vartwo_n})\esub{\vartwo_n}\Id (\var_{n+1}\Id)}
	\esub{\var_{n+1}}{\tmthree_{n+1}} } \\
	&\red\todb& \sctxvp{n}{
	  \evctxfp{\sctx^{\ndsym}_n\ctxholep{\avar\esub\avar{\vartwo_n}}
	\esub{\vartwo_n}\Id (\var_{n+1}\Id)}
	\esub{\var_{n+1}}{\tmthree_{n+1}} } \\
	&\tolsneed\tolsneed& \sctxvp{n}{
	  \evctxfp{\sctx^{\ndsym}_n\ctxholep{\Id\esub\avar\Id}
	\esub{\vartwo_n}\Id (\var_{n+1}\Id)}
	\esub{\var_{n+1}}{\tmthree_{n+1}} } \\
	&\red\todb& \sctxvp{n}{
	  \evctxfp{\sctx^{\ndsym}_n\ctxholep{\avartwo\esub\avartwo{\var_{n+1}\Id}\esub\avar\Id}
	\esub{\vartwo_n}\Id}
	\esub{\var_{n+1}}{\tmthree_{n+1}} } \\
	&\tolsneed& \sctxvp{n}{
	  \evctxfp{\sctx^{\ndsym}_n\ctxholep{\avartwo\esub\avartwo
	    {(\la{\vartwo_n'}{{\var_n}\Id({\var_n}\Id) {\vartwo_n'}})\Id}
	  \esub\avar\Id}
	\esub{\vartwo_n}\Id}
	\esub{\var_{n+1}}{\tmthree_{n+1}} } \\
	&\red\todb& \sctxvp{n}{
	  \evctxfp{\sctx^{\ndsym}_n\ctxholep{\avartwo\esub\avartwo
	    {({\var_n}\Id({\var_n}\Id)\vartwo_n')\esub{\vartwo_n'}\Id}
	  \esub\avar\Id}
	\esub{\vartwo_n}\Id}
	\esub{\var_{n+1}}{\tmthree_{n+1}}}
      \end{array} \]
      Now, consider the context $\evctxthree \defeq
      \evctxfp{\sctx^{\ndsym}_n\ctxholep{\avartwo\esub\avartwo
	  {(\ctxhole\vartwo_n')\esub{\vartwo_n'}\Id}
	\esub\avar\Id}
      \esub{\vartwo_n}\Id}
      \esub{\var_{n+1}}{\tmthree_{n+1}}$.
      By \reflemma{need-family-aux}, there is $\sctx^\bullet_n$ such that
      $\sctxvp{n}{\evctxthreefp{\var_n\Id(\var_n\Id)}}
      \toneed^* \sctxvp{n}{\evctxthreefp{\sctx^\bullet_n\ctxholep\Id}}$
      by doing $8 \times 2^n -5$ $\db$-steps.
      The evaluation continues as follows:
      \[ \begin{array}{lll}
	\sctxvp{n}{\evctxthreefp{\sctx^\bullet_n\ctxholep\Id}}
	&=& \sctxvp{n}{
	  \evctxfp{\sctx^{\ndsym}_n\ctxholep{\avartwo\esub\avartwo
	      {(\sctx^\bullet_n\ctxholep\Id\vartwo_n')\esub{\vartwo_n'}\Id}
	    \esub\avar\Id}
	  \esub{\vartwo_n}\Id}
	\esub{\var_{n+1}}{\tmthree_{n+1}}} \\
	&=& \sctxvp{n+1}{
	  \evctxfp{\sctx^{\ndsym}_n\ctxholep{\avartwo\esub\avartwo
	      {(\sctx^\bullet_n\ctxholep\Id\vartwo_n')\esub{\vartwo_n'}\Id}
	    \esub\avar\Id}
	  \esub{\vartwo_n}\Id}
	} \\
	&\red\todb& \sctxvp{n+1}{
	  \evctxfp{\sctx^{\ndsym}_n\ctxholep{\avartwo\esub\avartwo
	      {\sctx^\bullet_n\ctxholep{\avarthree\esub\avarthree{\vartwo_n'}}\esub{\vartwo_n'}\Id}
	    \esub\avar\Id}
	  \esub{\vartwo_n}\Id}
	} \\
	&\tolsneed\tolsneed& \sctxvp{n+1}{
	  \evctxfp{\sctx^{\ndsym}_n\ctxholep{\avartwo\esub\avartwo
	      {\sctx^\bullet_n\ctxholep{\Id\esub\avarthree\Id}\esub{\vartwo_n'}\Id}
	    \esub\avar\Id}
	  \esub{\vartwo_n}\Id}
	} \\
	&\tolsneed& \sctxvp{n+1}{
	  \evctxfp{\sctx^{\ndsym}_n\ctxholep{\sctx^\bullet_n\ctxholep{
		\Id \esub\avartwo\Id
	      \esub\avarthree\Id}
	    \esub{\vartwo_n'}\Id
	  \esub\avar\Id}
	  \esub{\vartwo_n}\Id}
	}
	\end{array} \]
      Then, setting
      $\sctx^{\ndsym}_{n+1} \defeq \sctx^{\ndsym}_n\ctxholep{\sctx^\bullet_n\ctxholep {\ctxhole \esub\avartwo\Id \esub\avarthree\Id}
      \esub{\vartwo_n'}\Id \esub\avar\Id}
      \esub{\vartwo_n}\Id$,
      we obtain that:
      \[\begin{array}{ccc}
      \evctxfp{\var_{n+1}\Id(\var_{n+1}\Id)} \esub{\var_{n+1}}{\tmtwo_{n+1}}
      &\toneed^* &\sctxvp{n+1}{\evctxfp{\sctx^{\ndsym}_n\ctxholep\Id}}
      \end{array}\]
      using
      $1 + (8 \times 2^n - 5) + 3 + (8 \times 2^n - 5) + 1
      = 8 \times 2^{n+1} - 10 + 5 +n + 1
      = 8 \times 2^{n+1} + (n + 1) - 5$ $\db$-steps.
  \end{itemize}
    The fact that $\size{\sctx_n^{\ndsym}} = \Omega(2^n)$ follows from the fact that $\size{\sctx_n^{\bullet}} = \Omega(2^n)$ by \reflemma{need-family-aux}.\qedhere
\end{proof}
\section{Proofs about Marked Terms}

\gettoappendix{l:unmarked-downward-closed}
\begin{proof}
\applabel{l:unmarked-downward-closed}
Straightforward induction on $\mctx$.
\end{proof}

\gettoappendix{prop:marked-skeleton-soundness}

\begin{proof}
\applabel{prop:marked-skeleton-soundness}
\hfill
\begin{enumerate}
\item Straightforward induction on $\tm$. We give the details nonetheless. Cases:
\begin{itemize}
\item \case{$\fv{\tm} \cap \skelvars = \emptyset$ and $\tm$ is not a variable}
Then by definition, $\skeldec\tm\skelvars = (\varthree,\ctxhole\esub\varthree\tm)$.
While, $\disc\tm\skelvars = \tm$ and
$\splitfun{\disc\tm\skelvars} = (\varthree,\ctxhole\esub\varthree\tm)$.
\item \case{$\tm = \var$}
Then $\disc\tm\skelvars = \var$ or $\disc\tm\skelvars = \var^\taken$. In both cases, by definition, $\splitfun{\disc\tm\skelvars} = (\var,\ctxhole)$. By definition, $\skeldec\var\skelvars = (\var,\ctxhole)$.
\item \case{$\tm = \la\var\tmtwo$ and $\fv\tm \cap \skelvars \neq \emptyset$}
By definition $\disc\tm\skelvars = \lat\var{\disc\tmtwo{\skelvars \cup \{\var\}}}$ and $\splitfun{\disc\tm\skelvars} = (\la\var\tmthree,\sctx)$ where $\splitfun{\disc\tmtwo{\skelvars \cup \{\var\}}} = (\tmthree,\sctx)$.
Moreover, by definition $\skeldec\tm\skelvars = (\la\var\tmfour,\sctxtwo)$ where $\skeldec\tmtwo{\skelvars \cup \{\var\}}=(\tmfour,\sctxtwo)$.
By \ih, $(\tmthree,\sctx) =
(\tmfour,\sctxtwo)$, so that $(\la\var\tmthree,\sctx)= (\la\var\tmfour,\sctxtwo)$.
\item \case{$\tm = \tmtwo\ap\tmthree$ and $\fv\tm \cap \skelvars \neq \emptyset$}
By definition:
\[\splitfun{\disc\tm\skelvars} = \splitfun{\disc\tmtwo\skelvars\appt\disc\tmthree\skelvars} = (\tmtwo'\ap\tmthree',\sctxptwo\sctx)\] where
$\splitfun{\disc\tmtwo\skelvars} = (\tmtwo',\sctx)$ and $\splitfun{\disc\tmthree\skelvars} = (\tmthree', \sctxtwo)$.
By \ih, we have $\skeldec\tmtwo\skelvars= (\tmtwo',\sctx)$ and $\skeldec\tmthree\skelvars = (\tmthree', \sctxtwo)$.
Therefore, $\skeldec\tm\skelvars = (\tmtwo'\ap\tmthree',\sctxptwo\sctx) = \splitfun{\disc\tm\skelvars}$.
\end{itemize}

\item By definition, $\skeldecabs{\la\var\tm} = (\la\var\tmtwo, \sctx)$ where $\skeldec\tm{\set\var}=(\tmtwo,\sctx)$. By Point 1, we have $\splitfun{\disc\tm{\set\var}}=(\tmtwo,\sctx)$. Thus, $\splitfun{\discabs{\la\var\tm}} = (\la\var\tmtwo,\sctx)=\skeldecabs{\la\var\tm}$.\qedhere
\end{enumerate}
\end{proof}

\subparagraph{Soundness of the Algorithm.} First of all, we need an auxiliary lemma.

\begin{lemma}
\label{l:disc-useless-var}
Let $\tm$ be a term, $\skelvars$ a set of variables, and $\vartwo \notin \fv{\tm}$. Then $\disc\tm\skelvars = \disc\tm{\skelvars \setminus \{\vartwo\}}$.
\end{lemma}
\begin{proof}
By induction on $\tm$.
\begin{itemize}
\item \case{$\tm$ is such that $\fv{\tm} \cap \skelvars = \emptyset$} Then, we also have $\fv{\tm} \cap (\skelvars \setminus \{\vartwo\}) = \emptyset$, so by definition $\disc\tm\skelvars = \tm = \disc\tm{\skelvars \setminus \{\vartwo\}}$.
\item \case{$\tm = \var$ and $\var \in \skelvars$} By assumption $\var \neq \vartwo$, hence by definition:
$\disc\tm\skelvars = \var^\taken = \disc\tm{\skelvars \setminus \{\vartwo\}}$.
\item \case{$\tm = \la\var\tmtwo$ and $\fv{\tm} \cap \skelvars \neq \emptyset$} Without loss of generality we can assume $\var \neq \vartwo$. Thus, since by assumption $\vartwo \notin \fv{\tm}$, we also have that $\vartwo \notin \fv{\tmtwo}$. Hence by \ih we have:
\begin{align*}
\disc\tm\skelvars &= \lat\var\disc\tmtwo{\skelvars \cup \{\var\}} \\
&=\lat\var\disc\tmtwo{(\skelvars \cup \{\var\}) \setminus \{\vartwo\}} \\
&=\lat\var\disc\tmtwo{(\skelvars \setminus \{\vartwo\}) \cup \{\var\}} \\
&=\disc{\la\var\tmtwo}{\skelvars \setminus \{\vartwo\}} \\
&=\disc{\tm}{\skelvars \setminus \{\vartwo\}}.
\end{align*}
\item \case{$\tm = \tmtwo\ap\tmthree$ and $\fv{\tm} \cap \skelvars \neq \emptyset$}
By assumption $\vartwo \notin \fv{\tm}$, therefore, $\vartwo \notin \fv{\tmtwo}$ and $\vartwo \notin \fv{\tmthree}$. Furthermore, $\fv{\tm} \cap (\skelvars \setminus \{\vartwo\}) \neq \emptyset$, so using the \ih:
\begin{align*}
\disc\tm\skelvars &= \disc\tmtwo\skelvars\appt\disc\tmthree\skelvars \\
&=\disc\tmtwo{\skelvars \setminus \{\vartwo\}}\appt\disc\tmthree\skelvars \\
&=\disc\tmtwo{\skelvars \setminus \{\vartwo\}}\appt\disc\tmthree{\skelvars \setminus \{\vartwo\}} \\
&=\disc\tm{\skelvars \setminus \{\vartwo\}}.\qedhere
\end{align*}
\end{itemize}
\end{proof}

\gettoappendix{l:disc-comp-step}
\begin{proof}
\applabel{l:disc-comp-step} By induction on the structure of $\tm$.
\begin{itemize}
\item \case{$\tm = \var$} By definition, $\mtm=\var$ and $\mtmtwo = \var^\taken$. Note that $\mtm\isub\var{\skelu{\var^\taken}} = \var\isub\var{\skelu{\var^\taken}} = \skelu{\var^\taken}$, as required.
Moreover, $k = 0$ and $\sizet{\mtmtwo} - \sizet{\mtm}-1=1+0-1=0$, as required.
\item \case{$\tm = \tm_l\ap\tm_r$ and $\fv{\tm} \cap \skelvars = \emptyset$} Now, observe that we also have that $\fv{\tm_l} \cap \skelvars = \emptyset$ and $\fv{\tm_r} \cap \skelvars = \emptyset$, therefore $\disc{\tm_l}\skelvars = \tm_l$ and $\disc{\tm_r}\skelvars=\tm_r$. Hence we can write $\mtm = \tm_l \tm_r$.
Concerning $\mtmtwo$, since $\var \in \fv{\tm}$, we have:
\[ \mtmtwo = \disc{\tm_l}{\skelvars\cup\{\var\}}\appt\disc{\tm_r}{\skelvars\cup\{\var\}}.\]
Let $\mtmtwo_l = \disc{\tm_l}{\skelvars \cup \{\var\}}$ and $\mtmtwo_r = \disc{\tm_r}{\skelvars \cup \{\var\}}$. We distinguish three cases:
\begin{itemize}
\item \case{$\var \in \fv{\tm_l} \cap \fv{\tm_r}$} We have that:
\[\begin{array}{cllllll}
\mtm\isub\var{\skelu{\var^\taken}} &= &
  \tm_l\isub\var{\skelu{\var^\taken}} \tm_r\isub\var{\skelu{\var^\taken}} 
  \\
  \mbox{(by \ih)} &\tomark^{k_1} &
  \skelu{\mtmtwo_l} \tm_r\isub\var{\skelu{\var^\taken}} 
  \\
  \mbox{(by \ih)} &\tomark^{k_2} &
  \skelu{\mtmtwo_l} \skelu{\mtmtwo_r} 
  \\
  &\tomarkblone &
  \skelu{(\mtmtwo_l \appt\skelu{\mtmtwo_r})} 
  \\
  &\tomarkwhtwo &
  \skelu{(\mtmtwo_l\appt\mtmtwo_r)} 
  &=& \skelu{\mtmtwo}
\end{array}\]
where $k_1 =_{\ih} \sizet{\mtmtwo_l} - \sizet{\tm_l} -1= \sizet{\mtmtwo_l} -1$ and $k_2 =_{\ih} \sizet{\mtmtwo_r} - \sizet{\tm_r}-1 = \sizet{\mtmtwo_r}-1$ because $\tm$ is unmarked so that $\sizet{\tmtwo_l} = \sizet{\tmtwo_r} = 0$.
Moreover:
\[\begin{array}{lllllll}
k & = &  k_1 + k_2 +2
\\
& =_{\ih} & \sizet{\mtmtwo_l} + \sizet{\mtmtwo_r} -2 +2
\\
& = & \sizet{\mtmtwo_l} + \sizet{\mtmtwo_r}
\\
& = & \sizet{\mtmtwo} -1

& = & \sizet{\mtmtwo} - \sizet{\mtm} -1
\end{array}\]
Then the statement holds, since $\mtm$ starts unmarked.

\item \case{$\var \in \fv{\tm_l}$ and $\var \notin \fv{\tm_r}$}
Observe that $\tm_r\isub\var{\skelu{\var^\taken}} = \tm_r$, so we have that:
\[\begin{array}{cllllll}
\mtm\isub\var{\skelu{\var^\taken}} &= &
  \tm_l\isub\var{\skelu{\var^\taken}} \tm_r\isub\var{\skelu{\var^\taken}} 
  \\
  &= &
  \tm_l\isub\var{\skelu{\var^\taken}} \tm_r
  \\
  \mbox{(by \ih)} &\tomark^{k_1} & 
  \skelu{\mtmtwo_l}  \tm_r 
  \\
  &\tomark&
   \skelu{(\mtmtwo_l \appt \tm_r)} 
   \\
  &=&
  \skelu{(\mtmtwo_l\appt\disc{\tm_r}\skelvars)} 
  \\
  &=_{\reflemmaeq{disc-useless-var}}&
  \skelu{(\mtmtwo_l\appt\disc{\tm_r}{\skelvars \cup \{\var\}})} 
  \\
  &=&\skelu{(\mtmtwo_l\appt\mtmtwo_r)} 
   &= &
   \skelu{\mtmtwo}
\end{array}\]
where $k_1 =_{\ih} \sizet{\mtmtwo_l} - \sizet{\tm_l} -1 = \sizet{\mtmtwo_l} -1$. Moreover:
\[\begin{array}{lllllll}
k & = & k_1 + 1
\\
& =_{\ih} & \sizet{\mtmtwo_l} -1 +1
\\
& = & \sizet{\mtmtwo_l}
\\
& = & \sizet{\mtmtwo_l} + \sizet{\mtmtwo_r}
\\
& = & \sizet{\mtmtwo} -1

& = & \sizet{\mtmtwo} - \sizet{\mtm} - 1
\end{array}\]
Then the statement holds, since $\mtm$ starts unmarked.

\item \case{$\var \notin \fv{\tm_l}$ and $\var \in \fv{\tm_r}$} As in the previous case, simply switching the roles of $\tm_l$ and $\tm_r$.
\end{itemize}

\item \case{$\tm = {\tm_l}\ap{\tm_r}$ and $\fv{\tm} \cap \skelvars \neq \emptyset$} By definition we have that $\mtm = \disc{\tm_l}\skelvars \appt \disc{\tm_r}\skelvars$, that is, $\mtm$ starts with $\taken$. Since $\fv{\tm} \cap \skelvars \cup \{\var\} \neq \emptyset$, we also have that $\mtmtwo = \disc{\tm_l}{\skelvars\cup\{\var\}}\appt\disc{\tm_r}{\skelvars\cup\{\var\}}$. Now, let $\mtm_l = \disc{\tm_l}\skelvars$, $\mtm_r = \disc{\tm_r}\skelvars$, $\mtmtwo_l = \disc{\tm_l}{\skelvars\cup\{\var\}}$, $\mtmtwo_r = \disc{\tm_r}{\skelvars\cup\{\var\}}$. There are three cases.
\begin{itemize}
\item \case{both $\fv{\tm_l}$ and $\fv{\tm_r}$ have non-empty intersection with $\skelvars$} Then both $\mtm_l$ and $\mtm_r$ start with $\taken$. Sub-cases:
\begin{itemize}
\item $\var \in \fv{\tm_l} \cap \fv{\tm_r}$. 
Then:
\[\begin{array}{cllllll}
\mtm\isub\var{\skelu{\var^\taken}}
  &= & 
  \mtm_l\isub\var{\skelu{\var^\taken}}\appt\mtm_r\isub\var{\skelu{\var^\taken}} 
  \\
   \mbox{(by \ih)}  &\tomark^{k_1} &
\mtmtwo_l\appt \mtm_r\isub\var{\skelu{\var^\taken}} 
  \\
    \mbox{(by \ih)}  &\tomark^{k_2} &
  \mtmtwo_l \appt \mtmtwo_r
  &=&
  \mtmtwo
\end{array}\]
where $k_1 =_{\ih} \sizet{\mtmtwo_l} - \sizet{\mtm_l}$ and $k_2 =_{\ih} \sizet{\mtmtwo_r} - \sizet{\mtm_r}$. 
Moreover:
\[\begin{array}{lllllll}
k & = & k_1 + k_2
\\
& =_{\ih} & 
\sizet{\mtmtwo_l} - \sizet{\mtm_l} + \sizet{\mtmtwo_r} - \sizet{\mtm_r}
\\
& = & 
\sizet{\mtmtwo_l} +\sizet{\mtmtwo_r} - \sizet{\mtm_l} - \sizet{\mtm_r}
\\
& = & 
(\sizet{\mtmtwo_l} +\sizet{\mtmtwo_r}+1) - (\sizet{\mtm_l} + \sizet{\mtm_r}+1)
\\
& = & 
\sizet{\mtmtwo} - \sizet{\mtm}
\end{array}\]
Then the statement holds, since $\mtm$ starts with $\taken$.

\item \emph{$\var \in \fv{\tm_l}$ and $\var\notin\fv{\tm_r}$}. By \reflemma{disc-useless-var}, $\var \notin \fv{\tm_r}$ implies $\mtmtwo_r=\disc{\tm_r}{\skelvars \cup \{\var\}} = \disc{\tm_r}{\skelvars} = \mtm_r$.
Then:
\[\begin{array}{cllllll}
\mtm\isub\var{\skelu{\var^\taken}}
  &= & 
  \mtm_l\isub\var{\skelu{\var^\taken}}\appt\mtm_r\isub\var{\skelu{\var^\taken}} 
  \\
   \mbox{(by \ih)}  &\tomark^{k} &
\mtmtwo_l\appt \mtm_r\isub\var{\skelu{\var^\taken}} 
  \\
 &= &
  \mtmtwo_l \appt \mtm_r
  \\
 &=_{\reflemmaeq{disc-useless-var}} &
  \mtmtwo_l \appt \mtmtwo_r
  &=&
  \mtmtwo
\end{array}\]
where $k =_{\ih} \sizet{\mtmtwo_l} - \sizet{\mtm_l}$. Clearly, from $\mtmtwo_r= \mtm_r$ it follows $\sizet{\mtmtwo_r}= \sizet{\mtm_r}$
Moreover:
\[\begin{array}{lllllll}
k & =_{\ih} & 
\sizet{\mtmtwo_l} - \sizet{\mtm_l} 
\\
& = & 
\sizet{\mtmtwo_l} +\sizet{\mtmtwo_r} - \sizet{\mtm_l} - \sizet{\mtm_r}
\\
& = & 
(\sizet{\mtmtwo_l} +\sizet{\mtmtwo_r}+1) - (\sizet{\mtm_l} + \sizet{\mtm_r}+1)
\\
& = & 
\sizet{\mtmtwo} - \sizet{\mtm}
\end{array}\]
Then the statement holds, since $\mtm$ starts with $\taken$.

\item \emph{$\var \notin \fv{\tm_l}$ and $\var\in\fv{\tm_r}$}. Specular to the previous one.
\end{itemize}

\item \case{only $\fv{\tm_l}$ has non-empty intersection with $\skelvars$}
Therefore, $\mtm_l$ starts with $\taken$ and $\mtm_r = \tm_r$ starts unmarked. 
\begin{itemize}
\item $\var \in \fv{\tm_l} \cap \fv{\tm_r}$. Then:
\[\begin{array}{cllllll}
\mtm\isub\var{\skelu{\var^\taken}}
  &= &
  \mtm_l\isub\var{\skelu{\var^\taken}}\appt\tm_r\isub\var{\skelu{\var^\taken}} 
  \\
    \mbox{(by \ih)} &\tomark^{k_1} &
  \mtmtwo_l\appt \tm_r\isub\var{\skelu{\var^\taken}} 
  \\
    \mbox{(by \ih)} &\tomark^{k_2} &
    \mtmtwo_l\appt\skelu{\mtmtwo_r} 
    \\
  &\tomarkwhtwo &
  \mtmtwo_l\appt\mtmtwo_r
  &= &
  \mtmtwo
\end{array}\]
where $k_1 =_{\ih} \sizet{\mtmtwo_l} - \sizet{\mtm_l}$ and $k_2 =_{\ih} \sizet{\mtmtwo_r} - \sizet{\tm_r} -1 = \sizet{\mtmtwo_r}-1$.
Moreover:
\[\begin{array}{cllllll}
k & = & k_1 + k_2 +1
\\
& =_{\ih} & 
\sizet{\mtmtwo_l} - \sizet{\mtm_l} + \sizet{\mtmtwo_r} -1 +1
\\
& = & 
\sizet{\mtmtwo_l} - \sizet{\mtm_l} + \sizet{\mtmtwo_r}
\\
& = & 
(\sizet{\mtmtwo_l} + \sizet{\mtmtwo_r} +1) - (\sizet{\mtm_l} +1)
& = & 
\sizet{\mtmtwo}-\sizet{\mtm}
\end{array}\]
Then the statement holds, since $\mtm$ starts with $\taken$.

\item \emph{$\var \in \fv{\tm_l}$ and $\var\notin\fv{\tm_r}$}. By \reflemma{disc-useless-var}, $\var \notin \fv{\tm_r}$ implies $\mtmtwo_r=\disc{\tm_r}{\skelvars \cup \{\var\}} = \disc{\tm_r}{\skelvars} = \tm_r$. Then:
\[\begin{array}{cllllll}
\mtm\isub\var{\skelu{\var^\taken}}
  &= & 
  \mtm_l\isub\var{\skelu{\var^\taken}}\appt\tm_r\isub\var{\skelu{\var^\taken}} 
  \\
   \mbox{(by \ih)}  &\tomark^{k} &
\mtmtwo_l\appt \tm_r\isub\var{\skelu{\var^\taken}} 
  \\
 &= &
  \mtmtwo_l \appt \tm_r
  \\
 &=_{\reflemmaeq{disc-useless-var}} &
  \mtmtwo_l \appt \mtmtwo_r
  &=&
  \mtmtwo
\end{array}\]
where $k =_{\ih} \sizet{\mtmtwo_l} - \sizet{\mtm_l}$.
 Clearly, from $\mtmtwo_r= \tm_r$ it follows $\sizet{\mtmtwo_r}= \sizet{\tm_r}=0$
Moreover:
\[\begin{array}{lllllll}
k & =_{\ih} & 
\sizet{\mtmtwo_l} - \sizet{\mtm_l} 
\\
& = & 
\sizet{\mtmtwo_l} +\sizet{\mtmtwo_r} - \sizet{\mtm_l} - \sizet{\tm_r}
\\
& = & 
(\sizet{\mtmtwo_l} +\sizet{\mtmtwo_r}+1) - (\sizet{\mtm_l} + \sizet{\tm_r}+1)
\\
& = & 
\sizet{\mtmtwo} - \sizet{\mtm}
\end{array}\]
Then the statement holds, since $\mtm$ starts with $\taken$.

\item \emph{$\var \notin \fv{\tm_l}$ and $\var\in\fv{\tm_r}$}. By \reflemma{disc-useless-var}, $\var \notin \fv{\tm_l}$ implies $\mtmtwo_l=\disc{\tm_l}{\skelvars \cup \{\var\}} = \disc{\tm_l}{\skelvars} = \mtm_l$. Then:
\[\begin{array}{cllllll}
\mtm\isub\var{\skelu{\var^\taken}}
  &= & 
  \mtm_l\isub\var{\skelu{\var^\taken}}\appt\tm_r\isub\var{\skelu{\var^\taken}} 
\\
  &= & 
  \mtm_l\appt\tm_r\isub\var{\skelu{\var^\taken}} 
  \\
 &=_{\reflemmaeq{disc-useless-var}} &
  \mtmtwo_l\appt\tm_r\isub\var{\skelu{\var^\taken}} 
    \\
   \mbox{(by \ih)}  &\tomark^{k_r} &
\mtmtwo_l \appt \skelu{\mtmtwo_r}
  \\
    &\tomarkwhtwo &
\mtmtwo_l\appt \mtmtwo_r
  &=&
  \mtmtwo
\end{array}\]
where $k_r =_{\ih} \sizet{\mtmtwo_r} - \sizet{\tm_r} -1$.
 Clearly, from $\mtmtwo_l= \mtm_l$ it follows $\sizet{\mtmtwo_l}= \sizet{\mtm_l}$.
Moreover:
\[\begin{array}{lllllll}
k &= & k_r +1
\\
& =_{\ih} & 
\sizet{\mtmtwo_r} - \sizet{\tm_r} -1 +1
\\
& = & 
\sizet{\mtmtwo_r} - \sizet{\tm_r}
\\
& = & 
\sizet{\mtmtwo_l} +\sizet{\mtmtwo_r} - \sizet{\mtm_l} - \sizet{\tm_r}
\\
& = & 
(\sizet{\mtmtwo_l} +\sizet{\mtmtwo_r}+1) - (\sizet{\mtm_l} + \sizet{\tm_r}+1)
\\
& = & 
\sizet{\mtmtwo} - \sizet{\mtm}
\end{array}\]
Then the statement holds, since $\mtm$ starts with $\taken$.
\end{itemize}

\item \case{only $\fv{\tm_r}$ has non-empty intersection with $\skelvars$} As in the previous case, simply switching the roles of $\tm_l$ and $\tm_r$.
\end{itemize}

\item \case{$\tm = \la\vartwo\tm_b$ and $\fv{\tm} \cap \skelvars = \emptyset$} Without loss of generality we can assume $\vartwo \notin \skelvars \cup \{\var\}$. By definition,
$\mtm = \la\vartwo\tm_b$. Now, since $\fv{\tm} \cap \skelvars = \emptyset$ and $\vartwo \notin \skelvars$, we also have that $\fv{\tm_b} \cap \skelvars = \emptyset$. Therefore, by definition we have that $\tm_b = \disc{\tm_b}\skelvars$ and so we can write
$\mtm = \la\vartwo\disc{\tm_b}\skelvars$. Let $\mtm_b \defeq \disc{\tm_b}\skelvars$. Observe that $\mtm_b$ starts unmarked and that $\var\in\fv\tm$ implies $\var\in\fv{\tm_b}$, so we have:
\[\begin{array}{cllllll}
\mtm\isub\var{\skelu{\var^\taken}}
&= & 
 \la\vartwo(\mtm_b\isub\var{\skelu{\var^\taken}}) 
\\
\mbox{(by \ih)} &\tomark^{k_1} & 
 \la\vartwo(\skelu{\disc{\tm_b}{\skelvars\cup\set\var}}) 
\\
&\tomarkblthree & \skelu{(\lat\vartwo\disc{\tm_b}{\skelvars\cup\set\var}\isub\vartwo{\skelu{\vartwo^\taken}})}
\end{array}\]
where $k_1 =_{\ih} \sizet{\disc{\tm_b}{\skelvars\cup\set\var}} - \sizet{\mtm_b} -1 = \sizet{\disc{\tm_b}{\skelvars\cup\set\var}}  -1$.
Now, we distinguish two cases.
\begin{itemize}
\item \case{$\vartwo \in \fv{{\tm_b}}$} By definition, $\sizet{\disc{\tm_b}{\skelvars \cup \{\var\}}}$ starts with $\taken$. Therefore, applying the \ih again:
\[\begin{array}{cllllll}
\mtm\isub\var{\skelu{\var^\taken}}
&\tomark^{k_1 + 1} &\skelu{(\lat\vartwo\disc{\tm_b}{\skelvars\cup\set\var}\isub\vartwo{\skelu{\vartwo^\taken}})}
\\
\mbox{(by \ih)} &\tomark^{k_2} &
\skelu{(\lat\vartwo\disc{\tm_b}{\skelvars \cup \{\var\} \cup \{\vartwo\})})} 
\\
&= & 
\skelu{\disc{\la\vartwo{\tm_b}}{\skelvars \cup \{\var\}}} 
&= & 
\skelu{\mtmtwo}
\end{array}\]
where $k_2 =_{\ih} \sizet{\disc{\tm_b}{\skelvars \cup \{\var\} \cup \{\vartwo\})}} - \sizet{\disc{\tm_b}{\skelvars \cup \{\var\}}}$.
Since also $\disc{\tm_b}{\skelvars\cup\{\var\}\cup\{\vartwo\}}$ starts with $\taken$, we have (remember that $\sizet{\mtm}=0$):
\[\begin{array}{lllllll}
k&=&k_1+1+k_2
\\
& =_{\ih} & \sizet{\disc{\tm_b}{\skelvars\cup\set\var}} -1 + 1 + \sizet{\disc{\tm_b}{\skelvars \cup \{\var\} \cup \{\vartwo\})}} - \sizet{\disc{\tm_b}{\skelvars \cup \{\var\}}}
\\
& = & 
\sizet{\disc{\tm_b}{\skelvars \cup \{\var\} \cup \{\vartwo\})}} 
\\
& = & 
\sizet{\mtmtwo} -1
\\
& = & 
\sizet{\mtmtwo} - \sizet{\mtm} -1.
\end{array}\]
Then the statement holds, since $\mtm$ starts unmarked.

\item \case{$\vartwo \notin \fv{{\tm_b}}$} Then, $\disc{\tm_b}{\skelvars \cup \{\var\}}\isub\vartwo{\skelu{\vartwo^\taken}} = \disc{\tm_b}{\skelvars \cup \{\var\}}$.
By \reflemma{disc-useless-var}, since $\vartwo \notin \fv{{\tm_b}}$ we also have that $\disc{\tm_b}{\skelvars \cup \{\var\}} = \disc{\tm_b}{\skelvars \cup \{\var\} \cup \{\vartwo\}}$. Therefore:
\[\begin{array}{lllllll}
\mtm\isub\var{\skelu{\var^\taken}}
&\tomark^{k_1 + 1} &
\skelu{(\lat\vartwo\disc{\tm_b}{\skelvars \cup \{\var\}}\isub\vartwo{\skelu{\vartwo^\taken}})} 
\\
&=&
\skelu{(\lat\vartwo\disc{\tm_b}{\skelvars \cup \{\var\}})} 
\\
&=_{\reflemmaeq{disc-useless-var}}&
\skelu{(\lat\vartwo\disc{\tm_b}{\skelvars \cup \{\var\} \cup \{\vartwo\}})} 
\\
&=&
\skelu{\disc{\la\vartwo{\tm_b}}{\skelvars \cup \{\var\}})} 
&= &
\skelu{\mtmtwo}
\end{array}\]
With:
\[\begin{array}{cllllll}
k & = & k_1+1
\\
& =_{\ih} & \sizet{\disc{\tm_b}{\skelvars\cup\set\var}} -1 + 1
\\
& = & \sizet{\disc{\tm_b}{\skelvars\cup\set\var}}
\\
& = & \sizet{\disc{\tm_b}{\skelvars\cup\set\var\cup\set\vartwo}}
\\
& = & 
\sizet{\lat\vartwo\disc{\tm_b}{\skelvars \cup \{\var\} \cup \{\vartwo\}}}  -1
\\
& = & 
\sizet{\mtmtwo} -1

& = & 
\sizet{\mtmtwo} - \sizet{\mtm} -1.
\end{array}\]
Then the statement holds, since $\mtm$ starts unmarked.
\end{itemize}

\item \case{$\tm = \la\vartwo{\tm_b}$ and $\fv{\tm} \cap \skelvars \neq \emptyset$} Without loss of generality we can assume $\vartwo \notin \skelvars \cup \{\var\}$. By definition, $\mtm = \lat\vartwo\disc{\tm_b}{\skelvars \cup \{\vartwo\}}$, that is, $\mtm$ starts with $\taken$. Let $\mtm_b = \disc{\tm_b}{\skelvars \cup \{\vartwo\}}$. Note that $\fv{\tm} \cap \skelvars \neq \emptyset$ implies $\fv{\tm_b} \cap \skelvars \neq \emptyset$, in turn implying that $\mtm_b$ starts with $\taken$. Then:
\[\begin{array}{cllllll}
\mtm\isub\var{\skelu{\var^\taken}} 
&= &
\lat\vartwo\mtm_b\isub\var{\skelu{\var^\taken}}
\\
\mbox{(by \ih)} &\tomark^{k} &
\lat\vartwo\disc{\tm_b}{\skelvars\cup\{\var\}\cup\{\vartwo\}} 
\\
&=&
\disc{\la\vartwo{\tm_b}}{\skelvars\cup\{\var\}} 
&= & \mtmtwo.
\end{array}\]
Where $k =_{\ih} \sizet{\disc{\tm_b}{\skelvars\cup\{\var\}\cup\{\vartwo\}}} - \sizet{\mtm_b}$.
Then, since both $\mtm$ and $\mtmtwo$ start with $\taken$:
\[\begin{array}{cllllll}
k
& =_{\ih} & \sizet{\disc{\tm_b}{\skelvars\cup\{\var\}\cup\{\vartwo\}}} - \sizet{\mtm_b}
\\
& = & 
\sizet{\disc{\tm_b}{\skelvars\cup\{\var\}\cup\{\vartwo\}}} +1 - \sizet{\mtm_b} -1
& = & 
\sizet{\mtmtwo} - \sizet{\mtm}\qedhere
\end{array}\]
\end{itemize}
\end{proof}

\section{Proofs omitted from \refsect{prel-machines} (Preliminaries about Abstract Machines)}
\label{sect:app-prel-machines}

\begin{lemma}[One-step simulation]
  \label{l:one-step-simulation}
  Let a machine $\mach$ and a structural strategy $(\tostrat,\streq)$ be a distillery.
  For any state $\state$ of $\mach$, if $\decode\state \tostrat \streq\tmtwo$ then there is a state $\statetwo$ of $\mach$ such that $\state \tomachsea^*\tomachpr \statetwo$ and $\decode{\statetwo} \streq \tmtwo$ and such that the label of the principal transition and of the step are the same.
\end{lemma}

\begin{proof}
  For any state $\state$ of $\mach$, let $\nfo{\state}$ be a normal form of $\state$ with respect to $\tomachsea$: such a state exists because search transitions terminate (\refpoint{def-overhead-terminate} of \refdef{distillery}).
  Since $\tomachsea$ is mapped on identities (by \emph{search transparency}, \ie \refpoint{def-overhead-transparency} of \refdef{distillery}), one has $\decode{\nfo{\state}} = \decode\state$.
  Since $\decode\state$ is not $\tostrat$-normal by hypothesis, the halt property (\refpoint{def-progress}) entails that $\nfo{\state}$ is not final, therefore $\state \tomachsea^* \nfo{\state} \tomachpr \statetwo$ for some state $\statetwo$, and thus $\decode\state = \decode{\nfo{\state}} \tostrat \streq\decode{\statetwo}$ by principal projection (\refpoint{def-beta-projection}).
  By determinism of $\tostrat$ (\refpoint{def-determinism}), one obtains $\decode{\statetwo} \streq \tmtwo$.
\end{proof}

\gettoappendix{thm:abs-impl}
\begin{proof}
\applabel{thm:abs-impl}
  According to \refdef{implem}, given a $\skterms$-term $\tm$ and an initial state $\compilrel\tm\state$, we have to show that:
  \begin{enumerate}
   \item \label{p:exec-to-deriv} \emph{Runs to evaluations}: for any $\mach$-run $\run: \compilrel\tm\state \tomachine^* \statetwo$ there exists a 
$\tostrat$-evaluation $\deriv: \tm \tostrat^* \streq\decode\statetwo$;

\item \label{p:deriv-to-exec} \emph{Evaluations to runs}: for every $\tostrat$-evaluation $\deriv: \tm \tostrat^* \tmtwo$ there exists a 
$\mach$-run $\run: \compilrel\tm\state \tomachine^* \statetwo$ such that $\decode\statetwo \streq \tmtwo$;
  \end{enumerate}
  Plus the principal matching constraint that shall be evident by the principal projection requirement and how the proof is built.

  \subparagraph{Proof of \refpoint{exec-to-deriv}.}  By induction on $\sizepr\run \in \nat$. Cases:
  \begin{itemize}
  \item $\sizepr\run = 0$. Then $\run \colon \compilrel\tm\state \tomachsea^* \statetwo$ and hence $\decode{\state} = \decode\statetwo$ by search transparency (\refpoint{def-overhead-transparency} of \refdef{distillery}).
  Moreover, $\tm = \decode{\state}$ since decoding is the inverse of initialization, therefore the statement holds with respect to the empty evaluation $\deriv$ with starting and end term $\tm$.
  
\item $\sizepr\run > 0$. Then, $\run \colon \compilrel\tm\state \tomachine^* \statetwo$ is the concatenation of a run $\runtwo \colon \compilrel\tm\state \tomachine^* \statethree$ followed by a run $\runthree \colon \statethree \tomachpr \statefour \tomachsea^* \statetwo$.
  By \ih applied to $\runtwo$, there exists an evaluation $\derivtwo \colon \tm \tostrat^* \streq\decode\statetwo$ satisfying the principal matching constraint.
  By principal projection (\refpoint{def-beta-projection} of \refdef{distillery}) and search transparency (\refpoint{def-overhead-transparency} of \refdef{distillery}) applied to $\runthree$, one obtains a one-step evaluation $\derivthree \colon \decode\statetwo \tostrat \streq\decode\statethree = \decode\state$ having the same principal label of the transition. 
Concatenating $\derivtwo$ and $\derivthree$, we obtain an evaluation:
\[\derivfour \colon \tm  \tostrat^* \streq\decode\statetwo\tostrat \streq\decode\state.\]
Postponing structural equivalence (\refprop{streq-is-a-strong-bisim}) we obtain the required evaluation $\deriv \colon \tm  \tostrat^* \streq\decode\state$, which satisfies the principal matching constraint.
 \end{itemize}
 
  \subparagraph{Proof of \refpoint{deriv-to-exec}.}  By induction on $\size\deriv \in \nat$. Cases:
  \begin{itemize}
  \item $\size\deriv = 0$. Then $\tm = \tmtwo$.
  Since decoding is the inverse of initialization, one has $\decode\state = \tm$.
  Then the statement holds with respect to the empty (\ie without transitions) run $\run$ with initial (and final) state $\state$.
  
 \item  $\size\deriv > 0$. Then, $\deriv\colon \tm \tostrat^* \tmtwo$ is the concatenation of an evaluation $\derivtwo \colon \tm \tostrat^* \tmtwop$ followed by the step $\tmtwop \tostrat \tmtwo$.
  By \ih, there exists a $\mach$-run $\runtwo\colon \compilrel\tm\state \tomachine^* \statethree$ such that $\decode\statethree \streq \tmtwop$ and verifying the principal matching constraint.
  Since $\decode\statethree \streq \tmtwop\tostrat \tmtwo$, by strong bisimulation of $\streq$ (\refprop{streq-is-a-strong-bisim}) we obtain $\decode\statethree \tostrat \streq\tmtwo$.
By one-step simulation (\reflemma{one-step-simulation}), there is a state $\statetwo$ of $\mach$ such that $\statethree \tomachsea^*\tomachpr \statetwo$ and $\decode\statetwo\streq \tmtwo$, preserving the label of the step/transition.
  Therefore, the run $\run \colon \compilrel\tm\state \tomachine^*\statethree \tomachsea^*\tomachpr \statetwo$ satisfies the statement.\qedhere
  \end{itemize}
\end{proof}
\section{Proofs omitted from \refsect{skeletal-MAD} (The MAD and the Skeletal MAD)}
\label{sect:app-skeletal-MAD}

\gettoappendix{l:skeletal-mam-qual-invariants}

\begin{proof}
\applabel{l:skeletal-mam-qual-invariants}
By induction on the length of the run. The base case trivially holds, the inductive case is by analysis of the last transition. The first two invariants are proved by straightforward inspection of the transitions. For the third one, the crucial point is that the well-bound invariant ensures that in a chain $\chain\cons(\var, \stack, \env\esub\var\cdot)$ there are no environment entries of variable $\var$ in $\env$ or $\chain$, so that in the definition of read-back $\decode{\chain\cons(\var, \stack, \env\esub\var\cdot)}	 \defeq  \decodep\env{\decodep\chain{\decodep\stack\var}}\esub\var\ctxhole$ the context $\decodep\env{\decodep\chain{\decode\stack}}$ does not bind $\var$ and it is indeed $\esub\var\ctxhole$ that binds it.
\end{proof}

\gettoappendix{thm:SkMAD-properties}

\begin{proof}
\applabel{thm:SkMAD-properties}\hfill
\begin{enumerate}
\item 
\begin{enumerate}
\item If $\state = \fourstate \chain {\la\var\tm} {\tmtwo \cons\stack}  \env \tomachb \fourstate \chain \tm \stack {\esub\var\tmtwo\cons\env} = \statetwo$ then:
\[\begin{array}{rllll}
 \decode{\fourstate \chain {\la\var\tm} {\tmtwo \cons\stack}  \env}
 & = &
 \decodep\env{\decodep\chain{\decodep{\tmtwo \cons\stack}{\la\var\tm}}}
\\ & = &
 \decodep\env{\decodep\chain{\decodep\stack{(\la\var\tm)\tmtwo}}}
\\ (\reflemmaeq{skeletal-mam-qual-invariants}.3)& \toskneeddb &
 \decodep\env{\decodep\chain{\decodep\stack{\tm\esub\var\tmtwo}}}
\\ & \streq_{\reflemmaeq{skeletal-mam-qual-invariants}.3} &
 \decodep{\env}{\decodep\chain{\decodep\stack{\tm}}\esub\var\tmtwo}
\\ & = &
 \decodep{\esub\var\tmtwo\cons\env}{\decodep\chain{\decodep\stack{\tm}}}
 & = &		
	\decode{\fourstate \chain \tm \stack {\esub\var\tmtwo\cons\env}}.
\end{array}\]
Note that the $\db$ steps applies because by the contextual chain read-back invariant (\reflemmaeq{skeletal-mam-qual-invariants}.3) $(\la\var\tm)\tmtwo$ occurs inside an evaluation context. Because of the same invariant, we can apply structural equivalence.

\item If $\state = \fourstate \chain \var \stack {\env\cons\esub\var\val\cons\envtwo}
	  	\tomachsk 
		\fourstate \chain \var \stack {\env\cons\skesub\var{\val'}\cons\env_\sctx\cons\envtwo}
 = \statetwo$ then:
\[\begin{array}{rllll}
 \decode{\fourstate \chain \var \stack {\env\cons\esub\var\val\cons\envtwo}}
 & = &
 \decodep{\env\cons\esub\var\val\cons\envtwo}{\decodep\chain{\decodep\stack{\var}}}
\\ & = &
 \decodep{\envtwo}{\decodep\env{\decodep\chain{\decodep\stack{\var}}}\esub\var\val}
\\ (\reflemmaeq{skeletal-mam-qual-invariants}.2,3)&\tosk &
 \decodep{\envtwo}{\sctxp{\decodep\env{\decodep\chain{\decodep\stack{\var}}}\skesub\var\valtwo}}
\\ & = &
 \decodep{\env\cons\env_\sctx\cons\skesub\var\valtwo\cons\envtwo}{\decodep\chain{\decodep\stack{\var}}}
 & = &		
	\decode{\fourstate \chain \var \stack {\env\cons\skesub\var{\val'}\cons\env_\sctx\cons\envtwo}}.
\end{array}\]
The $\sksym$ steps applies because:
\begin{enumerate}
\item As in the previous case, the contextual chain read-back invariant (\reflemmaeq{skeletal-mam-qual-invariants}.3) ensures that the step takes place in an evaluation context;
\item The well-bound invariant (\reflemma{skeletal-mam-qual-invariants}.2) ensures that there is only one environment entry on $\var$ in $\state$, so that there are no entries bounding $\var$ in $\env$ (if there were another entry $\esub\var\tm$ in $\env$ then the occurrence of $\var$ in $\esub\var\val$ would violate the invariant).
\end{enumerate}

\item The case $\tomachss$ is analogous to the $\sksym$ case.
\end{enumerate}

\item Cases of the search transition:
\begin{enumerate}
\item If $\state = \fourstate \chain {\tm\tmtwo} \stack \env
	  	\tomachseaone
	  	\fourstate \chain \tm {\tmtwo\cons\stack} \env = \statetwo$ then:
\[\begin{array}{lllllllll}
 \decode{\fourstate \chain {\tm\tmtwo} \stack \env}
 & = &
 \decodep{\env}{\decodep\chain{\decodep\stack{\tm\tmtwo}}}
& = &
 \decodep{\env}{\decodep\chain{\decodep{\tmtwo\cons\stack}\tm}}
 & = &		
	\decode{\fourstate \chain \tm {\tmtwo\cons\stack} \env}.
\end{array}\]

\item If $\state = \fourstate \chain \var \stack {\env\cons\esub\var\tm\cons\envtwo}
	  	\tomachseatwo 
	  	\fourstate {\chain\cons(\var, \stack, \env\esub\var\cdot)} \tm \stempty \envtwo  = \statetwo$ then:
\[\begin{array}{lllllllll}
 \decode{\fourstate \chain \var \stack {\env\cons\esub\var\tm\cons\envtwo}}
 & = &
 \decodep{\env\cons\esub\var\tm\cons\envtwo}{\decodep\chain{\decodep\stack{\var}}}
\\ & = &
 \decodep{\envtwo}{\decodep\env{\decodep\chain{\decodep\stack\var}}\esub\var\tm}
 & = &		
	\decode{\fourstate {\chain\cons(\var, \stack, \env\esub\var\cdot)} \tm \stempty \envtwo}.
\end{array}\]

\item If $\state = \fourstate {\chain\cons(\var, \stack, \env\esub\var\cdot)} \val \stempty \envtwo
		 \tomachseathree 
		\fourstate \chain \var \stack {\env\cons\esub\var{\val}\cons\envtwo} = \statetwo$ then:
\[\begin{array}{lllllllll}
 \decode{\fourstate {\chain\cons(\var, \stack, \env\esub\var\cdot)} \val \stempty \envtwo}
 & = &
 \decodep{\envtwo}{\decodep\env{\decodep\chain{\decodep\stack\var}}\esub\var\val}
\\ & = &
\decodep{\env\cons\esub\var\val\cons\envtwo}{\decodep\chain{\decodep\stack{\var}}}
 & = &		
	\decode{\fourstate \chain \var \stack {\env\cons\esub\var\val\cons\envtwo}}.
\end{array}\]
\end{enumerate}

\item Termination is implied by the complexity analysis of the next section.
\item Let us first show that final states $\state_f$ have shape $\fourstate\emptylist\val\emptylist\env$. We have the following points:
\begin{itemize}
\item \emph{The code of a final state is an abstraction}: it cannot be an application, because transition $\seasym_1$ would apply, and cannot be a variable $\var$, because by the closure invariant (\reflemma{skeletal-mam-qual-invariants}.1) there has to be an environment entry on $\var$, and so one among transitions $\seasym_2$, $\sksym$, and $\sssym$ would apply;
\item \emph{The stack is empty}: because, given that the code is an abstraction, is the stack is non-empty then a $\beta$ transition would apply;
\item \emph{The chain is empty}: because, given that the code is an abstraction and the stack is empty, then transition $\seasym_3$ would apply.
\end{itemize}
Now, final states decode to terms of shape $\sctxp\val$ which are $\toskneed$ normal.
\qedhere
\end{enumerate}
\end{proof}
\section{Proofs Omitted from \refsect{complexity} (Complexity Analysis)}
\label{sect:app-complexity}

\gettoappendix{l:skeleton-size}

\begin{proof}
\applabel{l:skeleton-size}
Point 1 is by induction on $\skeldec\tm\skelvars$. Point 2 follows from Point 1 and the definitions of skeleton and flesh.
\end{proof}

\gettoappendix{l:sub-term}

\begin{proof}
\applabel{l:sub-term}
By induction on the length of the run. The base case trivially holds, the inductive case is by analysis of the last transition. For all transitions but $\tomachsk$ it follows from the \ih and a straightforward inspection of the transitions. For $\tomachsk$, it is given by \reflemma{skeleton-size}.2.
\end{proof}

}{}
}

\end{document}